\documentclass[11pt]{article}
\usepackage{helvet}
\usepackage{authblk}

\usepackage{amssymb,amsbsy,amsfonts,amsmath,amsthm,xspace}
\usepackage{mathrsfs}
\usepackage{courier}
\usepackage[utf8]{inputenc}
\usepackage{tikz}
\usetikzlibrary{calc, matrix, arrows, shapes, positioning, fit, shapes.misc, shapes.geometric, calc, decorations.pathreplacing}
\usepackage{booktabs}
\usepackage{sectsty}
\sectionfont{\fontsize{11}{11}\selectfont}
\subsectionfont{\fontsize{11}{11}\selectfont}
\subsubsectionfont{\fontsize{11}{11}\selectfont}
\usepackage{graphicx}
\usepackage{caption}
\usepackage{setspace}
\usepackage{enumitem}
\usepackage{colortbl}

\theoremstyle{plain}
\newtheorem{theorem}{Theorem}

\newtheorem{proposition}{Proposition}
\newtheorem{corollary}[theorem]{Corollary}
\newtheorem{lemma}{Lemma}
\newtheorem{remark}{Remark}

    \newtheorem{assumption}{Assumption}

\theoremstyle{plain}
\newtheorem{definition}{Definition}[section]

\usepackage{subcaption}
\usepackage{multirow}
\usepackage{epstopdf}
\usepackage{epsfig}
\usepackage{hyperref}
\usepackage{url}
\usepackage[toc,page]{appendix}
\usepackage{float}
\usepackage{natbib}
\usepackage{color}
\usepackage{verbatim}
\usepackage{authblk}
\usepackage{wrapfig}
\usepackage[linesnumbered,ruled]{algorithm2e}

\usepackage[normalem]{ulem}

\usepackage[margin=1in]{geometry}
\newcommand{\tr}{\mathrm{tr}}

 \newcommand{\be}{\begin{eqnarray}}
\newcommand{\ee}{\end{eqnarray}}
\newcommand{\ba}{\begin{eqnarray*}}
	\newcommand{\ea}{\end{eqnarray*}}
\newcommand{\bei}{\begin{itemize}}
	\newcommand{\beiftnt}{\begin{itemize}\footnotesize}
		\newcommand{\eei}{\end{itemize}}

\usepackage{titlesec}

\titlespacing*{\section}
{0pt} 
{10pt} 
{5pt} 
\titlespacing*{\subsection}
{0pt} 
{6pt} 
{3pt} 

\def\text#1{\mbox{\rm #1}}

\usepackage{hyperref}

\hypersetup{
  colorlinks   = true, 
  urlcolor     = blue, 
  linkcolor    = blue, 
  citecolor   = blue 
}

\title{Perturbation-Robust Predictive Modeling of Social Effects by Network Subspace Generalized Linear Models}

\author[a]{Jianxiang Wang}
\author[b]{Can M. Le}
\author[c]{Tianxi Li}

\affil[a]{Rutgers University -- New Brunswick}
\affil[b]{University of California, Davis}
\affil[c]{University of Minnesota, Twin Cities}

\begin{document}

\maketitle

\begin{abstract}
Network-linked data, in which multivariate observations are interconnected by a network, are becoming increasingly prevalent in fields such as sociology and biology. These data often exhibit inherent noise and complex relational structures, complicating conventional modeling and statistical inference. Motivated by empirical challenges in analyzing such datasets, this paper introduces a family of network subspace generalized linear models designed for analyzing noisy, network-linked data. We propose a model inference method based on subspace-constrained maximum likelihood that emphasizes flexibility in capturing network effects and provides an inference framework that is robust under network perturbations. We establish the asymptotic distributions of the estimators under network perturbations, demonstrating the method’s accuracy through extensive simulations involving random network models and deep-learning-based embedding algorithms. The proposed methodology is applied to a comprehensive analysis of a large-scale study on school conflicts, where it identifies significant social effects, offering meaningful and interpretable insights into student behavior. 
\end{abstract}

\section{Introduction}

\noindent Network data analysis has become increasingly popular due to its wide-ranging applications in the social sciences \citep{holme2015modern,van2018social}, biological sciences \citep{ozgur2008identifying,zeng2018prediction}, and engineering \citep{le2014extending,cuadra2015critical}. A notable category of social network data concerns network-linked objects, in which the interactions or relationships among individuals are depicted through network structures, and each individual typically has associated response variables and covariates. Such structures are frequently encountered in studies examining social influences on human behavior \citep{michell1996peer,michell2000smoke,harris2009national,paluck2016changing}. In this paper, we focus on analyzing student behavior in the context of school conflicts, using data from \cite{paluck2016changing}. Despite the development of numerous statistical models to analyze network-linked data in recent years \citep{zhang2016community,li2019prediction,su2019testing,zhang2020logistic,sit2021event,mao2021consistent,mukherjee2021high,le2022linear,hayes2022estimating,fang2023group,he2023semiparametric,lunde2023conformal,zhu2017network,wu2023random,armillotta2023nonlinear,chang2024embedding}, the noisy nature of the network structures in this study necessitates non-trivial generalizations of the methods in the existing literature to effectively analyze the school conflict data. This challenge motivates the development of our new model. In the following sections, we introduce the school conflict study and review the current literature on predictive modeling for network-linked data.

\subsection{Social effect analysis in the school conflict study}\label{secsec:motivation}

\noindent A prospective study by \cite{paluck2016changing} investigated the effects of randomized anti-conflict interventions on social norms across 56 high schools in New Jersey. Data collection included official records from school administrations and student questionnaires, where students provided personal information, opinions on conflict-related events, and a list of their closest friends at both the beginning and end of the school year, allowing for the mapping of social networks within each school. In 25 of these schools, which were randomly selected from the 56, the experimenters introduced educational workshops aimed at a small group of students to reduce school conflicts. The field experiment sought to demonstrate that introducing educational interventions to students could help mitigate conflicts within schools. The anti-conflict impact was measured through the distribution of orange wristbands, which rewarded students for friendly or conflict-mitigating actions.

In this study, a key quantity of interest is the social influence, which could play a significant role in disseminating the effects of the intervention throughout the entire school. To facilitate the analysis of social influence, the experimenters recorded friendship relations in terms of ``how much time two students spent together." In addition to social influence, the study aims to understand the impact of various background covariates, such as gender, race, and family conditions. 
While the original study by \cite{paluck2016changing} utilized social relations to infer social effects, recent work by \cite{le2022linear} highlighted the importance of accounting for noisy observations in friendship relations to ensure valid inference. Specifically, two waves of surveys were administered within the same school year to capture social relations. However, the overlap between the two waves was limited. Figure~\ref{fig:overlappingedges} displays the edge overlap proportions across the 25 schools, showing that, in most schools, only about 50\% of the edges overlapped between the two periods. Such high levels of noise in the observations can jeopardize the validity of statistical inference if the network structure errors are not adequately addressed in the analysis.

\begin{figure}[ht]
    \centering
    \begin{subfigure}[b]{0.50\linewidth}
        \centering
        \vspace{-0.4cm}
        {\includegraphics[width=\linewidth]{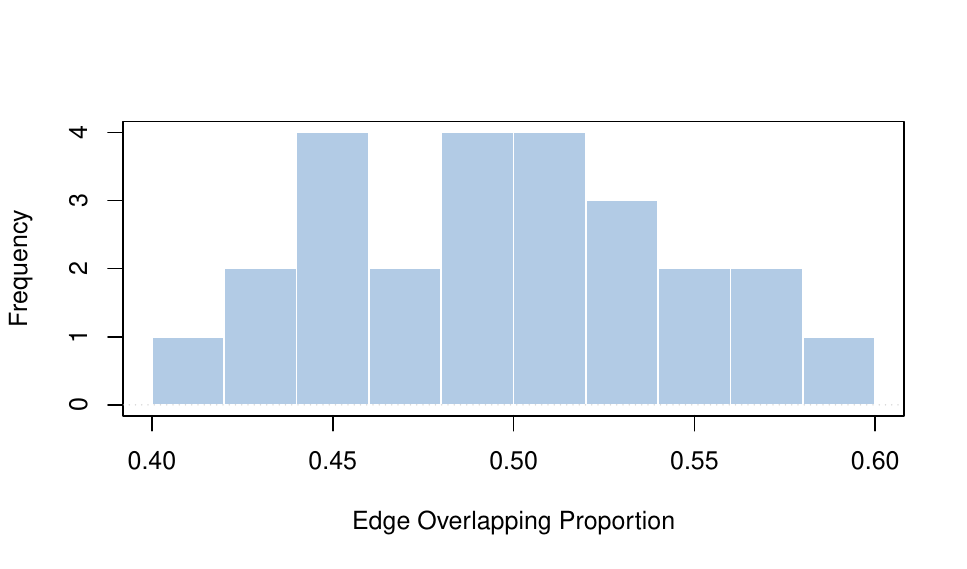}} 
        \subcaption{Overlapping Edges in Two Waves}
    \end{subfigure}
    \hfill
    \begin{subfigure}[b]{0.24\linewidth}
        \centering
        \vspace{-0.4cm}
        {\includegraphics[width=\linewidth]{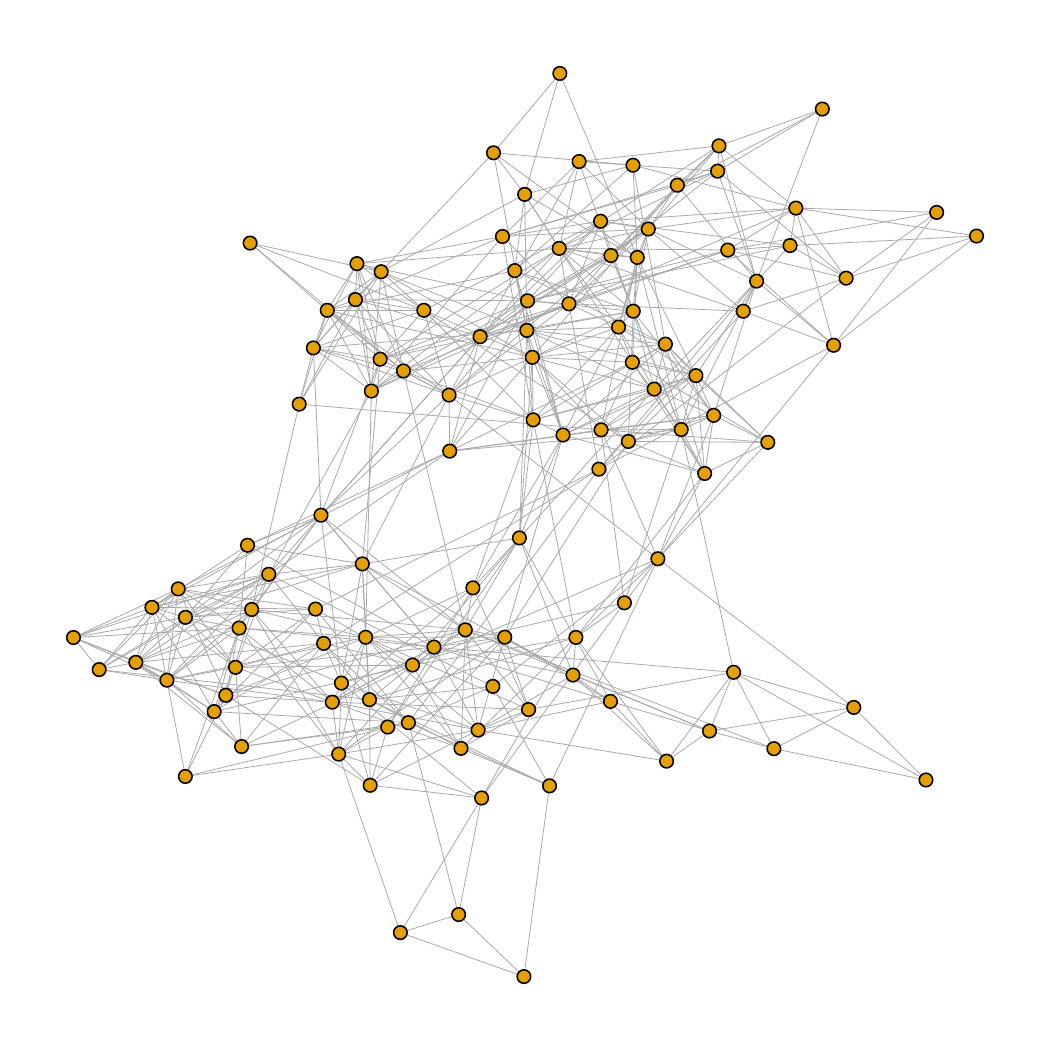}} 
        \label{Wave I}
        \subcaption{Wave I}
    \end{subfigure}
    \hfill
    \begin{subfigure}[b]{0.24\linewidth}
        \centering
        \vspace{-0.4cm}
        {\includegraphics[width=\linewidth]{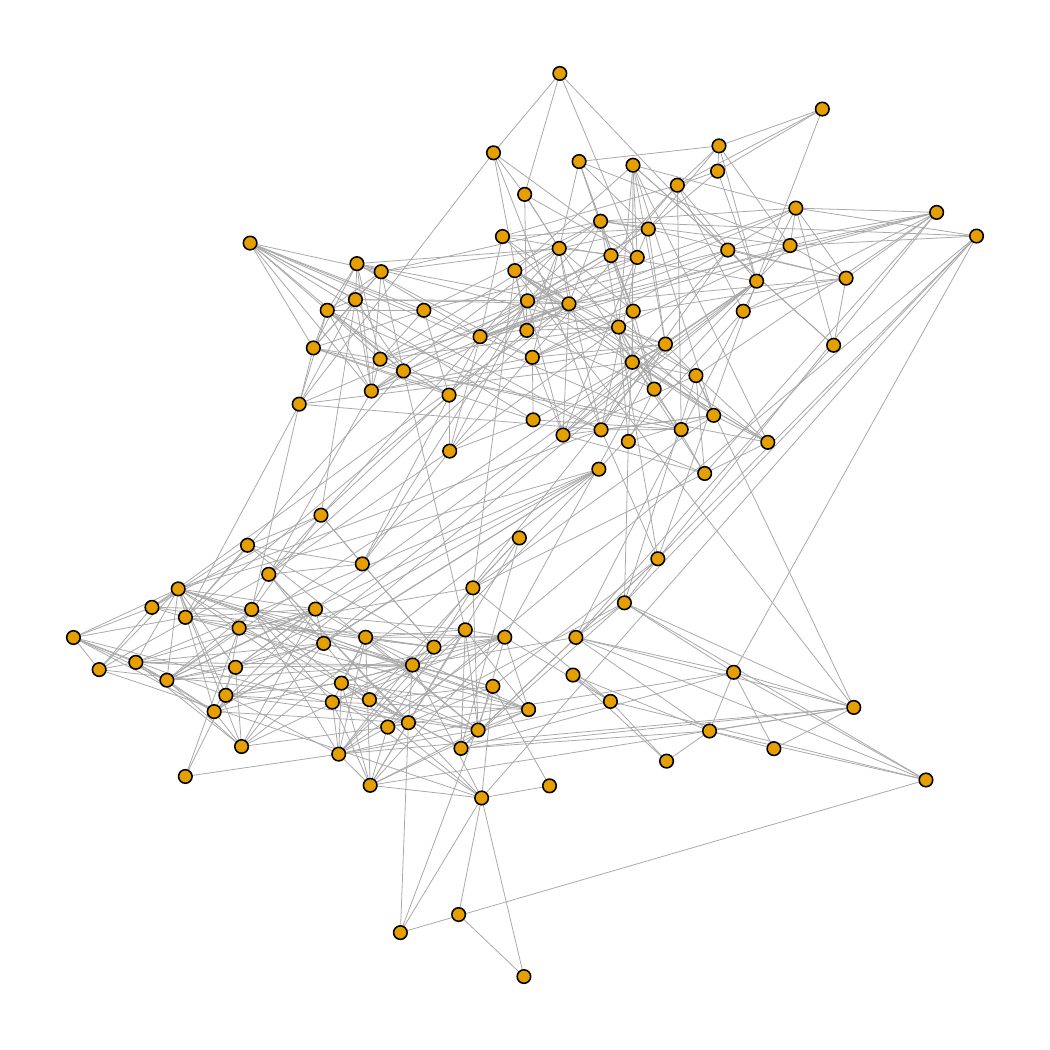}} 
        \label{Wave II}
        \subcaption{Wave II}
    \end{subfigure}
    \caption{Left panel: the proportion of overlapping edges in two waves in the study, across 25 schools. Right panel: the networks in two waves for one example school.}
    \label{fig:overlappingedges}
\end{figure}

\noindent More generally, noisy observations of network structures are frequently encountered in other empirical studies, particularly in the social sciences \citep{onnela2007structure,yu2008high,harris2009national}. Furthermore, popular graph-embedding methods in machine learning \citep{perozzi2014deepwalk,node2vec-kdd2016,rozemberczki2018fast}, which are commonly used to manage network data in modeling tasks, may introduce additional perturbations due to their inherent randomness. These challenges underscore the need for a general predictive modeling strategy that can account for network perturbations while ensuring valid inference. Addressing this issue is the central focus of our model development in this paper.

\subsection{Predictive models for network-linked data}\label{secsec:background}

\noindent We focus on predictive modeling for a node-level response variable. Among existing work, relevant predictive models can be categorized into three main classes based on their design. The first class treats the network as a generalized spatial structure and employs a graph-based autoregressive model to capture dependencies \citep{zhu2017network,armillotta2023nonlinear,wu2023random,chang2024embedding}. The second class of methods incorporates network information through a distance-based dependence structure, assuming that responses are independent if the distance between them exceeds a certain threshold \citep{su2019testing,sit2021event,mukherjee2021high}. Both of these classes rely on parametric forms of network effects and assume that the observed networks are accurate. While these methods provide informative model inference if the assumed network effect is appropriate, they may lead to misleading conclusions when the assumed parametric form is violated. Additionally, they tend to be vulnerable to errors in the observed network structure, which is a key issue in our motivating application.

The third class of methods employs nonparametric components to model network effects. For instance, \cite{li2019prediction} introduces the \emph{regression with network cohesion} (RNC) approach, which includes an individual node effects component along with a network smoothing penalty. This method has proven to be flexible for modeling network-linked responses and is applicable to various settings, such as generalized linear models. However, this approach lacks a formal statistical inference framework. A more recent model in this category is the \emph{subspace linear regression} proposed by \cite{le2022linear}. Instead of assuming a smooth network effect, this model posits that the effect lies within a latent subspace. It offers a valid inference framework and demonstrates robustness against network perturbations. However, the model fitting relies on a sequence of geometric projections, which are valid only for linear regression. This restriction can be limiting in practice, as categorical and discrete responses are common in social science applications, such as our motivating example of the school conflict study. In a separate line of work, \cite{hayes2022estimating} introduced a model in a similar vein to analyze network-mediated effects in causal problems, in which the network effect is parameterized by a linear combination of latent vectors. This method can handle other types of response variables, but the latent vectors are assumed to follow the random dot product graph model \citep{athreya2018statistical}.  

In this paper, we build upon the concept of subspace linear regression by introducing a new class of models called \emph{network-subspace generalized linear models} for network-linked data. Our model assumes that the predictive structure lies in the Minkowski sum of the column space of covariates and a latent subspace representing network relationships. We fit the model and conduct inference using subspace-constrained maximum likelihood, demonstrating that valid asymptotic statistical inference is guaranteed under essentially the same level of network perturbation as in the linear regression framework of \cite{le2022linear}. This advancement greatly expands the scope of robust predictive modeling and inference for network-linked data, accommodating both categorical and discrete response variables. Notably, the validity of our inference does not depend on a specific network perturbation model, allowing for application in a variety of settings with noisy network data.

We not only validate the inference of our model under traditional random network perturbations \citep{bickel2009nonparametric}, but also explore the integration of network effects through modern deep-learning-based embedding techniques commonly used in graph mining. Specifically, for the former, we show the effectiveness of our model for both low-rank and full-rank random network models. For the latter, we investigate three popular methods --- DeepWalk \citep{perozzi2014deepwalk}, Node2Vec \citep{node2vec-kdd2016}, and Diff2Vec \citep{rozemberczki2018fast}—demonstrating that the inherent noise and perturbations introduced by these algorithms are effectively managed by our model, ensuring accurate inference. Our work thus bridges the gap between rigorous statistical inference and general unsupervised strategies for incorporating network information.

\section{Methodology}
\label{sec:methodology}

\subsection{Notations}
Throughout the paper,  we use $c,C>0$  to denote absolute constants, the values of which may change from line to line. For two sequences of positive scalars $\{a_n\}_{n=1}^\infty$ and $\{b_n\}_{n=1}^\infty$, we write
$a_n=o(b_n)$ and $a_n=O(b_n)$ if $a_n/b_n$ converges to zero and $a_n/b_n$ is bounded, respectively. Similarly, for a sequence of random variables $\{X_n\}_{n=1}^\infty$,
we write $X_n=o_p(b_n)$ and $X_n=O_p(b_n)
$ if $X_n/b_n$ converges to zero and is bounded in probability, respectively. We use $I_n\in\mathbb{R}^{n\times n}$ to denote the identity matrix of size $n$.
For a matrix $A=(A_{ij})\in\mathbb{R}^{n\times n}$,
$\tr(A)=\sum_{i=1}^n A_{ii}$ is the trace, while $\lambda_{\min}(A)$ and 
$\lambda_{\max}(A)$ are the minimum and maximum eigenvalues of $A$, respectively, when $A$ is symmetric. For a vector $u$, $\|u\|$ is the Euclidean norm. 
For a matrix $W=(W_{ij})\in\mathbb{R}^{m\times n}$ with $1\le n\le m$ and the singular value decomposition $W=\sum_{i=1}^n \sigma_i u_i v_i^{\top}$, $\|W\|=\max_{1\le i \le n} \sigma_i$, $\|W\|_F=(\sum_{i=1}^n\sigma^2_i)^{1/2}$ and
$\|W\|_{\infty}=\max_{1\le i\le n} \sum_{j=1}^n|W_{ij}|$
represent the spectral norm, the Frobenius norm, and the infinity norm of $W$, respectively.
In addition, $W_{i}$ is the $i$-th column of $W$, and $W_{u:v}$ is the sub-matrix of $W$ with column vectors $W_i$ for $u\le i\le v$. We further use $W_{i,u:v}$ to denote the $i$-th row of $W_{u:v}$.
\subsection{Model}
We assume there exists a true unobserved relational matrix $P \in \mathbb{R}^{n \times n}$, where $P_{i j}$ describes the strength of the relationship between the nodes $i$ and $j$. Let $\hat{P}=(\hat{P}_{ij})\in \mathbb{R}^{n \times n}$ be an approximate relational matrix, which can be viewed as a noisy version of $P$ that is computable from observed relations between observations. An example from the random‑network modeling literature assumes that the entries of $\hat{P}$ are the observed adjacency connections between nodes, generated as independent Bernoulli random variables with $P = \mathbb{E}[\hat{P}]$, or some improved estimators based on certain statistical estimation methods \citep{li2023network}; another example discussed in detail in Section~\ref{secsec:embedding} involves stochastic embedding algorithms for which $\hat{P}$ is the similarity between the random embedding output. Intuitively, we expect that $\hat{P}$ does not significantly deviate from $P$.   

In addition to the relational matrix $\hat{P}$, for each node $i$, we observe $\left(x_i, y_i\right)$, where $x_i \in \mathbb{R}^p$ is a covariate vector and $y_i \in \mathbb{R}$ is a scalar response. Denote by $Y=\left(y_1, \ldots, y_n\right)^{\top} \in \mathbb{R}^n$ the response vector and by $X=\left(x_1, \ldots, x_n\right)^{\top} \in \mathbb{R}^{n \times p}$ the design matrix. 

Conditioning on $X$ and $P$, we assume that $y_1,...,y_n$ are independent random variables drawn from a \emph{generalized linear model} (GLM). 
Following \cite{mccullagh2019generalized}, the probability density or probability mass function of $y_i$ can be expressed in the following form:
\begin{equation}
\label{likelihood}
f(y ; \psi_i, \phi)=\exp\left(\frac{y \psi_i-b(\psi_i)}{a(\phi)}+d(y, \phi)\right), \quad i=1, \ldots, n.
\end{equation}
Here, $a$, $b$, and $d$ are specific functions depending on the distribution of $y_i$. For example, when $a(\phi)=1$, $b(\psi_i)=\log(1+e^{\psi_i})$, and $d(y, \phi)=0$, \eqref{likelihood} leads to a logistic regression; when $a(\phi)=1$, $b(\psi_i)=e^{\psi_i}$, and $d(y, \phi)=-\log(y!)$, a Poisson regression is obtained.
In addition, $\phi$ is a known dispersion parameter, and $\psi_i$ is the natural parameter. We write
$$ \psi_i=\psi(\mu_i), \quad \mu_i= \mathbb{E}[y_i| X, P].$$ 

We assume that the expected network-linked response vector
$\mu = \mathbb{E}[Y|X,P]$ depends on the column space spanned by $X$, denoted by $\operatorname{col}(X)$, and a network individual effect vector 
$$\omega \in S_K(P) \subset \mathbb{R}^{n}$$
through a link function, where
$S_K(P)$ is the linear subspace spanned by the $K$ leading eigenvectors of $P$. The assumption that $\omega$ belongs to $S_K(P)$ is natural and supported by existing evidence that leading eigenvectors of the relational matrix typically capture crucial network information \citep{ozgur2008identifying,zeng2018prediction,van2018social,lee2019document}. In particular, building on the modeling approach outlined in \cite{le2022linear}, we assume that $\mu$ is contained in the Minkowski sum of $\operatorname{col}(X)$ and $S_K(P)$ through the link function $h^{-1}$, which is assumed to be \emph{smooth} and \emph{increasing}:
\begin{align}
\label{tmp}
    h^{-1}\circ \mu  = X\upsilon + \omega
    \in \operatorname{col}(X)+ S_K(P):= \{u+v\mid u\in \operatorname{col}(X),v\in S_K(P)\}.
\end{align}
Here, by slight abuse of notation, we use $h^{-1}\circ \mu$ to denote the vector of values of $h^{-1}$ evaluated at entries of $\mu$.

A special and important case for $h^{-1}$ is the {\em natural link function} where $$h^{-1}=\psi, \quad \text{which implies} \quad \psi_i=x^{\top}_i\upsilon+\omega_i, \quad 1\le i\le n.$$ 
For example, the logistic regression assumes the natural link function $h^{-1}(\mu)=\log(\frac{\mu}{1-\mu})$ and Poisson regression assumes the natural link function $h^{-1}(\mu)=\log(\mu)$.

Note that $\operatorname{col}(X)$ and $S_K(P)$ may share a non-trivial subspace intersection, which occurs when both $X$ and $P$ depend on certain latent variables such as node cluster information. To ensure identifiability, we decompose $\operatorname{col}(X)+ S_K(P)$ based on the subspace intersection 
$$\mathcal{R}=\operatorname{col}(X) \cap S_K(P),$$
and parameterize the model as follows.

\begin{definition}
[Network subspace generalized linear model]
\label{def:GLM}
Consider a reparametrization of model \eqref{likelihood} as
\begin{align}
\label{model}
    h^{-1} \circ \mu =X\beta^*+\xi^*+\alpha^*,
\end{align}
where $\beta^* \in \mathbb{R}^p$ and $\xi^*, \alpha^* \in \mathbb{R}^n$ satisfy
\begin{equation}
\label{para}
    \xi^*=X\theta^* \in \mathcal{R}, \quad X \beta^* \perp \mathcal{R}, \quad \alpha^* \in S_K(P), \quad \alpha^* \perp \mathcal{R}.
\end{equation}
\end{definition}
It is straightforward to show that the parameterization in Definition~\ref{def:GLM} is identifiable. That is, if there exist $(\beta, \alpha, \theta)$ and $\left(\beta^{\prime}, \alpha^{\prime}, \theta^{\prime}\right)$ satisfying \eqref{model} and \eqref{para} simultaneously, then $\beta=\beta^{\prime}, \alpha=\alpha^{\prime}$, and $\theta=\theta^{\prime}$.

\subsection{Model fitting by the subspace-constrained maximum likelihood}
\label{sec:algorithm}

We now describe the model fitting procedure for the network subspace generalized linear model.
For ease of presentation, let us first outline this procedure, assuming we observe $S_K(P)$. At a high level, we want to use the restricted maximum likelihood estimator (MLE) under the subspace constraint under Definition~\ref{def:GLM}. Therefore, the estimation is done by solving the following optimization problem:
\begin{align}
\text{maximize}_{\beta, \xi, \alpha}~~&~~ \mathcal{L}(\beta, \xi,\alpha; Y, X)\\
\text{subject to}~~&~~ \alpha, \beta, \xi \text{~satisfy \eqref{para}}\notag
\end{align}
where $\mathcal{L}(\beta, \xi,\alpha; Y, X)$ is the log-likelihood of the data. To handle the subspace constraint in the optimization, we will introduce a reparameterization of our model for the estimation.


\medskip

{\em Reparameterization.} Using \eqref{para}, we first rewrite \eqref{model} in a more convenient form for estimation purposes. 
Denote by $\bar{Z} \in \mathbb{R}^{n \times p}$ a matrix whose columns form an orthonormal basis of the covariate subspace $\operatorname{col}(X)$. 
Similarly, let $\bar{W} \in \mathbb{R}^{n \times K}$ be the matrix whose columns are eigenvectors of $P$ that span the subspace $S_K(P)$. 
The singular value decomposition of matrix $\bar{Z}^{\top} \bar{W}$ takes the form 
$$\bar{Z}^{\top}\bar{W} =U \Sigma V^{\top}.$$ 
Here, $U \in \mathbb{R}^{p \times p}$ and $V \in \mathbb{R}^{K \times K}$ are orthonormal matrices of singular vectors, while $\Sigma \in \mathbb{R}^{p \times K}$ is the matrix with the following singular values on the main diagonal:
\begin{equation}
    \label{defin:model}    \sigma_1=\sigma_2=\cdots=\sigma_r=1>\sigma_{r+1} \geq \cdots \geq \sigma_{r+s}>0=\sigma_{r+s+1}=\cdots=0,
\end{equation}
where $r$ and $s$ denote the number of singular values equal to $1$ and those taking values strictly between 0 and 1, respectively. Note that it is possible for $r$ and $s$ to be zero.
To calculate a basis for the intersection subspace  $\mathcal{R}$, let us denote
\begin{equation}
\label{eq:Z matrix}
   Z=\sqrt{n}\bar{Z} U, \quad  {W}=\sqrt{n}\bar{W} V.
\end{equation}
It follows that
$$\mathcal{R}=\operatorname{col}(Z_{1:r})=\operatorname{col}(W_{1:r}),$$ where $Z_{1:r}\in\mathbb{R}^{n\times r}$ is the submatrix of the first $r$ columns of $Z$ and $W_{1:r}$ is similarly defined. Note that the factor $\sqrt{n}$ ensures that entries of $Z$, $W$, and $X$ are generally of comparable magnitudes.  
We use  
$$\mathcal{C}=\operatorname{col}(Z_{(r+1):p}), \quad \mathcal{N}= \operatorname{col}(W_{(r+1):K})$$
to denote the complement subspaces of $\mathcal{R}$ within $\operatorname{col}(X)$ and $S_K(P)$, respectively. With these notations, 
$$\operatorname{col}(X)+ S_K(P) = \mathcal{R}+\mathcal{C}+\mathcal{N}.$$
Therefore, there exists a vector $\gamma^*\in\mathbb{R}^{p+K-r}$ such that equation \eqref{model} is equivalent to  
\begin{equation}
\label{eq:mean equation 2}
    h^{-1} \circ \mu={Z}_{1:r}\gamma^*_{1:r}+{Z}_{(r+1):p} \gamma^*_{(r+1):p}+{W}_{(r+1):K}\gamma^*_{(p+1):(p+K-r)} =  
    \begin{bmatrix}
        Z ~~&~~{W}_{(r+1):K}
    \end{bmatrix}
    \gamma^*,
\end{equation}
where, for any positive integers $s\le t$, we use $\gamma^*_{s:t}\in\mathbb{R}^{t-s+1}$ to denote the sub-vector of $\gamma^*$ with entries indexed by integers between $s$ and $t$. Note that our parameters of interest are ultimately $\theta^*,\beta^*$, and $\alpha^*$, which can be calculated from $\gamma^*$ as follows:
\begin{eqnarray}
    \label{eq:par trans theta} 
    \theta^* &=&(X^\top X)^{-1} X^\top {Z}_{1:r}{\gamma}^*_{1:r},\\
    \label{eq:par trans beta}
    \quad \beta^* &=& (X^\top X)^{-1} X^\top {Z}_{(r+1):p}{\gamma}^*_{(r+1):p},\\
    \label{eq:par trans alpha}
    \alpha^* &=&
    {W}_{(r+1):K}{\gamma}^*_{(p+1):(p+K-r)}.
\end{eqnarray}
Although $Z$, $W$, and $\gamma^*$ depend on the choice of bases for $\operatorname{col}(X)$ and $S_K(P)$, parameters $\theta^*,\beta^*$, and $\alpha^*$ are invariant with respect to such choice. With these formulas, the problem of estimating parameters in \eqref{para} is equivalent to estimating $\gamma^*$, based on an arbitrary basis $Z$ and $W$ corresponding to the true $P$.

\medskip

{\em Estimating equation -- the ideal case.} We now proceed to estimate $\gamma^*$. 
In light of equation~\eqref{eq:mean equation 2},
let us first denote the $i$-th row of matrix $(Z \ {W}_{(r+1):K})$ 
by $g_i^\top$, or equivalently, 
$$
g_i = (Z_{i,1:p} ~~~~ W_{i,(r+1):K})^\top\in\mathbb{R}^{p+K-r}.
$$
Viewing $g_i$ as a new covariate vector for the $i$-th observation turns the model of Definition~\ref{def:GLM} into a typical generalized linear model with parameter $\gamma^*$ (if we do know $g_i$'s). Using the first-order stationary condition and setting the gradient of the likelihood function to zero leads to the following estimating equation:
\begin{equation}
\label{estimating equation real}S\left(\gamma\right)= \frac{1}{n}\sum_{i=1}^n {g}_i\frac{h^{\prime}\left({g}^{\top}_i\gamma\right)}{v\left({g}^{\top}_i\gamma\right)}  \left(y_i-h\left({g}^{\top}_i\gamma\right)\right)=0,
\end{equation}
where $h^{\prime}(\cdot)$ is the derivative of the inverse link function and  $v(g_i^{\top}\gamma) $ is the variance of $y_i$.
Taking the partial derivative of $-S(\gamma)$, we obtain the oracle information matrix
\begin{equation}
\label{F}    F\left(\gamma\right)= \frac{1}{n} \sum_{i=1}^n  \frac{\left(h^{\prime}\left({g}^{\top}_i\gamma\right)\right)^{2}}{v\left({g}^{\top}_i\gamma\right)} {g}_i {g}^{\top}_i.
\end{equation}
Later on, this matrix will be used to approximate the asymptotic variance in  \eqref{estimating equation}. It is unique up to a rotation due to the choice of basis for $\operatorname{col}(X)$ and $S_K(P)$.

\medskip

{\em Sample version estimators.}
In practice, instead of observing the relational matrix $P$ directly, we only have access to a noisy version $\hat{P}$ of $P$.
We replace $P$ with $\hat{P}$ everywhere in the above procedure. 
In particular,
let $\breve{W} \in \mathbb{R}^{n \times K}$ be the matrix whose columns are eigenvectors of $\hat{P}$ that span the subspace $S_K(\hat{P})$. 
The singular value decomposition of $\bar{Z}^{\top} \breve{W}$ takes the form 
$$\bar{Z}^{\top}\breve{W} =\Tilde{U}\Tilde{\Sigma} \Tilde{V}^{\top}.$$ 
%
Similarly, denote
\begin{equation}
\label{eq:tildeZ matrix}
   \Tilde{Z}=\sqrt{n}\bar{Z} \Tilde{U}, \quad  \Tilde{W}=\sqrt{n}\breve{W} \tilde{V}.
\end{equation}
%
We always assume that $r$, the dimension of $\mathcal{R}$, is known. If it is unknown, \cite{le2022linear} proposed a criterion to select $r$ and we can use it here.
Specifically, let $\hat{d}=\frac{1}{n} \sum_{i, j=1}^n \hat{P}_{i j}$. The following rule can be used to select $r$:
$$
\hat{r}=\max \left\{i: \hat{\sigma}_i \geq 1-\frac{4 \sqrt{p K \log n}}{\hat{d}}\right\}
$$
in which $\hat{\sigma}_i$'s are the singular values of $\bar{Z}^\top\breve{W}$. Under additional assumptions, \cite{le2022linear} showed that $\hat{r}$ can recover $r$ with high probability.

With the known $r$, we estimate $\mathcal{R}$, $\mathcal{C}$, and $\mathcal{N}$ by
\begin{equation}
    \label{subspaceEstimator}
\hat{\mathcal{R}}= \operatorname{col}(\tilde{Z}_{1:r}),\quad
\hat{\mathcal{C}}= \operatorname{col}(\tilde{Z}_{(r+1):p}),\quad \hat{\mathcal{N}}= \operatorname{col}(\Tilde{W}_{(r+1):K}).
\end{equation}
The sample version of the estimating equation takes the form 
\begin{equation}
\label{estimating equation}
\Tilde{S}\left(\gamma\right)= \frac{1}{n}\sum_{i=1}^n \Tilde{g}_i\frac{h^{\prime}\left(\tilde{g}^{\top}_i\gamma\right)}{v\left(\tilde{g}^{\top}_i\gamma\right)}  \left(y_i-h\left(\tilde{g}^{\top}_i\gamma\right)\right)=0,
\end{equation}
where $\Tilde{g}_i$ denote the $i$-th row vector of matrix $\begin{bmatrix}\Tilde{Z} ~~~~ \Tilde{W}_{(r+1):K}\end{bmatrix}$. We solve this equation using the iteratively reweighted least squares method \citep{green1984iteratively}.   Finally, the sample information matrix is given by 
\begin{equation}
\label{tildeF}
\tilde{F}\left(\gamma\right)= \frac{1}{n} \sum_{i=1}^n \frac{\left(h^{\prime}\left(\tilde{g}^{\top}_i\gamma\right) \right)^2}{v\left(\tilde{g}^{\top}_i\gamma\right)}  \Tilde{g}_i \tilde{g}^{\top}_i.
\end{equation}

A summary of this procedure is given in Algorithm \ref{algorithm1}. It is worth mentioning that our algorithm requires access to $K$.
Since the problem of estimating $K$ has been extensively studied  \citep{li2020network,le2022estimating,han2023universal},
we will assume throughout this paper that $K$ is known.

\begin{algorithm}
\caption{\label{algorithm1}Subspace-Constrained Maximum Likelihood Estimation Algorithm}
\KwIn{Design matrix $X\in \mathbb{R}^{n\times p}$, response vector $Y \in \mathbb{R}^{n}$, estimated relational matrix $\hat{P}\in \mathbb{R}^{n\times n}$ and dimension of the intersection subspace $r$.}
\KwOut{Estimators $\hat{\theta}$, $\hat{\beta}$, and $\hat{\alpha}$.}
Calculate the orthonormal basis of $\operatorname{col}(X)$ and form matrix $\bar{Z}\in \mathbb{R}^{n \times p}$ in \eqref{eq:Z matrix};  
calculate $K$ eigenvectors of $\hat{P}$ and form $\breve{W} \in \mathbb{R}^{n \times K}$.

Calculate the singular value decomposition $\bar{Z}^{\top} \breve{W}=\tilde{U} \tilde{\Sigma} \tilde{V}^{\top}$,  and form $\Tilde{Z}= \sqrt{n}\bar{Z}\tilde{U}, \Tilde{W}=\sqrt{n} \breve{W}\tilde{V}.$  

Find the root $\hat{\gamma}$ of the generalized estimating equation $\Tilde{S}(\gamma)=0$ using the iteratively reweighted least squares method, and obtain $\hat{\theta},\hat{\beta},\hat{\alpha}$ by replacing $\gamma$ with $\hat{\gamma}$ in \eqref{eq:par trans theta}, \eqref{eq:par trans beta}, and \eqref{eq:par trans alpha}, respectively.
\end{algorithm}
%
%
%
%

\section{Statistical Inference Properties}\label{sec:theoretical results}
This section provides theoretical results for estimation consistency and statistical inference of the proposed method. 
To this end, we need the following regularity conditions. 

\begin{assumption}[Scaling]
\label{cond:A1}
    $\|X_j\|=\sqrt{n}$ for all columns of $X$. In addition, there exists a constant $C$ such that $\|X\beta^*\|, \|X\theta^*\|$ and $\|\alpha^*\|$ are bounded by $C\sqrt{n}$.
\end{assumption}

\begin{assumption}[Well-conditioned covariates]
    \label{cond:A5}
    There exists a constant $C>0$ such that 
$G=(X^{\top}X/n)^{-1}$ satisfies
\begin{equation*}
    1/C \leq \lambda_{\min }(G) \leq \lambda_{\max }(G) \leq C.
\end{equation*}
\end{assumption}

\begin{assumption}[Boundedness of design vectors]
    \label{cond:A2} 
    There exists a constant $C>0$ such that $\left\|g_i\right\|\leq C$ for all $1\le i \le n$.
\end{assumption}
\begin{assumption}[Well-conditioned information matrix]
    \label{cond:A3} 
There exist constants $\delta,C>0$ such that when  $\|\gamma-\gamma^*\|<\delta$ the oracle information matrix defined in \eqref{F} satisfies that 
\begin{equation*}
    1/C \leq \lambda_{\min }(F(\gamma)) \leq \lambda_{\max }(F(\gamma)) \leq C. 
\end{equation*}
\end{assumption}
\begin{assumption}[Small Projection Perturbation]
    \label{cond:A4} 
 The approximate relational matrix $\hat{P}$ satisfies
$$
\tau_n := n^{-3/2}\frac{\|(\tilde{W} \tilde{W}^{\top}-W W^{\top}) Z\|}{\min \left\{\left(1-\sigma_{r+1}\right)^3, \sigma_{r+s}^3\right\}} ,
$$
for any $n$, where $\sigma_{r+1}$ and $\sigma_{r+s}$ are the singular values in \eqref{defin:model},
and 
$$
\tau_n= o(n^{-1/2}).
$$
\end{assumption}

Assumption \ref{cond:A4}, which is also used in \cite{le2022linear}, is our essential requirement for the level of tolerable network perturbation. This assumption is not directly verifiable unless one specifies both the network's perturbation mechanism (typically unknown in practice) and the choice of $\hat{P}$. For example, under the ``inhomogeneous Erd\H{o}s-R\'enyi'' model, choosing the adjacency matrix $A$ as $\hat{P}$ may require a dense network (average degree above $\sqrt{n}$) for Assumption \ref{cond:A4} to hold, as suggested by \cite{le2022linear}. However, \cite{le2022linear} also shows that using parametric estimation to denoise $A$ can yield a $\hat{P}$ that requires a much weaker assumption under specific models. Generally speaking, one should leverage more efficient estimators of the probability matrix $P$ where appropriate to make Assumption \ref{cond:A4} easier to hold. Notable examples include the nonparametric estimators proposed by \citet{zhang2017estimating} and \citet{li2023network} for general network models, as well as model-specific estimators developed in \citet{ma2020universal} and \citet{rubin2022statistical}. A rigorous theoretical analysis of these estimators falls outside the scope of the present work. However, readers should note that we do not restrict ourselves to the inhomogeneous Erd\H{o}s-R\'enyi framework. Assumption \ref{cond:A4} should be interpreted more broadly—as a robustness criterion applicable beyond a specific network generative model. In our simulation study (Section~\ref{sec:Simulation Studies}), for instance, we consider a scenario in which the perturbation is from a deep-learning-based embedding (which is clearly not an inhomogeneous Erd\H{o}s-R\'enyi model) and the corresponding $\hat{P}$ is the similarity matrix of the embeddings. Empirically, we show that our method yields valid inference in this setting as well.

\begin{assumption}[Moment constraints for responses]
    \label{cond:A6} 
    There exist constants $c>0$, $M_0>0$ and $\xi>2$ such that 
$$\min_{1\le i\le n}\operatorname{Var}\left(y_i\right)>c,
\qquad 
 \mathbb{E}\big|y_i-\mathbb{E}[y_i]\big|^\xi<M_0.
$$
\end{assumption}

Assumption \ref{cond:A6} provides a sufficient condition for the Lindeberg-Feller Central Limit Theorem to hold. A similar constraint has been adopted in \cite{yin2006asymptotic,gao2012asymptotic}. 


\begin{theorem}[Existence and Consistency] 
\label{existence+consistency}
Consider the estimating equation \eqref{estimating equation} and assume that 
   Assumptions \ref{cond:A1}--\ref{cond:A6} hold. There exists $\hat{\gamma}$ such that as $n\to \infty$, 
\begin{equation}
\label{existence_core}
\mathbb{P}\left(\Tilde{S}\left(\hat{\gamma}\right)=0 \right) \rightarrow 1.
\end{equation}
Moreover, the corresponding estimates $\hat{\theta},\hat{\beta}$, and $\hat{\alpha}$, obtained by replacing $\gamma^*$ with $\hat{\gamma}$ in \eqref{eq:par trans theta}, \eqref{eq:par trans beta}, and \eqref{eq:par trans alpha}, respectively, satisfy
\begin{equation}
\label{Consistencyforpara}
    \big\|\hat{\theta}-\theta^*\big\|=o_p(1), \quad
    \big\|\hat{\beta}-\beta^*\big\|=o_p(1), \quad 
    \|\hat{\alpha}-\alpha^*\|=o_p(n^{1/2}).
\end{equation}
\end{theorem}

Theorem~\ref{existence+consistency} shows that for each $n$, there exists a solution to the estimating equation \eqref{estimating equation} with high probability. In addition, the sequences of corresponding estimates for the true parameters in \eqref{model} are consistent. 
It is worth noting that similar to \cite{yin2006asymptotic}, Theorem~\ref{existence+consistency} itself does not guarantee the uniqueness of the solution $\hat{\gamma}$. This is because the log-likelihood function is generally not concave for certain link functions. However, Corollary~\ref{cor:uniqueness} below shows that restricting the model space to the class with natural link functions, or more generally, link functions ensuring concavity, leads to the uniqueness.

\begin{corollary}[Uniqueness]
\label{cor:uniqueness}
Suppose Assumptions \ref{cond:A1} to \ref{cond:A6} hold and the link function is natural. That is, $h^{-1}=\psi$. Then the estimates in Theorem~\ref{existence+consistency} are unique for sufficiently large $n$. 
\end{corollary}



Our next result concerns the 
asymptotic distributions of the proposed estimates for $\theta^*$, $\beta^*$, and $\alpha^*$. Since these parameters depend on $\gamma^*$ through equations \eqref{eq:par trans theta}, \eqref{eq:par trans beta}, and \eqref{eq:par trans alpha}, we need the covariance matrices of $\hat{\gamma}_{1:r}$, $\hat{\gamma}_{(r+1):p}$, and $\hat{\gamma}_{(p+1):(p+K-r)}$. These matrices can be estimated by the diagonal blocks of the inverse of the sample information matrix in \eqref{tildeF}, which we denote by $\Tilde{F}^{-1}_1(\hat{\gamma})$, $\Tilde{F}^{-1}_2(\hat{\gamma})$, and $\Tilde{F}^{-1}_3(\hat{\gamma})$, respectively. Thus,    
\begin{align*}
\tilde{F}^{-1}(\hat{\gamma})&=
\left(\begin{array}{ccc}
\tilde{F}^{-1}_1(\hat{\gamma}) & *&* \\
*&\tilde{F}^{-1}_2(\hat{\gamma}) & *\\
* &*& \tilde{F}^{-1}_3(\hat{\gamma})
\end{array}\right),
\end{align*}
where 
$\Tilde{F}^{-1}_1(\hat{\gamma})\in\mathbb{R}^{r\times r}$, $\Tilde{F}^{-1}_2(\hat{\gamma})\in\mathbb{R}^{(p-r)\times(p-r)}$, and $\Tilde{F}^{-1}_3(\hat{\gamma})\in\mathbb{R}^{(K-r)\times(K-r)}$.  In addition, we use $\kappa(\hat{\gamma})\in\mathbb{R}^{n\times n}$ to denote the diagonal matrix with entries $(h^{\prime}(\tilde{g}^{\top}_i\hat{\gamma}))^2/v(\tilde{g}^{\top}_i\hat{\gamma})$, $1\le i\le n$, on the diagonal:
$$\kappa(\hat{\gamma})=\operatorname{diag}\left(\frac{(h^{\prime}(\tilde{g}^{\top}_i\hat{\gamma}))^2}{v(\tilde{g}^{\top}_i\hat{\gamma})}\right).
$$
We are now ready to state the asymptotic distributions of the proposed estimates.

\begin{theorem}[Asymptotic Distributions]
\label{theorem 2}
Assume that Assumptions \ref{cond:A1} to \ref{cond:A6} hold. For each $n$, let $\hat{\theta}$, $\hat{\beta}$, and $\hat{\alpha}$, 
be the estimates based on  $\hat{\gamma}$ satisfying Theorem \ref{existence+consistency}. We have the following results. 
\begin{enumerate}[label=(\alph*)]
\item As $n$ tends to infinity,
\begin{equation}
\label{gamma}
    n\left(\hat{\alpha}-\frac{1}{n}\Tilde{W}_{(r+1):K}\Tilde{W}^{\top}_{(r+1):K} \alpha^*\right)^{\top} \tilde{O}\left(\hat{\alpha}-\frac{1}{n}\Tilde{W}_{(r+1):K}\Tilde{W}^{\top}_{(r+1):K} \alpha^*\right) \ \to \  \chi_{K-r}^2,
\end{equation}
in distribution, where $\chi_{K-r}^2$ denotes the $\chi^2$ distribution with $K-r$ degrees of freedom, and 
$ \tilde{O}=n^{-1}\left(\kappa(\hat{\gamma})-\kappa(\hat{\gamma}) \tilde{Z}\left(\tilde{Z}^{\top} \kappa(\hat{\gamma}) \tilde{Z}\right)^{-1} \tilde{Z}^{\top} \kappa(\hat{\gamma})\right)$.
\item For any fixed unit vector $u \in \mathbb{R}^{p}$, 
assume that 
\begin{equation}
\label{beta_asymptotic_condition}
    n^{-1}\big\|{Z}_{(r+1):p}^{\top} X G u\big\| \geq c
\end{equation}
for some constant $c>0$ and  sufficiently large $n$. Then, 
\begin{equation}
\label{beta}
\frac{\sqrt{n}\big(u^{\top}\hat{{\beta}}-u^{\top}\beta^*\big)}{ n^{-1}\left( u^{\top}  G X^{\top} \Tilde{Z}_{(r+1):p} \Tilde{F}^{-1}_2(\hat{\gamma})\Tilde{Z}_{(r+1):p}^{\top} X G u \right)^{1/2}} \ \to 
 \ \mathcal{N}(0,1), 
\end{equation}
where $\mathcal{N}(0,1)$ denotes the standard normal distribution. 
\item Similarly, for any fixed unit vector $u \in \mathbb{R}^{p}$, assume that 
\begin{equation}
\label{theta_asymptotic_condition}
    n^{-1}\left\|{Z}_{1:r}^{\top} X G u\right\| \geq c
\end{equation}
for some constant $c>0$ and sufficiently large $n$. Then, 
\begin{equation}
\label{theta}    \frac{\sqrt{n}\left(u^{\top}\hat{{\theta}}-u^{\top}\theta^*\right)}{ n^{-1}\left( u^{\top}  G X^{\top}\Tilde{Z}_{1:r} \Tilde{F}^{-1}_1(\hat{\gamma})\Tilde{Z}_{1:r}^{\top} X G u \right)^{1/2}} \ \rightarrow \mathcal{N}(0,1).
\end{equation}
\end{enumerate}

\end{theorem}

To understand condition \eqref{beta_asymptotic_condition}, note that according to \eqref{eq:par trans beta}, $u^{\top}\beta^*$ lies in the linear space spanned by coordinates of $Z_{(r+1):p}^\top X G u$. Condition  \eqref{beta_asymptotic_condition} essentially requires that this projected design does not vanish asymptotically. Otherwise, the inference of $u^{\top}\beta^*$ would not be meaningful. Condition \eqref{theta_asymptotic_condition} has a similar interpretation. These conditions are also needed in \cite{le2022linear}. Note also that in Theorem~\ref{theorem 2},  $n^{-1}\|\tilde{Z}_{(r+1):p}^{\top} X G u\|$, $n^{-1}\|\tilde{Z}_{1:r}^{\top} X G u\|$, and $\Tilde{O}$ are invariant to the choices of bases for $S_K(\hat{P})$ and $\operatorname{col}(X)$. 

Corollary~\ref{cor:uniqueness} and Theorem~\ref{theorem 2} provide the asymptotic distributions for $\hat{\alpha}$, $\hat{\beta}$, and $\hat{\theta}$ that can be used for inference purposes. In particular, \eqref{gamma}, \eqref{theta}, and \eqref{beta}  allow us to test the presence of pure network effect (against $\alpha^*=0$), pure covariate effect (against $\beta^*=0$), and the shared information between the two (against $\theta^*=0$), respectively. For example, when testing against $H_0: \alpha^*=0$, Theorem~\ref{theorem 2} indicates that we can use $n\hat{\alpha}^{\top}\tilde{O}\hat{\alpha}$ as the statistic for a $\chi^2$ test with $K-r$ degrees of freedom.


\section{Simulation Studies}\label{sec:Simulation Studies}
We next present simulation experiments evaluating estimation and inference under two perturbation mechanisms—random‑network perturbations and embedding‑induced perturbations. We study two instances of our model: subspace logistic and subspace Poisson regression.

\subsection{Perturbations from random network models}\label{secsec:random-graph}
We first study the performance of the proposed methods when the observed networks are subject to the perturbations introduced by random network models. In particular, the true relational matrix $P$ in our model is assumed to be a probability matrix taking values in $[0,1]$. Our true model is defined based on $S_K(P)$. The observed network is generated from $P$ following the ``inhomogeneous Erd\"{o}s-R\'{e}nyi" framework: for each $i<j, i,j\in [n]$, generate edges $A_{ij}\sim \text{Bernoulli}(P_{ij})$. Different matrices $P$ tend to generate networks with different structures and the perturbation comes from the randomness of this generating process.

\medskip

\emph{Random network models.} Regarding the network generative mechanisms, we use two low-rank models and a full-rank model. The first is the stochastic block model (SBM) of \cite{holland1983stochastic} with three communities, and the out-in-ratio between blocks is set to be $0.3$. The second model is the degree-corrected block model (DCBM) of \cite{karrer2011stochastic}, where the community connection matrix is the same as the SBM with additional degree parameters varying from 0.2 to 1 (before rescaling). These two models are generated by the R package \emph{randnet} \citep{randnet}. The full-rank model is the one from \cite{zhang2017estimating}, in which $P$ is constructed from the graphon function $g(\mu,\nu)=c /\{1+\exp [15(0.8|\mu-\nu|)^{4 / 5}-0.1]\}$. This graphon model gives a banded matrix along the diagonal. We therefore refer to it as the ``diagonal graphon" model. In all experiments, we vary the sample size $n$ from $500$ to $4000$, and the expected average degree is set to be $\varphi_n=2 \log n, \sqrt{n}, n^{2/3}$ to demonstrate the effect of varying network density.

\medskip

\emph{Subspace and covariates.} Following \cite{le2022linear}, we construct $X \in \mathbb{R}^{n \times p}$ using the eigenvectors $w_1, \ldots, w_n$ from $P$ as follows: Set $X_1 / \sqrt{n}=w_1 ;$ Set $X_2 / \sqrt{n}= w_2/5+ 2\sqrt{6}w_4/5$ \footnote{The main purpose of this design is to find an eigenvector that is orthogonal to $w_1,..., w_3$, so we can easily control the values such as $r, s, \sigma_{r+1}$, etc. It does not have to be $w_4$. }. This configuration yields a design with $r=1, s=1$, 
and the singular values of $Z^\top W$ in \eqref{defin:model} are well separated: $\sigma_1 = 1, \quad \sigma_2 = 1/5, \quad \sigma_3 = 0, $  
ensuring a clear distinction between signal and noise components.  
This separation guarantees that all regularity conditions are satisfied, except for Assumption \ref{cond:A4}, which specifically concerns network perturbation magnitude.  
By keeping  
$\min \left\{\left(1-\sigma_{r+1}\right)^3, \sigma_{r+s}^3\right\} $  
fixed, we can then systematically control and vary the magnitude of perturbations by the average degree of the network.
We set $\beta^*=(0,0.5)^{\top}$ and $\theta^*=(0.5,0)^{\top}$ in all settings. 
Similarly, we set $\gamma_{3:4}^*=(0.5,0.5)^{\top}$.
Then we generate $Y$ from the logistic regression model or Poisson regression model separately, following
\begin{eqnarray*}
Y \mid X &\sim& \text{Bernoulli} \left\{  \frac{\exp\left( X\beta^* + X\theta^* + \alpha^* \right)}{1 + \exp\left( X\beta^* + X\theta^* + \alpha^* \right)} \right\}, \\
Y \mid X &\sim& \text{Poisson} \left\{ \exp\left( X\beta^* + X\theta^* + \alpha^* \right) \right\}.
\end{eqnarray*}
In the model fitting process, we always use the observed adjacency matrix $A$ to approximate the true eigenspace.

\medskip

\emph{Evaluation criterion.} For model estimation accuracy, we measure the performance by the mean squared error (MSE) on $\beta_2$, defined as $|\hat{\beta}_2-\beta_2^*|^2$, the mean square prediction error (MSPE) defined as $\|\hat{Y}-\mathbb{E}Y\|^2/n$. For inference, we evaluate the coverage probability of the 95\% confidence interval for $\beta_2$ \footnote{It can be shown that, for $\beta_1$, where $u = (1, 0)^\top$, we have
\(
\left\| Z_2^\top X G u \right\| = 0.
\) 
This configuration violates the requirement
\(
\frac{1}{n} \left\| Z_{r+1:p}^\top X G u \right\| \geq c
\)
in \eqref{beta_asymptotic_condition} for Theorem~\ref{theorem 2}. This implies that the parameter subspace relevant for inference is degenerate. We construct this setting intentionally to examine the theoretical assumption. 
}

\begin{table}[ht]
        \caption{Median MSE ($\times 10^{2}$) and coverage probability for subspace logistic regression under random network perturbations.}
        \centering
{\begin{tabular}{cccccccc}
\hline
\hline
\multirow{2}{*}{ n } & \multirow{2}{*}{ avg.\ degree } & \multicolumn{2}{c}{ SBM} & \multicolumn{2}{c}{DCBM} & \multicolumn{2}{c}{Diag}  \\
 & & MSE  & Coverage & MSE   & Coverage & MSE  & Coverage \\
\hline 
\multirow{3}{*}{ 500 } & $2 \log n$ & 1.16 & 94.6\% & 1.18 & 95.4\% & 1.31 & 92.4\% \\
& $\sqrt{n}$ & 1.15 & 94.8\% & 1.19 & 94.8\% & 1.18 & 93.8\% \\ 
& $n^{2 / 3}$ & 1.13 & 95.2\% & 1.22 & 95.3\% & 1.13 & 94.4\% \\
\hline 
\multirow{3}{*}{ 1000} & $2 \log n$ & 0.56 & 94.7\% & 0.56 & 95.1\% & 0.64 & 93.5\% \\
& $\sqrt{n}$  & 0.57 & 94.9\% & 0.57 & 95.0\% & 0.63 & 93.9\% \\
& $n^{2 / 3}$ & 0.58 & 95.0\% & 0.57 & 95.1\% & 0.60 & 94.7\% \\
\hline 
\multirow{3}{*}{ 2000 } & $2 \log n$ & 0.35  & 93.1\% & 0.29 & 95.1\% & 0.28 & 93.7\% \\
& $\sqrt{n}$  & 0.31 & 94.7\% & 0.28 & 95.0\% & 0.27 & 94.4\% \\
& $n^{2 / 3}$ & 0.30 & 95.1\% & 0.28 & 95.1\% & 0.26 & 94.9\% \\
\hline 
\multirow{3}{*}{ 4000 }& $2 \log n$ & 0.16 & 92.7\% & 0.15 & 94.2\% & 0.15 & 93.5\% \\ 
& $\sqrt{n}$ & 0.14 & 94.9\% & 0.14 & 95.0\% & 0.14 & 94.6\% \\
& $n^{2 / 3}$  & 0.14 & 95.0\% & 0.14 & 95.1\% & 0.14  & 94.8\% \\ 
\hline
\hline
\end{tabular}}
    \label{tableLogit3}
\end{table}

\begin{table}[h]
        \caption{Median MSPE ($\times 10^2$) for subspace logistic regression and benchmarks under traditional random network perturbations.}
        \centering
{\begin{tabular}{cccccccccc}
\hline \hline
n&Network&avg.\ degree&Our Model &Logistic Reg&RNC\\
\hline \hline
\multirow{9}{*}{500}
&& $2 \log n$ & 1.11 &1.34 & 2.56 \\
&SBM & $\sqrt{n}$ & 0.64 & 1.34  & 2.54 \\
&& $n^{2 / 3}$ & 0.31 & 1.34  & 2.50 \\
\cline{2-6}
&& $2 \log n$ & 1.05 & 2.48  & 2.47 \\
&DCBM & $\sqrt{n}$ & 0.60 & 2.48 & 2.46\\
&& $n^{2 / 3}$ & 0.28 & 2.48 & 2.47\\
\cline{2-6}
&& $2 \log n$ & 0.38 & 0.67 & 2.33 \\
&Diag & $\sqrt{n}$ & 0.26 & 0.67 & 2.33\\
&& $n^{2 / 3}$ & 0.18 & 0.67 & 2.32\\
\cline{1-6}
\multirow{9}{*}{1000}&& $2 \log n$ & 0.96 & 2.01  & 2.22 \\
&SBM & $\sqrt{n}$ & 0.43 & 2.01  & 2.19 \\
&& $n^{2 / 3}$& 0.18 & 2.01 & 2.17 \\
\cline{2-6}
&& $2 \log n$ & 0.76& 1.99  & 2.47 \\
&DCBM & $\sqrt{n}$ & 0.40 & 1.99 & 2.49\\
&& $n^{2 / 3}$ & 0.16  & 1.99 & 2.49 \\
\cline{2-6}
&& $2 \log n$ & 0.32 & 2.13  & 2.37 \\
&Diag & $\sqrt{n}$ & 0.17 & 2.13 & 2.36 \\
&& $n^{2 / 3}$ & 0.10 & 2.13 & 2.34 \\
\hline
\hline
\end{tabular}}
    \label{MSPELogit3}
\end{table}

\begin{table}[ht]
        \caption{Median MSE ($\times 10^{-2}$) for $\hat{\beta}_2$ and rejection rate of $\chi^2$ test for subspace logistic regression, under random network perturbations when $\left\|\alpha^*\right\|=0$.}
        \centering
{\begin{tabular}{cccccccc}
\hline
\hline
\multirow{2}{*}{ n } & \multirow{2}{*}{ avg.\ degree } & \multicolumn{2}{c}{ SBM} & \multicolumn{2}{c}{DCBM} & \multicolumn{2}{c}{ Diag}  \\
 & & MSE  & Rejection & MSE   & Rejection & MSE  & Rejection \\
\hline 
\multirow{3}{*}{ 500 } & $2 \log n$ & 1.06 & 4.9\% & 1.10 & 4.6\% & 1.04 & 5.0\% \\
& $\sqrt{n}$ & 1.05 & 4.8\% & 1.11 & 4.8\% & 1.03 & 4.8\% \\ 
& $n^{2 / 3}$ & 1.07 & 4.8\% & 1.10 & 5.0\% & 1.00 & 4.8\% \\
\hline 
\multirow{3}{*}{ 1000} & $2 \log n$ & 0.51 & 4.8\% & 0.53 & 4.8\% & 0.52 & 4.7\% \\
& $\sqrt{n}$  & 0.51 & 4.9\% & 0.53 & 4.8\% & 0.51 & 4.7\% \\
& $n^{2 / 3}$ & 0.52 & 4.9\% & 0.53 & 5.0\% & 0.50 & 4.7\% \\
\hline 
\multirow{3}{*}{ 2000 } & $2 \log n$ & 0.27  & 4.9\% & 0.27 & 4.9\% & 0.25 & 4.9\% \\
& $\sqrt{n}$  & 0.27 & 4.9\% & 0.27 & 4.8\% & 0.24 & 5.0\% \\
& $n^{2 / 3}$ & 0.27 & 5.0\% & 0.27 & 5.0\% & 0.24 & 5.1\% \\
\hline 
\multirow{3}{*}{ 4000 }& $2 \log n$ & 0.13 & 4.9\% & 0.13 & 4.9\% & 0.12 & 5.0\% \\ 
& $\sqrt{n}$ & 0.13 & 5.0\% & 0.13 & 5.1\% & 0.12 & 5.1\% \\
& $n^{2 / 3}$  & 0.13 & 5.0\% & 0.13 & 4.8\% & 0.12 & 4.8\% \\ 
\hline
\hline
\end{tabular}}
    \label{tableLogit3_0}
\end{table}

\emph{Benchmark methods.} A standard logistic regression model without the network component was also included for comparison. In addition, we include the RNC method from \cite{li2019prediction}. The model fitting parameter is chosen by 10-fold cross-validation.

\medskip

\emph{Calculation procedure.}  In order to assess the model’s performance with the randomness from both the response $Y$ and adjacency matrix $A$, we generate $100$ unique adjacency matrices $A$ based on one relational matrix $P$ for each simulation scenario.  For each given $A$, $1000$ replicates of $Y$'s are generated, and the performance metrics (MSE and MSPE) and coverage probability are computed based on the Monte Carlo approximation from these 1000 instantiations. In the outer loop, we repeatedly generate $A$ 100 times, and the median value of the resulting coverage probabilities and MSEs are reported.

\medskip

\begin{table}[ht]
    \centering
    \caption{Median MSE ($\times 10^{2}$) and coverage probability for subspace Poisson regression under random network perturbations.}
{\begin{tabular}{cccccccc}
\hline
\hline
\multirow{2}{*}{ n } & \multirow{2}{*}{ avg.\ degree } & \multicolumn{2}{c}{ SBM} & \multicolumn{2}{c}{DCBM} & \multicolumn{2}{c}{ Diag}  \\
 &  & MSE  & Coverage & MSE  & Coverage & MSE & Coverage \\
\hline 
\multirow{3}{*}{ 500 } & $2 \log n$ & 0.35 & 75.8\% & 0.21 & 75.2\% & 0.53 & 72.4\% \\
& $\sqrt{n}$  & 0.16 & 86.2\% & 0.11 & 90.4\% & 0.38 & 84.9\% \\ 
& $n^{2 / 3}$ & 0.10 & 93.5\% & 0.09 & 93.8\% & 0.23 & 93.4\% \\
\hline 
\multirow{3}{*}{ 1000} & $2 \log n$ & 0.08 & 87.2\% & 0.27 & 29.9\% & 0.27 & 57.7\% \\
& $\sqrt{n}$  & 0.06 & 92.9\% & 0.07 & 82.4\% & 0.11 & 88.8\% \\
& $n^{2 / 3}$  & 0.06 & 94.2\% & 0.04 & 93.5\% & 0.08 & 93.8\% \\ 
\hline 
\multirow{3}{*}{ 2000 } & $2 \log n$ & 0.13 & 43.4\% & 0.14 & 29.1\% & 0.09 & 77.0\% \\
& $\sqrt{n}$ & 0.03 & 86.3\% & 0.03 & 86.4\% & 0.05 & 91.7\% \\
& $n^{2 / 3}$ & 0.03 & 94.2\% & 0.02 & 94.0\% & 0.04 & 93.7\% \\
\hline 
\multirow{3}{*}{ 4000 } & $2 \log n$ & 0.04 & 71.4\% & 0.04 & 89.8\% & 0.61 & 0\% \\ 
& $\sqrt{n}$ & 0.01 & 93.4\% & 0.01 & 94.3\% & 0.06 & 58.3\% \\
& $n^{2 / 3}$ & 0.01 & 94.5\% & 0.01 & 94.6\% & 0.02 & 93.0\% \\ 
\hline
\hline
\end{tabular}}
    \label{tablePoi3}
\end{table}

\begin{table}[ht]
        \caption{Median MSPE for subspace Poisson regression and benchmarks under traditional random network perturbations.}
        \centering
{\begin{tabular}{cccccccccc}
\hline \hline
n&Network&avg.\ degree&Our Model &Poisson Reg&RNC\\
\hline \hline
\multirow{9}{*}{500}
&& $2 \log n$ & 2.30 & 2.81 & 1.68 \\
&SBM & $\sqrt{n}$ & 1.41 & 2.81  & 1.69 \\
&& $n^{2 / 3}$ & 0.53 & 2.81  & 1.80 \\
\cline{2-6}
&& $2 \log n$ & 1.09 & 2.58 & 3.93 \\
&DCBM & $\sqrt{n}$ & 0.72 & 2.58 & 3.89\\
&& $n^{2 / 3}$ & 0.32 & 2.58 & 3.89\\
\cline{2-6}
&& $2 \log n$ & 0.44 & 1.11 & 1.06 \\
&Diag & $\sqrt{n}$ & 0.25 & 1.11 & 1.06\\
&& $n^{2 / 3}$ & 0.07 & 1.11 & 1.05\\
\cline{1-6}
\multirow{9}{*}{1000}&& $2 \log n$ & 1.66 & 4.20  & 3.22\\
&SBM & $\sqrt{n}$ & 0.82 & 4.20  & 3.39 \\
&& $n^{2 / 3}$& 0.27 & 4.20 & 2.72 \\
\cline{2-6}
&& $2 \log n$ & 0.96 & 5.58 & 2.48\\
&DCBM & $\sqrt{n}$ & 0.58 & 5.58 & 2.20\\
&& $n^{2 / 3}$ & 0.21 & 5.58 & 2.29 \\
\cline{2-6}
&& $2 \log n$ & 0.22 & 0.74 & 1.51 \\
&Diag & $\sqrt{n}$ & 0.09 & 0.74 & 1.45 \\
&& $n^{2 / 3}$ & 0.03 & 0.74 & 1.28 \\
\hline
\hline
\end{tabular}}
    \label{MSPEPoi3}
\end{table}

\begin{table}[ht]
    \centering
    \caption{Median MSE ($\times 10^{-2}$) for $\hat{\beta}_2$ and rejection rate of $\chi^2$ test for subspace Poisson regression, under random network perturbations when $\left\|\alpha^*\right\|=0$.}
{\begin{tabular}{cccccccc}
\hline
\hline
\multirow{2}{*}{ n } & \multirow{2}{*}{ avg.\ degree } & \multicolumn{2}{c}{ SBM} & \multicolumn{2}{c}{DCBM} & \multicolumn{2}{c}{ Diag}  \\
 &  & MSE  & Rejection & MSE  & Rejection & MSE & Rejection \\
\hline 
\multirow{3}{*}{ 500 } & $2 \log n$ & 0.11 & 4.8\% & 0.11 & 5.0\% & 0.14 & 4.8\% \\
& $\sqrt{n}$  & 0.11 & 5.0\% & 0.11 & 4.8\% & 0.14 & 5.2\% \\ 
& $n^{2 / 3}$ & 0.11 & 5.0\% & 0.11 & 5.2\% & 0.12 & 5.1\% \\
\hline 
\multirow{3}{*}{ 1000} & $2 \log n$ & 0.06 & 4.9\% & 0.06 & 5.2\% & 0.08 & 4.9\% \\
& $\sqrt{n}$  & 0.06 & 4.9\% & 0.05 & 5.0\% & 0.07 & 4.8\% \\
& $n^{2 / 3}$  & 0.06 & 4.8\% & 0.05 & 5.1\% & 0.06 & 4.9\% \\ 
\hline 
\multirow{3}{*}{ 2000 } & $2 \log n$ & 0.03 & 5.0\% & 0.03 & 5.1\% & 0.04 & 5.0\% \\
& $\sqrt{n}$ & 0.03 & 4.8\% & 0.03 & 4.9\% & 0.03 & 4.8\% \\
& $n^{2 / 3}$ & 0.03 & 4.9\% & 0.03 & 4.8\% & 0.03 & 5.1\% \\
\hline 
\multirow{3}{*}{ 4000 }& $2 \log n$ & 0.01 & 4.9\% & 0.01 & 5.1\% & 0.02 & 5.0\% \\ 
& $\sqrt{n}$ & 0.01 & 4.8\% & 0.01 & 4.9\% & 0.02 & 5.0\% \\
& $n^{2 / 3}$ & 0.01 & 4.8\% & 0.01 & 5.0\% & 0.02 & 5.0\% \\ 
\hline
\hline
\end{tabular}}
    \label{tablePoi3_0}
\end{table}

Table \ref{tableLogit3} shows how our method performs under the network subspace logistic regression model. Overall, the performance improves with the sample size $n$ and the expected average degree of the network model. The denser networks make the problem easier because the concentration of the adjacency matrix to the true $P$ is better. Table \ref{tableLogit3} shows that if $A$ is used under the current network generative procedure, an average degree higher than $\sqrt{n}$ is sufficient for good inference accuracy. This is consistent with the observation in  \cite{le2022linear}. Table \ref{MSPELogit3} presents the MSE comparison between our model and the two benchmarks. Our method clearly outperforms the RNC and standard logistic regression.

Table~\ref{tableLogit3_0} summarizes the model's performance while $\left\|\alpha^{*}\right\|=0$. The specific focus is the rejection rate of the $\chi^2$ test at the level 0.05.  The results suggest that the $\chi^2$ test performs well under the null hypothesis with the desired level of type I error control.

Under the Poisson model, we use the same configuration except for replacing the logistic distribution with the Poisson distribution. The same results are presented in Table \ref{tablePoi3}, Table \ref{MSPEPoi3} and Table \ref{tablePoi3_0}. When the average degree surpasses the order of $\sqrt{n}$, the asymptotic validity holds. Under the diagonal graphon model, the perturbation has a stronger impact, but the inference remains approximately correct with the current sample size for sufficiently dense networks.  
The overall message remains the same as in the logistic regression setting.


\subsection{Network perturbations from deep-learning-based embedding methods}
\label{secsec:embedding}
We now consider another application scenario in which the proposed model can be used.  
Suppose we want to use embedding methods from deep-learning community to extract the network information. Multiple recent works \citep{pozek2019performance,pranathi2021node,liu2024controlling} take this strategy to incorporate network information, with the belief that these methods can capture high-order network relations more effectively by their highly nonlinear operations. 

\begin{table}[ht]
    \centering
        \caption{Median MSE ($\times 10^{2}$), coverage probability and MSPE ($\times 10^{2}$) for subspace logistic regression with different types of network of size $1000$ under network embedding perturbations.}
\begin{tabular}{ccccccccccc}
\hline \hline
\multirow{2}{*}{Method} &\multirow{2}{*}{Network}&\multirow{2}{*}{avg.\ degree}& 
\multirow{2}{*}{MSE}&\multirow{2}{*}{Coverage} &\multicolumn{2}{c}{MSPE}\\
&&&&&Our Model &Logistic Reg\\
\hline \hline
\multirow{9}{*}{DeepWalk}
&& $2 \log n$ & 0.60 &94.6\% & 0.12 & 0.61 \\
&SBM & $\sqrt{n}$ & 0.59 & 94.9\%  & 0.11 & 1.75 \\
&& $n^{2 / 3}$ & 0.58 & 94.9\% & 0.11 & 2.04 \\
\cline{2-7}
&& $2 \log n$ & 0.75 & 94.8\%  & 0.13 & 0.34 \\
&DCBM & $\sqrt{n}$ & 0.62 & 94.5\%  & 0.13 & 1.98\\
&& $n^{2 / 3}$ & 0.60 & 94.8\%  & 0.12 & 1.31 \\
\cline{2-7}
&& $2 \log n$ & 0.77 & 95.2\%  & 0.08 & 1.08 \\
&Diag & $\sqrt{n}$ & 0.58 & 95.1\%  & 0.09 & 1.56 \\
&& $n^{2 / 3}$ & 0.62 & 95.2\%  & 0.09 & 1.54 \\
\cline{1-7}
\multirow{9}{*}{Node2Vec}&& $2 \log n$ & 0.59 & 94.6\%  & 0.12 & 1.61 \\
&SBM & $\sqrt{n}$ & 0.59 & 94.9\%  & 0.12 & 1.35 \\
&& $n^{2 / 3}$& 0.60 & 94.9\% & 0.12 & 1.90 \\
\cline{2-7}
&& $2 \log n$ & 0.69 & 94.8\%  & 0.12 & 0.29 \\
&DCBM & $\sqrt{n}$ & 0.68 & 94.9\% & 0.12 & 0.24\\
&& $n^{2 / 3}$ & 0.65 & 94.8\% & 0.12 & 1.11 \\
\cline{2-7}
&& $2 \log n$ & 0.56 & 94.9\%  & 0.09 & 1.45\\
&Diag & $\sqrt{n}$ & 0.52 & 95.0\%  & 0.08 & 1.25 \\
&& $n^{2 / 3}$ & 0.58 & 94.9\%  & 0.09 & 1.31 \\
\cline{1-7}
\multirow{9}{*}{Diff2Vec}&& $2 \log n$ & 0.56 & 94.9\% & 0.09 & 1.12 \\
&SBM & $\sqrt{n}$ & 0.52 & 95.0\% & 0.09 & 1.08 \\
&& $n^{2 / 3}$ & 0.58 & 94.9\% & 0.09 & 1.26 \\
\cline{2-7}
&& $2 \log n$ & 0.66 & 94.3\% & 0.14 & 1.18 \\
&DCBM & $\sqrt{n}$ & 0.72 & 94.7\%  & 0.12 & 1.21\\
&& $n^{2 / 3}$ & 0.62 & 94.8\% & 0.11 & 1.19 \\
\cline{2-7}
&& $2 \log n$ & 0.56 & 94.9\%  & 0.09 & 0.86 \\
&Diag & $\sqrt{n}$ & 0.52 & 95.0\% & 0.09 & 1.01 \\
&& $n^{2 / 3}$ & 0.58 & 94.9\%  & 0.09 & 1.09 \\
\hline
\hline
\end{tabular}
\label{embed1000logistic}
\end{table}

Our subspace generalized linear model, with its flexibility in specifying a proper subspace $S_K(P)$, can seamlessly leverage this embedding information. Specifically, we can assume the inner product similarities of the embedded vectors as the perturbed relational information $\hat{P}$, with the true relational matrix being an unobserved similarity matrix that can be different from the random embedded similarities. In these cases, even if the network is usually treated as fixed, the embedding algorithms are typically random by nature. This randomness in embeddings raises concerns about the validity of modeling and inference if one uses a specific embedding in the model. In this section, we use simulation experiments to evaluate the validity of our model's inference under such perturbations of embeddings. The study of statistical properties of the embedding methods is rare in the literature. To our knowledge, \citet{zhang2023theoretical} provides related analysis for community detection; we are not aware of prior empirical studies examining how deep‑learning–based embeddings affect downstream inference.

We consider three popular network embedding methods, DeepWalk \citep{perozzi2014deepwalk}, Node2Vec \citep{node2vec-kdd2016}, and Diff2Vec \citep{rozemberczki2018fast} to demonstrate these scenarios. DeepWalk was one of the earliest graph embedding methods from the deep learning community, and Node2Vec is a generalization of DeepWalk. Diff2Vec uses the more recent diffusion framework to define the embeddings. The implementations of DeepWalk and Node2Vec are available in the Python package \emph{node2vec} \citep{node2vec-kdd2016}, and Diff2Vec is implemented in the Python package \emph{karateclub} \citep{karateclub}. In our simulation, we always use the recommended configurations of these methods. For DeepWalk and Node2Vec, each embedding is based on $10$ walks per node of length $80$. For Node2Vec, the return probability is set to $0.5$. For Diff2Vec, we use $20$ trees per node of size $80$. The network embedding dimension is always set to $3$.

\emph{Design of relational matrix.} We first generate a network $A$ from one of the three models in Section~\ref{secsec:random-graph}, and fix the network $A$. Given $A$, all three embedding methods are random and result in different embeddings each time. Therefore, for each embedding method, suppose $\mathcal{F}$ is the embedding of $A$ and, intuitively, we can use $\mathcal{F}\mathcal{F}^\top$ as the available similarity matrix from data. The perturbation of network information comes from the randomness of $\mathcal{F}$. Specifically, in this context, we set the true relational matrix as the oracle central similarity $P = \mathbb{E}[\mathcal{F}\mathcal{F}^\top]$, and $\hat{P} = \mathcal{F}\mathcal{F}^\top$. The design matrix $X$ and other quantities are generated in the same manner as in Section~\ref{secsec:random-graph} based on the current $P$.

\begin{table}[ht]
    \centering
        \caption{Median MSE ($\times 10^{2}$), coverage probability and MSPE ($\times 10^{2}$) for subspace logistic regression with different types of network of size $2000$ under network embedding perturbations.}
\begin{tabular}{ccccccccccc}
\hline \hline
\multirow{2}{*}{Method} &\multirow{2}{*}{Network}&\multirow{2}{*}{avg.\ degree}& 
\multirow{2}{*}{MSE}&\multirow{2}{*}{Coverage} &\multicolumn{2}{c}{MSPE}\\
&&&&&Our Model &Logistic Reg\\
\hline \hline
\multirow{9}{*}{DeepWalk}
&& $2 \log n$ & 0.30 &94.4\% & 0.08 & 1.30 \\
&SBM & $\sqrt{n}$ & 0.29 & 94.8\%  & 0.07 & 2.18 \\
&& $n^{2 / 3}$ & 0.29 & 94.7\% & 0.08 & 1.90 \\
\cline{2-7}
&& $2 \log n$ & 0.48& 94.5\%  & 0.09 & 1.90 \\
&DCBM & $\sqrt{n}$ & 0.32 & 94.8\%  & 0.08 & 0.71\\
&& $n^{2 / 3}$ & 0.30 & 94.8\%  & 0.08 & 1.68 \\
\cline{2-7}
&& $2 \log n$ & 0.31 & 95.0\%  & 0.04 & 1.17 \\
&Diag & $\sqrt{n}$ & 0.29 & 95.0\%  & 0.05 & 1.17 \\
&& $n^{2 / 3}$ & 0.29 & 95.0\%  & 0.05 & 1.24 \\
\cline{1-7}
\multirow{9}{*}{Node2Vec}&& $2 \log n$ & 0.30 & 94.3\%  & 0.09 & 1.80 \\
&SBM & $\sqrt{n}$ & 0.28& 94.8\%  & 0.08 & 2.08 \\
&& $n^{2 / 3}$& 0.29 & 94.9\% & 0.08 & 2.08 \\
\cline{2-7}
&& $2 \log n$ & 0.37& 94.7\%  & 0.08 & 1.09 \\
&DCBM & $\sqrt{n}$ & 0.32 & 94.8\% & 0.08 & 1.17\\
&& $n^{2 / 3}$ & 0.30 & 94.8\% & 0.08 & 1.79 \\
\cline{2-7}
&& $2 \log n$ & 0.31 & 94.9\%  & 0.05 & 0.88\\
&Diag & $\sqrt{n}$ & 0.29 & 95.0\%  & 0.05 & 1.21 \\
&& $n^{2 / 3}$ & 0.29 & 95.1\%  & 0.05 & 1.18 \\
\cline{1-7}
\multirow{9}{*}{Diff2Vec}&& $2 \log n$ & 0.32 & 93.5\% & 0.10 & 1.07 \\
&SBM & $\sqrt{n}$ & 0.30 & 94.8\% & 0.07 & 1.21 \\
&& $n^{2 / 3}$ & 0.30 & 94.6\% & 0.07 & 1.08 \\
\cline{2-7}
&& $2 \log n$ & 0.30 & 94.2\% & 0.10 & 1.15 \\
&DCBM & $\sqrt{n}$ & 0.32 & 94.4\%  & 0.08 & 1.10\\
&& $n^{2 / 3}$ & 0.30 & 94.6\% & 0.07 & 0.99 \\
\cline{2-7}
&& $2 \log n$ & 0.30 & 93.8\%  & 0.12 & 1.22 \\
&Diag & $\sqrt{n}$ & 0.33 & 94.7\% & 0.09 & 1.19 \\
&& $n^{2 / 3}$ & 0.32 & 94.8\%  & 0.08 & 1.14 \\
\hline
\hline
\end{tabular}
\label{embed2000logistic}
\end{table}

\begin{table}[ht]
    \centering
        \caption{Median MSE ($\times 10^{2}$), coverage probability, and MSPE for subspace Poisson regression with different types of network of size $1000$ under network embedding perturbations.}
\begin{tabular}{ccccccccccc}
\hline \hline
\multirow{2}{*}{Method} &\multirow{2}{*}{Network}&\multirow{2}{*}{avg.\ degree}&\multirow{2}{*}{MSE}& 
\multirow{2}{*}{Coverage} &\multicolumn{2}{c}{MSPE}\\
&&&&&Our Method &Poisson Reg\\
\hline \hline
\multirow{9}{*}{DeepWalk}
&& $2 \log n$ & 0.14 & 93.3\%  & 2.12 & 12.1 \\
&SBM & $\sqrt{n}$   & 0.12 & 94.5\% & 2.02 & 62.7 \\
&& $n^{2 / 3}$  & 0.10 & 94.5\% & 1.92 & 56.7 \\
\cline{2-7}
&& $2 \log n$  & 0.16 & 93.5\% & 2.97 & 23.0 \\
&DCBM & $\sqrt{n}$ & 0.09 & 92.9\%  & 3.45 & 59.9 \\
&& $n^{2 / 3}$   & 0.12 & 94.3\% & 1.83 & 21.7 \\
\cline{2-7}
 && $2 \log n$  & 0.20 & 94.9\% & 1.12 & 11.7 \\
&Diag & $\sqrt{n}$  & 0.11 & 94.8\% & 0.88 & 24.3 \\
&& $n^{2 / 3}$  & 0.13 & 94.9\% & 0.84 & 25.0\\
\cline{1-7}
\multirow{9}{*}{Node2Vec}&& $2 \log n$   & 0.12 & 93.0\% & 3.09 & 64.7 \\
&SBM & $\sqrt{n}$   & 0.13 & 94.2\% & 1.84 & 19.4 \\
&& $n^{2 / 3}$   & 0.13 & 94.3\% & 1.97 & 43.5 \\
\cline{2-7}
&& $2 \log n$  & 0.09 & 93.7\% & 3.74 & 27.8\\
&DCBM & $\sqrt{n}$ & 0.09 & 93.9\%  & 2.37 &  4.33 \\
&& $n^{2 / 3}$  & 0.08 & 93.7\% & 3.58 & 42.1 \\
\cline{2-7}
&& $2 \log n$   & 0.16 & 94.7\% & 0.62 & 12.7\\
&Diag & $\sqrt{n}$  & 0.10 & 95.0\% & 0.80 & 13.7 \\
&& $n^{2 / 3}$  & 0.11 & 94.7\% & 0.89 & 22.2 \\
\cline{1-7}
\multirow{9}{*}{Diff2Vec}
&& $2 \log n$ & 0.16 & 90.9\% & 2.30 & 25.5\\
&SBM & $\sqrt{n}$  & 0.15 & 93.9\% & 1.72 & 16.0 \\
&& $n^{2 / 3}$  & 0.15 & 94.4\% & 1.33 & 27.3 \\
\cline{2-7}
&& $2 \log n$ & 0.15 & 91.9\% & 3.12 & 27.7\\
&DCBM & $\sqrt{n}$  & 0.18 & 93.5\% & 2.19 & 24.5 \\
&& $n^{2 / 3}$  & 0.11 & 93.9\% & 2.14 & 20.2 \\
\cline{2-7}
&& $2 \log n$ & 0.16 & 90.9\% & 2.30 & 6.44 \\
&Diag & $\sqrt{n}$ & 0.27 & 94.9\% & 0.92 & 9.21 \\
&& $n^{2 / 3}$  & 0.18 & 94.9\% & 0.79 & 12.1 \\
\hline
\hline
\end{tabular}
\label{embed1000Poisson}
\end{table}

\begin{table}[ht]
    \centering
        \caption{Median MSE ($\times 10^{2}$), coverage probability and MSPE for subspace Poisson regression with different types of network of size $2000$ under network embedding perturbations.}
\begin{tabular}{ccccccccccc}
\hline \hline
\multirow{2}{*}{Method} &\multirow{2}{*}{Network}&\multirow{2}{*}{avg.\ degree}&\multirow{2}{*}{MSE}& 
\multirow{2}{*}{Coverage} &\multicolumn{2}{c}{MSPE}\\
&&&&&Our Method &Poisson Reg\\
\hline \hline
\multirow{9}{*}{DeepWalk}
&& $2 \log n$ & 0.06 & 92.9\%  & 1.99 & 28.6 \\
&SBM & $\sqrt{n}$   & 0.06 & 94.0\% & 1.58 & 45.1 \\
&& $n^{2 / 3}$  & 0.06 & 93.8\% & 1.27 & 22.4 \\
\cline{2-7}
&& $2 \log n$  & 0.07& 91.8\% & 2.70 & 41.1 \\
&DCBM & $\sqrt{n}$ & 0.04 & 93.7\%  & 2.21 & 19.4 \\
&& $n^{2 / 3}$   & 0.06& 93.8\% & 1.59 & 21.1 \\
\cline{2-7}
 && $2 \log n$   & 0.07& 94.9\% & 0.91 & 12.2 \\
&Diag & $\sqrt{n}$   & 0.05& 94.8\% & 0.65 & 23.5 \\
&& $n^{2 / 3}$  & 0.05 & 94.7\% & 0.63 & 25.4 \\
\cline{1-7}
\multirow{9}{*}{Node2Vec}&& $2 \log n$   & 0.08 & 93.2\% & 1.93 & 40.8 \\
&SBM & $\sqrt{n}$   & 0.06& 93.7\% & 1.54 & 35.1 \\
&& $n^{2 / 3}$   & 0.07& 94.0\% & 1.23 & 24.7 \\
\cline{2-7}
&& $2 \log n$  & 0.06 & 91.9\% & 2.77 & 60.7 \\
&DCBM & $\sqrt{n}$ & 0.05 & 94.0\%  &2.23 &  22.1 \\
&& $n^{2 / 3}$   & 0.05& 93.7\% & 2.40 & 18.0 \\
\cline{2-7}
&& $2 \log n$   & 0.06 & 94.7\% & 0.62 & 19.2 \\
&Diag & $\sqrt{n}$  & 0.06 & 94.7\% & 0.61 & 21.6 \\
&& $n^{2 / 3}$  & 0.05 & 94.5\% & 0.69 & 21.3 \\
\cline{1-7}
\multirow{9}{*}{Diff2Vec}&& $2 \log n$  & 0.09 & 87.7\%  & 2.58 & 22.1 \\
& SBM & $\sqrt{n}$   & 0.07 & 94.0\% & 1.19 & 34.7 \\
&& $n^{2 / 3}$   & 0.07 & 93.9\% & 1.42 & 13.4 \\
\cline{2-7}
&& $2 \log n$ & 0.07 & 93.2\% & 2.41 & 26.4\\
&DCBM & $\sqrt{n}$  & 0.07 & 93.7\% & 1.85 & 20.7 \\
&& $n^{2 / 3}$  & 0.06 & 93.0\% & 1.59 & 10.0 \\
\cline{2-7}
&& $2 \log n$ & 0.12 & 90.5\% & 1.31 & 17.2 \\
&Diag & $\sqrt{n}$ & 0.07 & 93.9\% & 0.86 & 15.0 \\
&& $n^{2 / 3}$  & 0.08 & 94.4\% & 0.67 & 15.1 \\
\hline
\hline
\end{tabular}
\label{embed2000Poisson}
\end{table}

Tables \ref{embed1000logistic} to \ref{embed2000Poisson} present the performance metrics under perturbations from different embedding algorithms, evaluated across three network types with varying average degrees for sample sizes $n = 1000,\ 2000$. 
Additional results for $n=500$ are provided in Section~\ref{sec:embed n=500}.
Note that the previous benchmark method, RNC, is not applicable in this new setting and was removed from the comparison. This also demonstrates the flexibility of our subspace-based model.

For each of the three embedding mechanisms, the estimation accuracy and inference correctness (measured by the coverage probability) exhibit mild improvements with the increase of network density. But overall, the density is no longer a very clear indicator of the perturbation level in this case, because the perturbation is contributed by the randomness of the embedding algorithms. Among the three mechanisms, Diff2Vec is more vulnerable to density change, and overall results in larger perturbations. For example, on networks with an average degree of $2\log n$ for sample size $n=2000$, the coverage probability misses the target level by a lot. DeepWalk and Node2Vec are more robust and the resulting perturbations tend to satisfy the small projection perturbation requirements. For embedding methods, our model estimation remains accurate, and the statistical inference is still approximately correct. This result demonstrates the applicability of our inference framework in broader scenarios in modern machine learning.

\section{Social and Educational Effect Study of School Conflicts}\label{sec:data}

With the proposed model, we will analyze data from the school conflict study introduced in Section~\ref{secsec:motivation}. Following the original strategy of \cite{paluck2016changing}, we use self-reported wearing of an orange wristband to assess the impact of anti-conflict interventions on behaviors that promote a positive school climate: 
Each week, orange wristbands were distributed to students who were observed engaging in positive, conflict‑reducing behavior. All students in the schools were eligible to receive a wristband as recognition for the conflict-mitigating behaviors. 
If a student is reported wearing an orange wristband,
the response variable $Y$ is set to $1$; otherwise, it is $0$. We fit the subspace logistic regression described in Section~\ref{sec:methodology} to analyze this response.

To identify features strongly associated with the allocation of orange wristbands, we use individual attributes from the supplementary materials of \cite{paluck2016changing} as potential predictors. These include Treatment (participation in weekly training: Yes/No), Gender (Male/Female), Race (White/Hispanic/Black/Asian/Others), Grade, ``Friends-like-house" (friends say I have a nice house: Yes/No), and Home-language (speaks another language at home: Yes/No). Additionally, we incorporate GPA (grade point average on a 4.0 scale) and a binary covariate, Influencer (nominated by the teacher as influential). 
An individual school effect parameter is introduced for each school to account for school-level differences. We remove students with missing values and then use the largest connected component from each school. 
The final dataset contains 8,685 students from 25 schools. Each network has an average size of 347.1 students and an average degree of 10.7. Among all students, 1,391 received an orange wristband.

Based on our evaluation, embedding methods mentioned before, such as Node2Vec, do not improve predictive performance (see Section~\ref{sec:embed}) for the current dataset. Therefore, we will directly use the observed networks for better interpretability.

\subsection{Model fitting and interpretations}

We use the average of two friendship adjacency matrices from two survey waves (at the start and end of the school year) as our denoised $\hat{P}$ to measure the relations between students. It turns out that in this example using either matrix alone yields similar analyses, thanks to the robustness of our framework to network perturbations (see Section~\ref{appendix:robustness}). This robustness is a crucial advantage of our framework.

\begin{figure}[h]
\vspace{-0.5cm}
\includegraphics[width=1\linewidth]{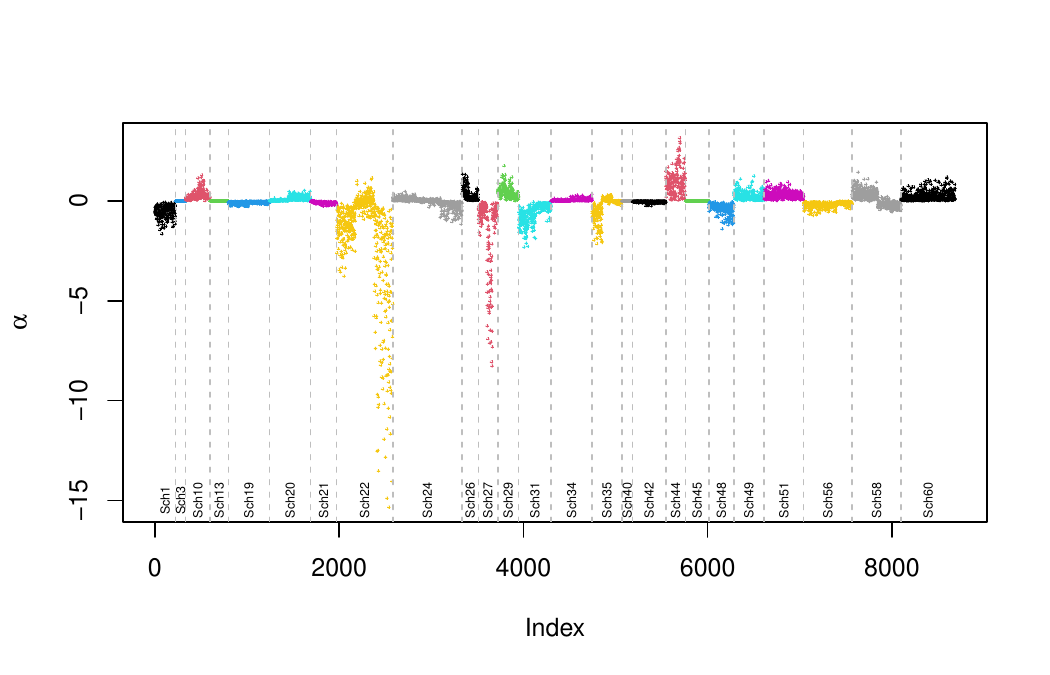}
\vspace{-1cm}
    \caption{Fitted $\hat{\alpha}$ from our model} 
    \label{Alpha}
\end{figure}

Our model incorporates the top 31 eigenvectors of the adjacency matrix to capture network effects, selected by \cite{chatterjee2015matrix}.
The $\chi^2$ test for network effects gives a very small $p$-value  ($< 10^{-8}$), indicating a significant contribution of the network information. The estimated $\hat{\alpha}_i$'s are shown in Figure \ref{Alpha}. The corresponding estimated coefficients of covariates are included in Section~\ref{sec:more-data}. 
The estimated value of 
$r$ is $0$, suggesting that there is no overlap between the covariate and the network structure.

Based on Figure~\ref{Alpha}, network effects vary significantly in magnitude between different schools. A few of the schools exhibit strong social influences, while many other schools exhibit minor social network effects. Aggregating all schools together  would dilute the significance of social effects. To gain a more comprehensive understanding of these dynamics, we introduce another layer of analysis on the schools with the most pronounced network effects, allowing us to explore the underlying factors driving these stronger social influences.

We define a school-specific network-effect-strength $t_j := \sum_{i\in O_j}|\hat{\alpha}_i|/\sum_{i\in O_j}|x_i^{\top}\hat{\beta}|$, where $O_j$ is the index set of students in the $j$th school. We select the five schools with the largest $t_j$ values (School ID 1, 22, 27, 31, and 48 in the dataset) for further analysis. We then apply our subspace logistic model, standard logistic regression, and the RNC logistic regression to the data. 
For our model and the standard logistic regression, important predictors are selected using backward elimination, whereby variables (including school fixed effect) with the largest $p$-values exceeding 0.05 after Bonferroni correction are removed sequentially until no further elimination is possible.
As the RNC model lacks an inference framework, we retain all variables in that model. The results of the three fitted models and their corresponding $p$-values, before and after backward elimination, are summarized in Table \ref{ModelFitting_alter_before} and \ref{ModelFitting_alter}, respectively.\footnote{Note that the $p$-values do not account for the selection of the five schools based on data and the backward elimination of variables. These results are used primarily for qualitative interpretations. In Section~\ref{secsec:validation}, we validate the models more rigorously by their prediction performance, accounting for the variable selection procedure.}

In the selected dataset of five schools, our model picks $K=13$ while $r$ is estimated to be 0. Again, the $\chi^2$ test gives a very small $p$-value, providing strong evidence of social effects. 
To better interpret the estimated network effect $\alpha$, we compute the correlation between $\left|\alpha\right|$ and a binary indicator of seed-eligible students, identified using the algorithm described in the supplement of \cite{paluck2016changing}, along with four network centrality metrics: degree centrality, betweenness centrality, eigenvector centrality, closeness centrality. The results are presented in Table \ref{tab:sub_corr_table}.

It can be observed that $\hat{\alpha}$ has a moderate correlation with degree centrality and eigenvector centrality. But it is only weakly correlated with the other centrality metrics. It is also marginally correlated with the seed eligibility of students. This observation indicates that the network effects capture signals that cannot be primarily explained by these commonly used node-level statistics. Another observation is that the correlation values of $\left|\hat{\alpha}\right|$ across different centrality measures are higher for the selected schools than in the full dataset, indicating that inference on these selected schools is more effective at identifying network effects.

The estimated treatment coefficient is similar across the three models. This might be expected due to the random assignment implemented by the experimenters, making this variable uncorrelated with other effects. However, unlike the RNC, our method and the standard logistic regression can use their $p$-values to show that the treatment effect is indeed significant.

\begin{table}[ht]
     \caption{Model fitting and inference results (before variable selection through backward elimination) on the five schools with the strongest network effects.}
     \centering
{
\begin{tabular}{rllllll}
\hline \hline 
& \multicolumn{2}{c}{ Our Model } & \multicolumn{2}{c}{ Logistic Reg } & \multicolumn{2}{c}{ RNC } \\
& coef. & $p$-value & coef. & $p$-value & coef. & $p$-value \\
\hline 
Treatment  &0.636& 0.005 & 0.678 & 0.004 &0.684 &  \\
Gender: Male & -0.454 & 0.002 &-0.440 & 0.003 &-0.521& \\
Grade & 0.131 & 0.141 & -0.119 & 0.126  &-0.253 & \\
Friends like house & 0.011 & 0.472 & -0.066 & 0.650 &-0.077 & \\
Home language& 0.203 & 0.137 & 0.271 & 0.134 & 0.212& \\
GPA& 0.106 & 0.245 & 0.054 & 0.721 & -0.230& \\
Influencer & 0.393 & 0.047 & 0.366 & 0.092 &0.504 &  \\
Race: White & -0.369 & 0.058 & -0.296 & 0.192  &-0.488 & \\
Race: Black &-0.085 & 0.401 & -0.102 & 0.756 &-0.229 & \\
Race: Hispanic & -0.457 & 0.025 &-0.512 & 0.024 & -0.544& \\
Race: Asian& 0.250 &0.217&0.260 &0.404  & 0.158& \\
Network Effect&-- & $6.0\times 10^{-3}$ &-- &-- & --& --\\
\hline
School 22& -1.447 & 0.002 & -1.964 &$<10^{-3}$ & --& --\\
School 27 & 0.960 & 0.018 & -0.002 & 0.993 & --& --\\
School 31& -0.033 & 0.470 & -0.010 &0.965 & --& --\\
School 48 & -0.455 & 0.164 & 0.065 & 0.784 & --& --\\
\hline\hline
\end{tabular}}
    \label{ModelFitting_alter_before}
\end{table}

\begin{table}[ht]
     \caption{Model fitting and inference results (after variable selection through backward elimination)  on the five schools with strongest network effects.}
     \centering
{
\begin{tabular}{rllllll}
\hline \hline 
& \multicolumn{2}{c}{ Our Model } & \multicolumn{2}{c}{ Logistic Reg } & \multicolumn{2}{c}{ RNC } \\
& coef. & $p$-value & coef. & $p$-value & coef. & $p$-value \\
\hline 
Treatment  &0.644& $<10^{-3}$ & 0.629 & $<10^{-3}$ &0.684 &  \\
Gender: Male & -0.450 & $<10^{-3}$ &-0.568 & $<10^{-3}$ &-0.521& \\
Grade & & & -0.180 & 0.002  &-0.253 & \\
Friends like house &  &  &  &  &-0.077 & \\
Home language& & &  &  & 0.212& \\
GPA&  &  &  &  & -0.230& \\
Influencer & &  &  &  &0.504 &  \\
Race: White &  &  & -0.701 & $<10^{-3}$  &-0.488 & \\
Race: Black &  & &  & &-0.229 & \\
Race: Hispanic &  &  &-0.724 & $<10^{-3}$ & -0.544& \\
Race: Asian& & &  & & 0.158& \\
Network Effect&-- & 8.5$\times 10^{-3}$ &-- &-- & --& --\\
\hline
School 22& -1.281 & $<10^{-3}$ & -2.101 &$<10^{-3}$ & --& --\\
School 27 & 1.003 & 0.002 &  & & --& --\\
\hline\hline
\end{tabular}}
    \label{ModelFitting_alter}
\end{table}

\begin{table}[ht]
\centering
\caption{Correlation between $\left|\alpha\right|$ and degree, betweenness, eigenvector, and closeness centrality, and seed eligibility, for the full dataset and for schools with strong network effects.}
\begin{tabular}{rrrrrrr}
  \hline
 & Degree Cen & Betweenness Cen &  Eigenvector Cen & Closeness Cen &  Seed Eligibility \\ 
  \hline
  Full dataset  & 0.258 & 0.112 & -0.043 & -0.025 & -0.003 \\ 
Selected dataset  & 0.300 & 0.199 & 0.477 & -0.189 & -0.031 \\ 
   \hline
\end{tabular}
\label{tab:sub_corr_table}
\end{table}

Since the RNC does not provide inference or variable selection, we focus on comparing our method with standard logistic regression. The two models yield very different inferences for the effects of Gender, Grade, and Race. Notably, Gender is the only predictor besides Treatment that remains in the final selection based on our model. The standard logistic regression estimates a 25\% stronger gender effect and finds statistically significant negative effects for Grade and Race. The main difference between our model and standard logistic regression is the inclusion of network effects, suggesting that the differential predictors may be cohesive according to network structures. This phenomenon is intuitively reasonable. For example, students are more likely to be friends with others in the same grade. We can empirically verify these conjectures. Figure \ref{fig:school_plots-school1} shows the gender, grade, and race information in one of the five schools: students tend to befriend others of the same gender, grade, and race. Similar patterns can be observed in other schools (see Figure 7 in Section~\ref{sec:more-data}l). Therefore, these predictors exhibit network cohesion, explaining the differential results between our method and standard logistic regression.

To further support the statement 
``students tend to befriend others of the same gender, grade, and race" quantitatively, we include the following summary table. It reports, for each attribute, the proportion of same-attribute friendships, the expected proportion under random mixing, and the assortativity coefficient. 
The assortativity coefficient proposed in \cite{newman2003mixing} is defined as the Pearson correlation between the attribute values:
\begin{equation*}
r = \frac{\sum_{i} e_{ii} - \sum_{i} a_i b_i}{1 - \sum_{i} a_i b_i},
\end{equation*}
where $e_{ij}$ denotes the fraction of edges connecting a node of type $i$ to a node of type $j$, and $a_i = \sum_j e_{ij}$, $b_j = \sum_i e_{ij}$ represent the fraction of edges attached to nodes of type $i$ and $j$, respectively.  
In the case of undirected networks, $e_{ij} = e_{ji}$ and $a_i = b_i$.

An assortativity of 
$r=1$ corresponds to perfect homophily, 
$r=0$ to random mixing, and $r<0$ to disassortative mixing.
We estimated its standard error using a jackknife procedure, by sequentially removing each edge, recalculating the assortativity $r_i$, and computing the variance as
$$s^2_r:= \sum_{i=1}^M \left(r_i-r\right)^2,$$
where $M$ is the number of edges and $r$ is the observed assortativity.
The standard error $s_r$ is the square root of this sum.
The results are summarized below:

\begin{table}[ht]
\centering
\caption{Observed and expected proportions of same-attribute friendships and assortativity coefficients based on Gender, Grade, and Race in School 1’s Friendship Network}
\begin{tabular}{r|cc|cc}
\hline \hline
& \multicolumn{2}{c|}{Same-Attribute Edges Proportion} & \multicolumn{2}{c}{Assortativity Statistics}  \\
 & Observed  & Expected  &  Assortativity  & Std. Error\\
\hline
Gender & 76.3\% & 50.1\% & 0.524 & 0.021\\
Grade  & 88.5\% & 33.5\% & 0.827 & 0.012\\
Race   & 49.0\% & 42.3\% & 0.122 & 0.019\\
\hline \hline
\end{tabular}
\label{tab:school1_assort}
\end{table}

All three covariates exhibit significant assortative mixing, consistent with the discussion in \cite{newman2003mixing}, with the effect being especially strong for gender and grade. This confirms that these covariates are highly correlated with the network structure, which helps explain the different coefficient magnitudes observed between our model and the standard logistic regression. Although race has the lowest assortativity among the three, the effect is still highly statistically significant (with a value more than six standard deviations from zero). This finding aligns with our observation that our model and the standard logistic regression differed in their variable selection for race only for the selected schools, but not in the full dataset.

\begin{figure}[ht]
    \centering
    \begin{subfigure}[b]{0.30\linewidth}
        \centering
        {\includegraphics[width=\linewidth]{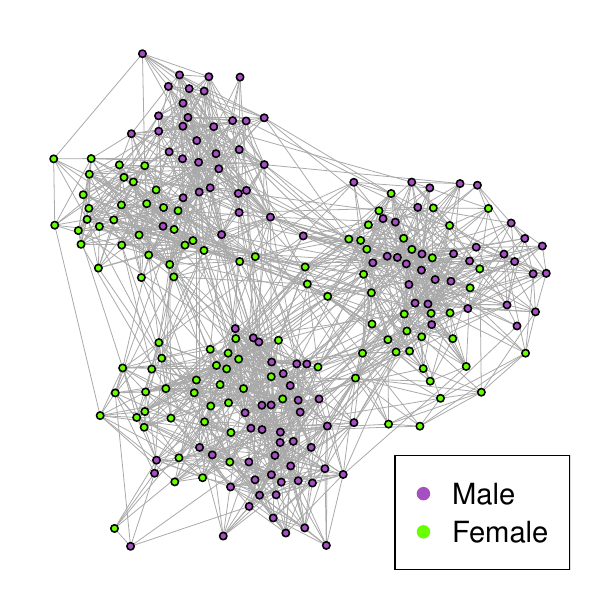}}
        \subcaption{School 1: Gender}
        \label{School 1: Gender}
    \end{subfigure}
    \hfill
    \begin{subfigure}[b]{0.30\linewidth}
        \centering
        {\includegraphics[width=\linewidth]{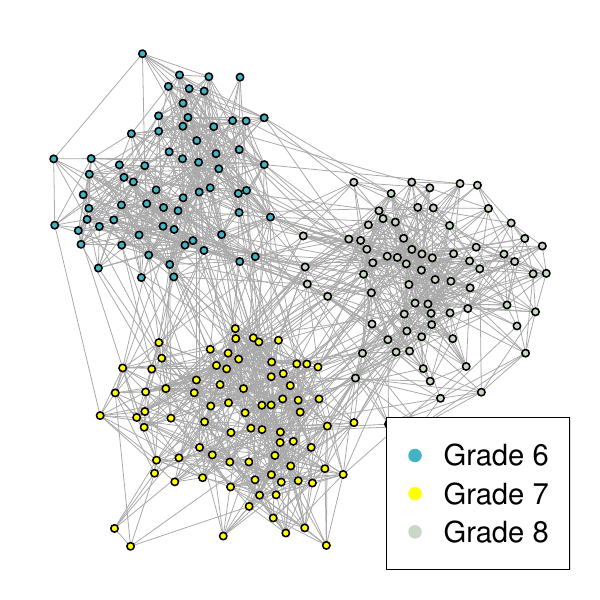}} 
        \subcaption{School 1: Grade}
        \label{School 1: Grade}
    \end{subfigure}
    \hfill
    \begin{subfigure}[b]{0.30\linewidth}
        \centering
        {\includegraphics[width=\linewidth]{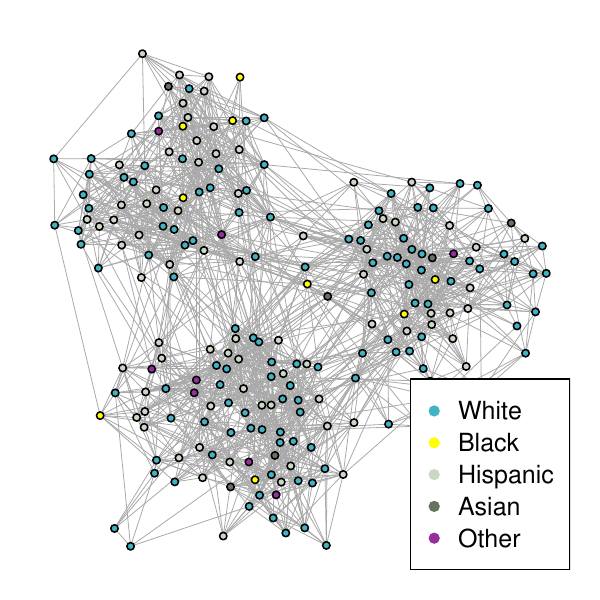}} 
        \subcaption{School 1: Race}
        \label{School 1: Race}
    \end{subfigure}
    \caption{Friendship network of School 1, along with the corresponding gender, grade, and race information.}
    \label{fig:school_plots-school1}
\end{figure}

The positive effect of the workshop observed in our analysis is consistent with the main conclusion of the original study by \citet{paluck2016changing}. However, we emphasize that this agreement is only at a high level. Our analysis differs from that of \citet{paluck2016changing} in several important respects. First, \citet{paluck2016changing} focus on a subpopulation of students who are connected to at least one potentially treated peer, whereas our analysis does not impose such a restriction and includes all available observations from treated schools. Second, their analysis is explicitly causal in nature, relying on the original randomized design and causal inference methods. In contrast, our results are descriptive and inferential but do not carry a causal interpretation.

In summary, we have shown that social network information has important impacts in the current problem. Though both the RNC and our model can incorporate network information in building the logistic regression model, the available inference framework in our model provides a substantial advantage in understanding the data with more conclusive insights: both the social effect and the conflict-mitigating training are statistically significant in this example. Compared with the standard logistic regression, all the qualitative differences in estimated effects can be explained by the network cohesion phenomenon, which can be empirically verified.

All previous discussions focus on model interpretation and we have seen that differences between our model and the standard logistic regression are reasonable. Next, we use prediction performance to validate the effectiveness of our model compared to the standard logistic regression.

\subsection{Predictive Model Validation}\label{secsec:validation}

We use out-of-sample prediction performance to validate the practical significance of the network effects. Consider the scenario where the response is only partially observed. It is then useful to assess the performance of the models when they make predictions on the unobserved response variable based on the full set of covariates and the network. In particular, we use 200-fold cross-validation to assess the performance: all the students are partitioned into 200 folds randomly. We hold out one fold of the response variable and make predictions based on the fitted model from the 199 folds (with all the needed tuning). This procedure is repeated for each of the 200 folds. Since the current task is a binary classification problem, we use the ROC curves and the area under the curve (AUC) of the predicted probabilities (aggregated over the 200 iterations) as the performance metric.

\begin{figure}[h]
    \centering
    \begin{subfigure}[b]{0.48\linewidth}
        \centering
        \vspace{-0.4cm}
        \includegraphics[width=1\linewidth]{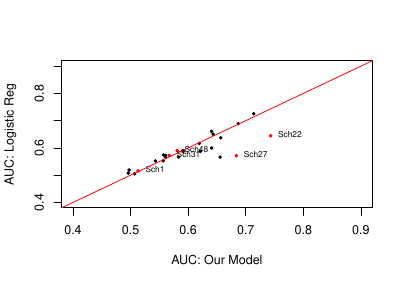}
        \subcaption{Our Model vs. Logistic Reg}
    \end{subfigure}
    \hfill
    \begin{subfigure}[b]{0.48\linewidth}
        \centering
        \vspace{-0.4cm}
        \includegraphics[width=1\linewidth]{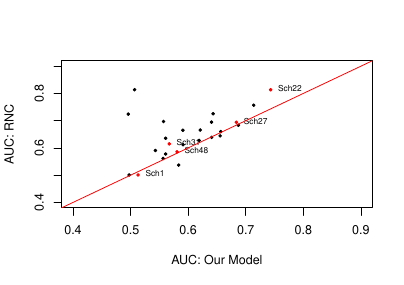} 
        \subcaption{Our Model vs. RNC}
    \end{subfigure}
    \caption{Prediction performance comparison for each school between our method with logistic regression and RNC.}
    \label{fig:overall-AUC}
\end{figure}

Figure~\ref{fig:overall-AUC} shows the AUC values calculated based on predictions in each individual school by our model and the two benchmarks. The five selected schools in the previous analysis are colored red.
The results show that our model consistently outperforms standard logistic regression and RNC, especially in the five schools with strong network effects. This indicates that our model effectively exploits the network information to provide more accurate predictions. The result also shows that overall, the social network effects are sufficiently influential to exhibit differential prediction accuracy.

Since our prediction model interpretations in Table \ref{ModelFitting_alter} are based on selected schools, we also want to evaluate the effects of this selection procedure. Therefore, we apply the aforementioned 200-fold cross-validation but include the school selection procedure: In each iteration, we first select five schools based on the 199 folds and then focus on model fitting and predicting the hold-out fold constrained within the selected five schools. Note that different schools may be selected for each iteration in this procedure. Thus, this evaluation also includes the randomness of the selection. The ROC curves aggregated over the 200 folds are shown in Figure~\ref{fig:selected-ROC}. The conclusion from Figure~\ref{fig:selected-ROC} is consistent with Figure~\ref{fig:overall-AUC}.

\begin{figure}[ht]   
\centering
\includegraphics[width=0.8\linewidth]{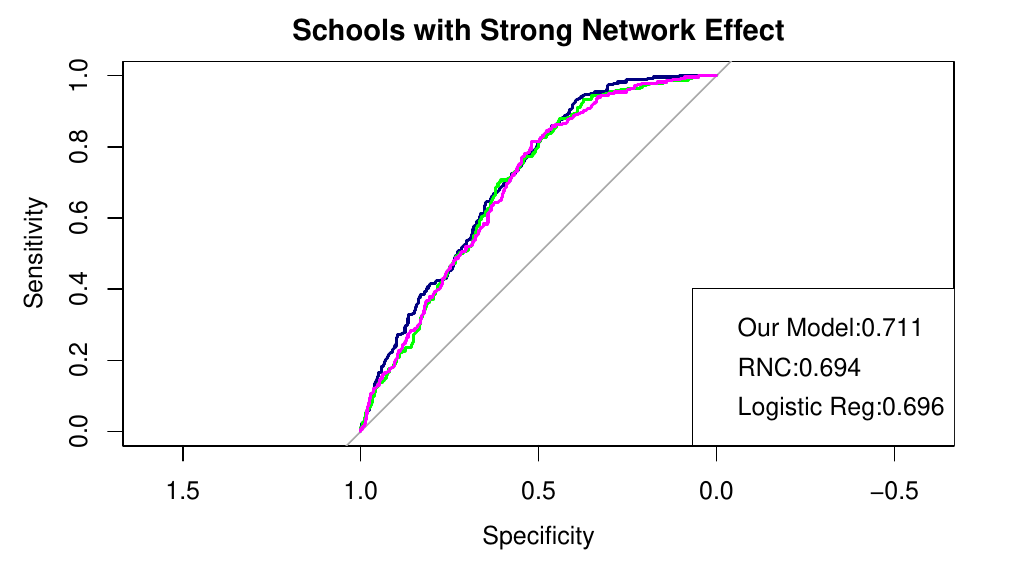}
 \caption{ROC curves of three methods restricted to selected schools from the 200-fold cross-validation procedure.}
 \label{fig:selected-ROC}
\end{figure}

In summary, our validation experiments show that, whether we consider the selected subset of schools or all of them, the social network effects are strong, and ignoring them results in inferior prediction performance. Our model provides the best predictive power among the three models. 
Compared to RNC, the proposed framework offers interpretability, valid inference, and accurate predictions, with provable robustness to network perturbations.

\section{Discussion}
We have introduced a class of generalized linear models linked by network subspace assumptions. The advantage of this framework lies in its flexibility due to the nonparametric network effects, the availability of a statistical inference framework, and its proven robustness to network structure perturbations. We have empirically verified that the inference is valid for network perturbations from random network models and algorithmic perturbations from network embedding methods.

Several interesting directions for expanding our study remain. One particularly intriguing problem is incorporating more general graph neural networks \citep{scarselli2008graph,kipf2016semi} into similar subspace models for network-linked data and extending the inference framework to such situations. In a related direction, conformal predictions have been studied for network regression problems \citep{lunde2023conformal}, but adapting a formal inference framework to handle these additional complications would be both more widely useful and more challenging. Finally, a fundamental problem is using such subspace models to handle spill-over effects of randomized experiments on social networks or even more general causal analysis with network effects \citep{sinclair2012detecting,phan2015natural,lee2021network,hayes2022estimating}, or even other causal analysis involving network-mediated effects. Formulating a spill-over or mediation causal model in the subspace format would be a crucial step for generalizing the proposed framework for such analyses.


\bibliographystyle{abbrvnat}
\bibliography{reference}      

\newpage
\appendix

\section{Additional results of the school conflict analysis}\label{sec:more-data}

\subsection{Finalized model on all 25 schools}

The estimated parameters and $p$-values before and after variable selection by backward elimination using the averaged network of two waves in all schools are summarized in Tables \ref{ModelFitting_before} and \ref{ModelFitting}. Under our model, $r$ is detected to be $0$, and the $\chi^2$ test for the existence of the network effect yields a $p$-value $< 10^{-8}$, suggesting the statistical significance of the network information. 

\newpage
\begin{table}[H]
     \caption{Estimated coefficients and $p$-values (before variable selection through backward elimination) of our model, standard logistic regression and RNC using the average network of two waves involving all schools. }
     \centering{
     \begin{tabular}{rllllll}
\hline \hline 
& \multicolumn{2}{c}{ Our Model } & \multicolumn{2}{c}{ Logistic Reg } & \multicolumn{2}{c}{ RNC } \\
& coef. & $p$-value & coef. & $p$-value & coef. & $p$-value \\
\hline 
Gender: Male & -0.417 & $<10^{-3}$ &-0.405 & $<10^{-3}$ &-0.497& --\\
Grade & -0.270 & $<10^{-3}$ & -0.271 & $<10^{-3}$ &-0.359 & --\\
Friends like house & 0.100 & 0.059 & 0.102 & 0.106 &0.051 & --\\
Home language& 0.300 & $<10^{-3}$ & 0.305 & $<10^{-3}$ & 0.241& --\\
Treatment &0.846& $<10^{-3}$ & 0.828 & $<10^{-3}$ &0.800 & -- \\
GPA& 0.085 & 0.066 & 0.067 & 0.228 & -0.228& --\\
Influencer & 0.187 & 0.054 & 0.234 & 0.036 &0.311 & --\\
Race: White & -0.079 & 0.218 & -0.082 & 0.416 &-0.342 & --\\
Race: Black & 0.087  & 0.241 & 0.045 & 0.715 &-0.211 & --\\
Race: Hispanic & -0.093  & 0.174 & -0.113 & 0.254 & -0.209 & --\\
Race: Asian& 0.086 & 0.293 & 0.074 & 0.637 & -0.048& --\\
 \hline
 School 3 & 0.668 & 0.040 & 0.916 & $<10^{-3}$ & --& --\\
 School 10 & -0.345  & 0.195 & 0.157 & 0.481 & --& --\\
 School 13 & -0.558 & 0.068 & -0.310 & 0.214 & --& --\\
 School 19 & -1.326 & 0.002 & -1.140 & $<10^{-3}$ & --& --\\
 School 20 & -1.874 & $<10^{-3}$ & -1.566 & $<10^{-3}$ & --& --\\
 School 21 & -1.343 & 0.005 & -1.209 & $<10^{-3}$ & --& --\\
 School 22 & -1.479 & 0.001 & -1.893 & $<10^{-3}$ & --& --\\
 School 24 & -1.130 & 0.002 & -0.888 & $<10^{-3}$ & --& --\\
 School 26 & -1.020 & 0.013 & -0.392 &   0.137& --& --\\
 School 29 & -1.467 & 0.002 & -0.721 & 0.006 & --& --\\
 School 31 & -0.278 & 0.259 & -0.048 & 0.819 & --& --\\
 School 34 & -0.611 & 0.052 & -0.345 & 0.105 & --& --\\
 School 35 & -0.162 & 0.340 & 0.212 & 0.315 & --& --\\
 School 40 & 0.805 & 0.017 & 1.047 & $<10^{-3}$ & --& --\\
 School 42 & -0.412 & 0.151 & -0.243 & 0.273  & --& --\\
 School 44 &-1.376 & 0.002 & -0.138 & 0.564 & --& --\\
 School 45 & -1.770 & $<10^{-3}$ & -1.497 & $<10^{-3}$ & --& --\\
 School 48 & -0.098 & 0.406 & 0.112 & 0.617 & --& --\\
 School 49 & -0.354 & 0.181 & 0.193 & 0.365 & --& --\\
 School 51 & -0.294 & 0.233 & 0.273 & 0.178 & --& --\\
 School 56 & -0.237 & 0.272 & -0.187 & 0.375 & --& --\\
 School 58 & -1.163 & 0.002 & -0.770 & $<10^{-3}$ & --& --\\
 School 60 & -1.175 & 0.001 & -0.716 & 0.001 & --& --\\
\hline\hline
\end{tabular}}
    \label{ModelFitting_before}
\end{table}

\newpage
\begin{table}[H]
     \caption{Estimated coefficients and $p$-values (after variable selection through backward elimination) of our model, standard logistic regression and RNC  using the average network of two waves involving all schools. }
     \centering{
     \begin{tabular}{rllllll}
\hline \hline 
& \multicolumn{2}{c}{ Our Model } & \multicolumn{2}{c}{ Logistic Reg } & \multicolumn{2}{c}{ RNC } \\
& coef. & $p$-value & coef. & $p$-value & coef. & $p$-value \\
\hline 
Gender: Male & -0.443 & $<10^{-3}$ &-0.419 & $<10^{-3}$ &-0.497& --\\
Grade & -0.323 & $<10^{-3}$ & -0.256 & $<10^{-3}$ &-0.359 & --\\
Friends like house &  &  &  &  &0.051 & --\\
Home language& 0.272 & $<10^{-3}$ & 0.324 & $<10^{-3}$ & 0.241& --\\
Treatment &0.837& $<10^{-3}$ & 0.820 & $<10^{-3}$ &0.800 & -- \\
GPA&  &  &  &  & -0.228& --\\
Influencer& &  &  &  &0.311 &  --\\
Race: White &  &  &  & &-0.342 & --\\
Race: Black &  & &  & &-0.211 & --\\
Race: Hispanic &  &  & &  & -0.209& --\\
Race: Asian& & &  & & -0.048& --\\
 \hline
 School 3 & & & 0.969 & $<10^{-3}$ & --& --\\
 School 10 & -0.803  & $<10^{-3}$ & &  & --& --\\
 School 13 & -0.977 & $<10^{-3}$ &  & & --& --\\
 School 19 & -1.709 & $<10^{-3}$ & -1.122& $<10^{-3}$ & --& --\\
 School 20 & -2.363 & $<10^{-3}$ & -1.487 & $<10^{-3}$ & --& --\\
 School 21 & -1.815 &$<10^{-3}$ & -1.248 & 0.001 & --& --\\
 School 22 & -1.885 & $<10^{-3}$ & -1.850 & $<10^{-3}$ & --& --\\
 School 24 & -1.557 & $<10^{-3}$ & -0.852 & $<10^{-3}$ & --& --\\
 School 26 & -1.375 & $<10^{-3}$ & &  & --& --\\
 School 29 & -1.869 & $<10^{-3}$ & -0.702 & 0.001 & --& --\\
 School 34 & -1.145 & $<10^{-3}$ & -0.306 & 0.024 & --& --\\
 School 40 &  &  & 1.009 & $<10^{-3}$ & --& --\\
 School 42 & -0.845 & $<10^{-3}$ &  &  & --& --\\
 School 44 &-1.864 & $<10^{-3}$ & &  & --& --\\
 School 45 & -2.080 & $<10^{-3}$ & -1.457 & $<10^{-3}$ & --& --\\
 School 49 & -0.777 & $<10^{-3}$ & &  & --& --\\
 School 51 & -0.724 & $<10^{-3}$ & 0.275 & 0.022 & --& --\\
 School 56 & -0.619 & 0.002 & &  & --& --\\
 School 58 & -1.580 &$<10^{-3}$ & -0.732 & $<10^{-3}$ & --& --\\
 School 60 & -1.511 & $<10^{-3}$& -0.612 & $<10^{-3}$ & --& --\\
\hline\hline
\end{tabular}}
    \label{ModelFitting}
\end{table}

The RNC columns for school fixed effects are blank because fixed effects are not identifiable under RNC’s penalty.
Additionally, we observe that many school effect terms differ significantly between the standard logistic regression and our model. This suggests that the network effect captured by our model better explains the heterogeneity within schools.

\subsection{Robustness validation with three network constructions}\label{appendix:robustness}

\begin{table}[H]
    \caption{Estimated coefficients and $p$-values of our model using three versions of the friendship network (Wave I, Wave II, and Wave I-II average). The blanks indicate that the variables are removed in the backward elimination procedure.}
    \centering
    \begin{tabular}{rllllll}
        \hline \hline
        & \multicolumn{2}{c}{ Wave Average } & \multicolumn{2}{c}{ Wave I } & \multicolumn{2}{c}{ Wave II } \\
        & coef. & $p$-value & coef. & $p$-value & coef. & $p$-value \\
        \hline
        Gender: Male & -0.443 & $<10^{-3}$ & -0.429 & $<10^{-3}$ & -0.427 & $<10^{-3}$ \\
        Grade & -0.323 & $<10^{-3}$ & -0.252 & $<10^{-3}$ & -0.257 & $<10^{-3}$ \\
        Friends like house &  &  &  &  &  &  \\
        Home language & 0.272 & $<10^{-3}$ & 0.316 & $<10^{-3}$ & 0.366 & $<10^{-3}$ \\
        Treatment & 0.837 & $<10^{-3}$ & 0.842 & $<10^{-3}$ & 0.851 & $<10^{-3}$ \\
        GPA &  &  &  &  &  &  \\
        Influencer &  &  &  &  &  &  \\
        Race: White &  &  &  &  &  &  \\
        Race: Black &  &  &  &  &  &  \\
        Race: Hispanic &  &  &  &  &  &  \\
        Race: Asian &  &  &  &  &  &  \\
         \hline
 School 3 & & & &  & 1.258  & $<10^{-3}$  \\
 School 10 & -0.803  & $<10^{-3}$ &-0.858 & $<10^{-3}$ &  &  \\
 School 13 & -0.977 & $<10^{-3}$ & -1.076 &$<10^{-3}$ &  &  \\
 School 19 & -1.709 & $<10^{-3}$ & -1.779 & $<10^{-3}$ &  &  \\
 School 20 & -2.363 & $<10^{-3}$ & -2.324 & $<10^{-3}$ & -1.283 &  $<10^{-3}$ \\
 School 21 & -1.815 &$<10^{-3}$ & -2.033 & $<10^{-3}$ & & \\
 School 22 & -1.885 & $<10^{-3}$ & -2.331 & $<10^{-3}$ & & \\
 School 24 & -1.557 & $<10^{-3}$ & -1.682 & $<10^{-3}$ & -0.616 & 0.002 \\
 School 26 & -1.375 & $<10^{-3}$ & -1.429 & $<10^{-3}$ &  &  \\
 School 27 &  &  & -0.766 & $<10^{-3}$ & 0.882  &  $<10^{-3}$ \\
 School 29 & -1.869 & $<10^{-3}$ & -1.543 & $<10^{-3}$ & & \\
  School 31 &  &  & -0.748 & $<10^{-3}$ & & \\
 School 34 & -1.145 & $<10^{-3}$ & -1.157 & $<10^{-3}$ & & \\
  School 35 &  &  & -0.909 & $<10^{-3}$ & & \\
 School 40 &  &  & & & 1.281 &  $<10^{-3}$\\
 School 42 & -0.845 & $<10^{-3}$ & -0.936 & $<10^{-3}$ & & \\
 School 44 &-1.864 & $<10^{-3}$ & -0.968 & $<10^{-3}$ & & \\
 School 45 & -2.080 & $<10^{-3}$ & -2.304& $<10^{-3}$ & -1.165 &  $<10^{-3}$\\
 School 48 & & & -0.726 & $<10^{-3}$  & & \\
 School 49 & -0.777 & $<10^{-3}$ & -0.843 & $<10^{-3}$ & & \\
 School 51 & -0.724 & $<10^{-3}$ & 0.729 & $<10^{-3}$ & & \\
 School 56 & -0.619 & 0.002 & -0.936 &  $<10^{-3}$& & \\
 School 58 & -1.580 &$<10^{-3}$ & -1.627 & $<10^{-3}$ & & \\
 School 60 & -1.511 & $<10^{-3}$& -1.622 & $<10^{-3}$ & & \\
        \hline\hline
    \end{tabular}
    \label{MergedModelFitting}
\end{table}

\noindent To evaluate the robustness of our inference results to network perturbations, we fit our model separately using the Wave I and Wave II networks. In Table~\ref{MergedModelFitting}, we present the results of the same analysis with backward elimination, showing the estimated parameters and $p$-values for the models based on the three versions of the networks. Despite substantial edge‑level variation between waves (Fig.~\ref{fig:overlappingedges}), the selected variables are identical (up to minor numeric differences). These findings demonstrate the robustness of our framework. As noted in \cite{le2022linear}, this stability arises because, although the individual edges in the friendship networks experienced substantial changes, the overall spectral structure of the adjacency matrices remained stable.

\newpage

\subsection{Predictive comparisons with embedding-based methods}\label{appendix:embedding}

\noindent We have also evaluated the possibility of using deep-learning-based embedding methods to incorporate the network information, as discussed in  Section~\ref{secsec:embedding}. Figure~\ref{fig:embedding-AUC} shows the predictive AUC of the fitted models based on embedded similarity relations from DeepWalk, Node2Vec, and Diff2Vec, compared with the fitted model based on the observed network structure. The evaluation follows the same procedure described in Section~\ref{sec:data}. It can be seen that the relational data learned from the embedding methods do not lead to better predictive performance. This may indicate that the relevant relational information in the current problem is already reflected in the observed adjacency matrix, and the additional nonlinear transformations introduced by these embedding methods do not provide further benefits.

\begin{figure}[ht]
    \centering
    \begin{subfigure}[b]{0.325\linewidth}
        \centering
        \vspace{-0.4cm}
        \includegraphics[width=1\linewidth]{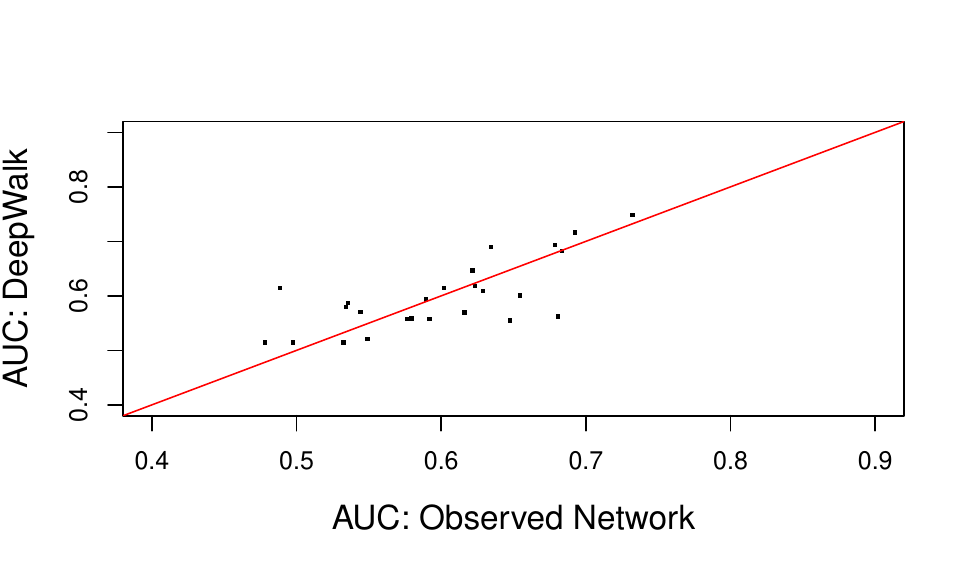}
        \subcaption{DeepWalk}
    \end{subfigure}
    \begin{subfigure}[b]{0.325\linewidth}
        \centering
        \vspace{-0.4cm}
        \includegraphics[width=1\linewidth]{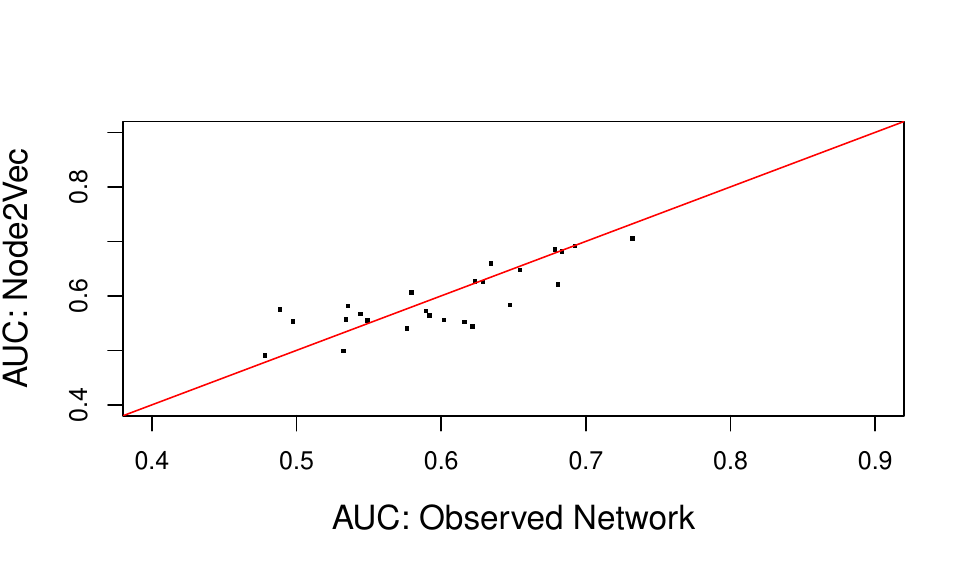} 
        \subcaption{Node2Vec}
    \end{subfigure}
    \begin{subfigure}[b]{0.325\linewidth}
        \centering
        \vspace{-0.4cm}
        \includegraphics[width=1\linewidth]{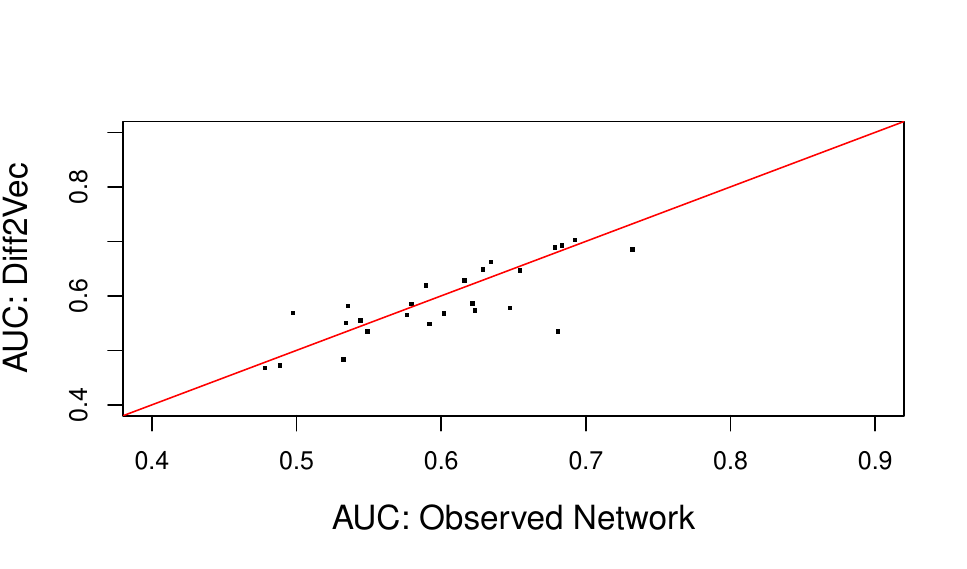} 
        \subcaption{Diff2Vec}
    \end{subfigure}
    \caption{Prediction performance comparison for each school between the observed network (Wave I + Wave II), and the embedding similarity of DeepWalk, Node2Vec and Diff2Vec.}
    \label{fig:embedding-AUC}
\end{figure}

\subsection{Covariate correlation with network structures in the refined analysis}

Figure~\ref{fig:school_plots-others} displays the gender, grade and race information for the other four selected schools in our refined analysis. These covariates display a cohesive pattern based on the network structure, which explains why inference can differ between our model and logistic regression without network information.

\newpage
\begin{figure}[H]
    \centering
    \begin{subfigure}[b]{0.25\linewidth}
        \centering
        {\includegraphics[width=\linewidth]{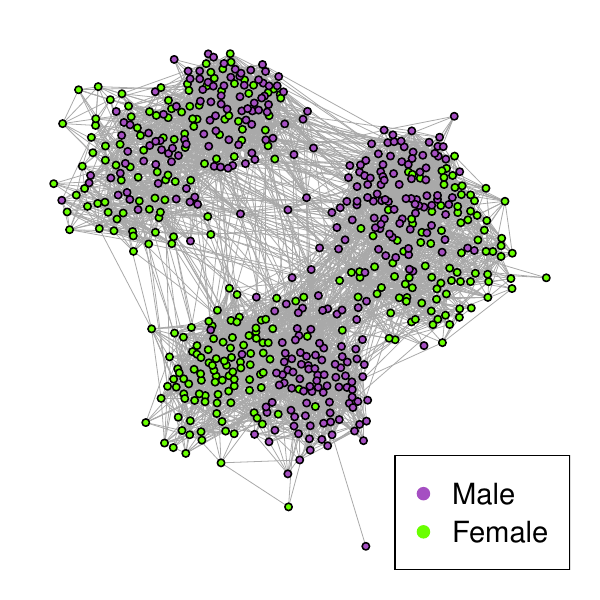}}     
        \subcaption{School 22: Gender}
        \label{School 22: Gender}
    \end{subfigure}
    \hfill
    \begin{subfigure}[b]{0.25\linewidth}
        \centering
        {\includegraphics[width=\linewidth]{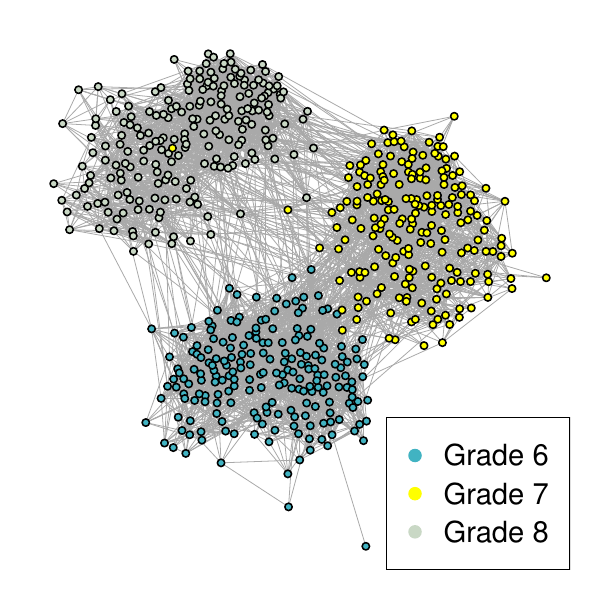}} 
        \label{School 22: Rrade}
        \subcaption{School 22: Grade}
    \end{subfigure}
    \hfill
    \begin{subfigure}[b]{0.25\linewidth}
        \centering
        {\includegraphics[width=\linewidth]{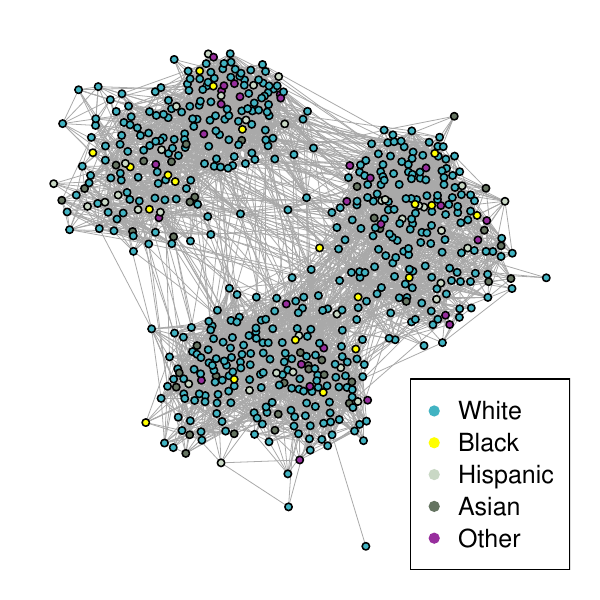}} 
        \label{School 22: Race}
        \subcaption{School 22: Race}
    \end{subfigure}
    \hfill
    \begin{subfigure}[b]{0.25\linewidth}
        \centering
        {\includegraphics[width=\linewidth]{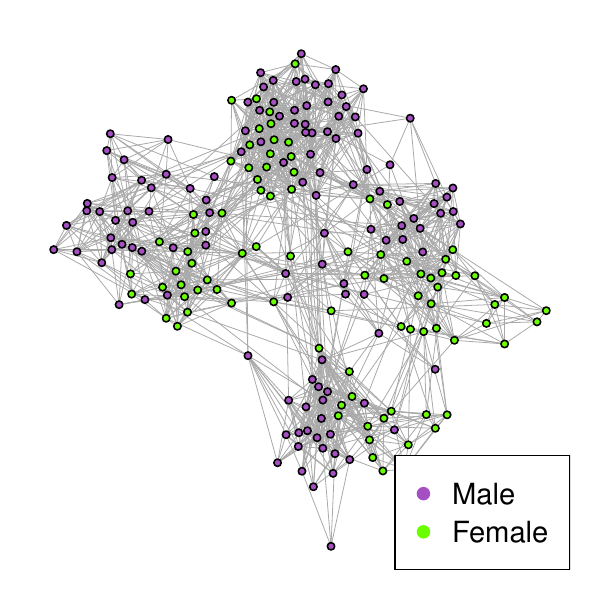}}
        \label{School 27: Gender}
        \subcaption{School 27: Gender}
    \end{subfigure}
    \hfill
    \begin{subfigure}[b]{0.25\linewidth}
        \centering
        {\includegraphics[width=\linewidth]{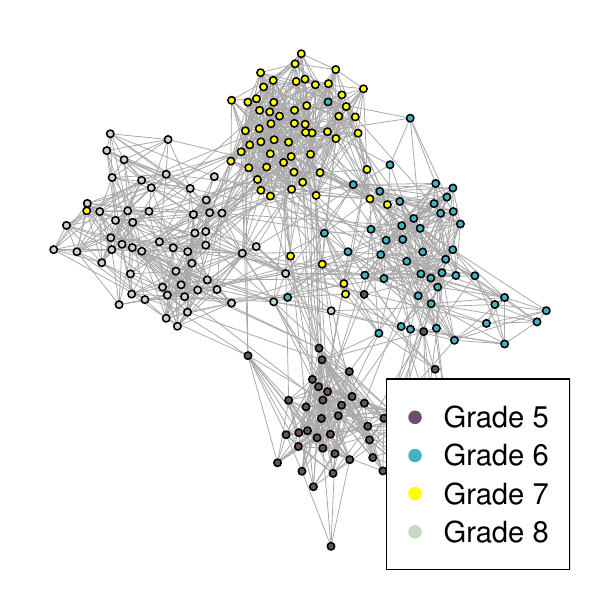}}
        \label{School 27: Grade}
        \subcaption{School 27: Grade} 
    \end{subfigure}
    \hfill
    \begin{subfigure}[b]{0.25\linewidth}
        \centering
        {\includegraphics[width=\linewidth]{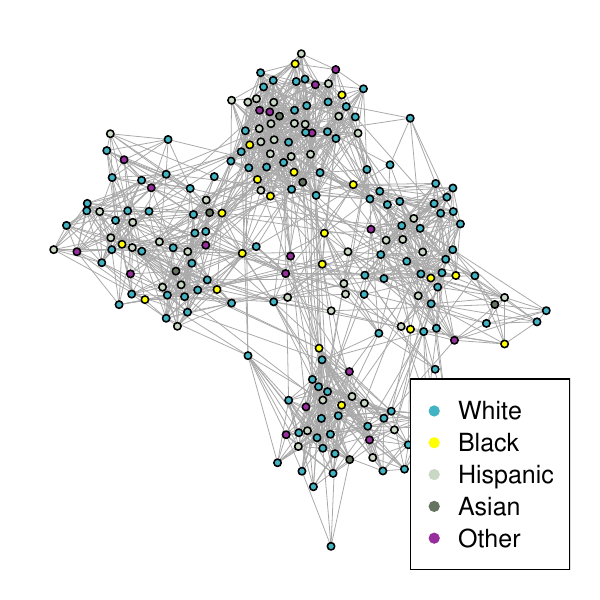}}
        \label{School 27: Race}
        \subcaption{School 27: Race}
    \end{subfigure}
    \hfill
    \begin{subfigure}[b]{0.25\linewidth}
        \centering
        {\includegraphics[width=\linewidth]{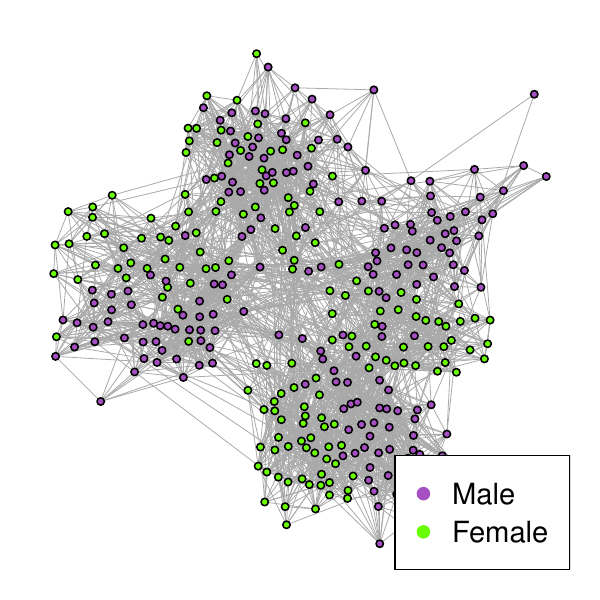}}
        \label{School 31: Gender}
        \subcaption{School 31: Gender}
    \end{subfigure}
    \hfill
    \begin{subfigure}[b]{0.25\linewidth}
        \centering
        {\includegraphics[width=\linewidth]{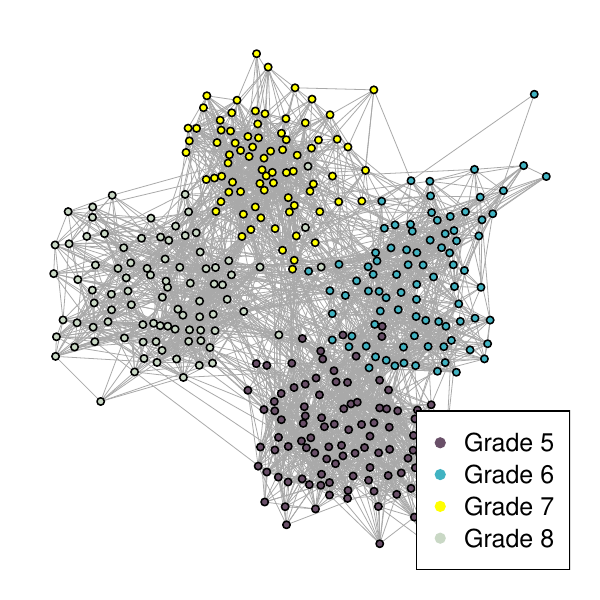}}
        \label{School 31: Grade}
        \subcaption{School 31: Grade}
    \end{subfigure}
    \hfill
    \begin{subfigure}[b]{0.25\linewidth}
        \centering
        {\includegraphics[width=\linewidth]{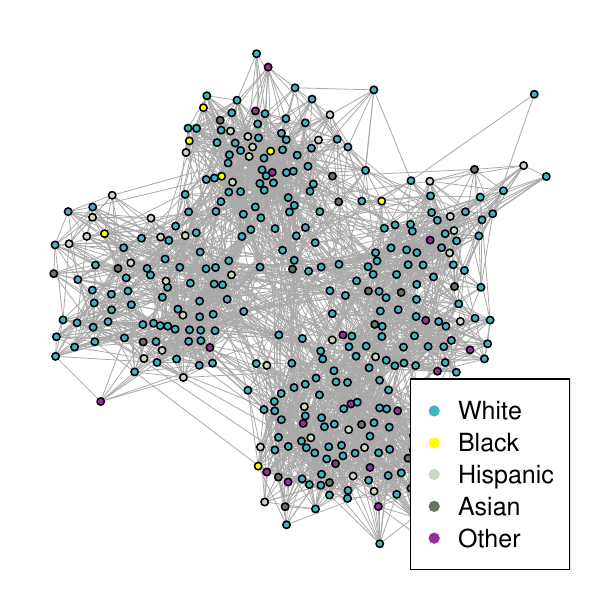}}
        \label{School 31: Race}
        \subcaption{School 31: Race}
    \end{subfigure}
    \hfill
    \begin{subfigure}[b]{0.25\linewidth}
        \centering        
        {\includegraphics[width=\linewidth]{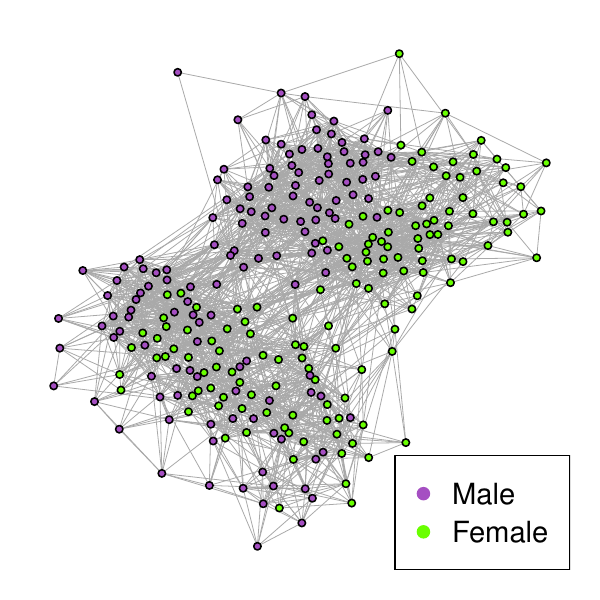}}
        \label{School 48: Gender}
        \subcaption{School 48: Gender}
    \end{subfigure}
    \hfill
    \begin{subfigure}[b]{0.25\linewidth}
        \centering
        {\includegraphics[width=\linewidth]{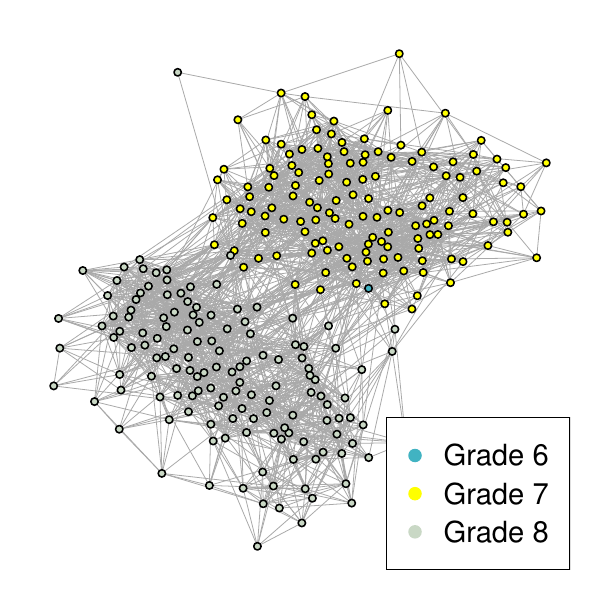}}  \label{School 48: Grade}
        \subcaption{School 48: Grade}
    \end{subfigure}
    \hfill
    \begin{subfigure}[b]{0.25\linewidth}
        \centering
        {\includegraphics[width=\linewidth]{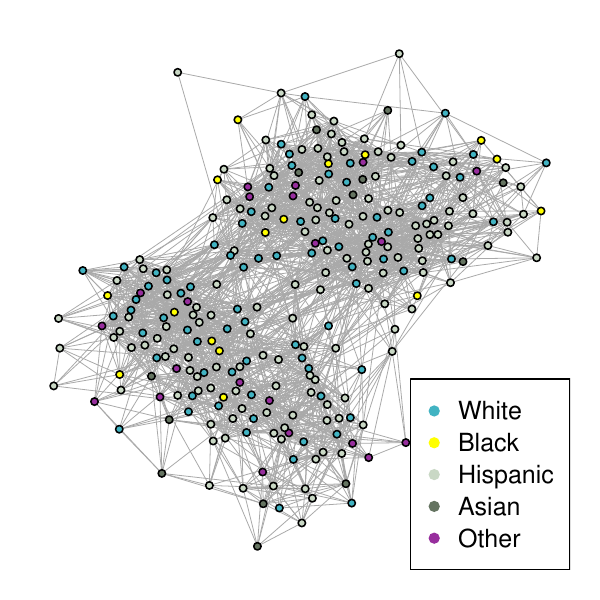}}
        \label{School 48: Race}
        \subcaption{School 48: Race}
    \end{subfigure}
    \hfill
    \caption{Friendship networks of four schools, along with the corresponding gender, grade, and race information.}
\label{fig:school_plots-others}
\end{figure}

\newpage
\section{Proofs for theoretical results}

Before proving our main results in subsequent appendices,  
let us gather here some important properties for the 
subspace estimators $\hat{\mathcal{R}}$, $\hat{\mathcal{C}}$, and $\hat{\mathcal{N}}$ in \eqref{subspaceEstimator}. Similar to \cite{le2022linear}, we will show that these estimators are sufficiently close to $\mathcal{R},\mathcal{C}$ and $\mathcal{N}$, respectively, which are defined in Section~\ref{sec:algorithm}. Instead of studying these subspaces directly, we will work with the orthogonal projections onto them. To this end, denote  
\begin{eqnarray*}
    \mathcal{P}_R = n^{-1} Z_{1:r} Z^{\top}_{1:r}, \quad
    \mathcal{P}_C = n^{-1} Z_{(r+1):p} Z^{\top}_{(r+1):p},
    \quad \mathcal{P}_N =n^{-1} W_{(r+1):K} W^{\top}_{(r+1):K}.
\end{eqnarray*}
Similarly, denote 
\begin{eqnarray*}
    \hat{\mathcal{P}}_R =n^{-1} \tilde{Z}_{1:r} \tilde{Z}^{\top}_{1:r}, \quad
    \hat{\mathcal{P}}_C = n^{-1} \tilde{Z}_{(r+1):p} \tilde{Z}^{\top}_{(r+1):p},
    \quad \hat{\mathcal{P}}_N =n^{-1} \tilde{W}_{(r+1):K} \tilde{W}^{\top}_{(r+1):K}.
\end{eqnarray*}
We first recall Corollary 5 in \cite{le2022linear}, which provides an error bound for the subspace estimation.  

\begin{proposition}[Subspace perturbation]
\label{proposition 1}
Assume that Assumption \ref{cond:A4} holds. There exists a constant $C_1>0$ such that,
\begin{equation*}
    \max \left\{\big\|\hat{\mathcal{P}}_R-\mathcal{P}_R\big\|,\big\|\hat{\mathcal{P}}_C-\mathcal{P}_C\big\|,\big\|\hat{\mathcal{P}}_N-\mathcal{P}_N\big\|\right\} \leq C_1 \tau_n,
\end{equation*}
where 
\begin{equation}
\label{tau_defin}
    \tau_n= \frac{n^{-3/2}\big\|(\tilde{W} \tilde{W}^{\top}-W W^{\top}) Z\big\|}{\min \left\{\left(1-\sigma_{r+1}\right)^3, \sigma_{r+s}^3\right\}},
\end{equation}
and the singular values $\sigma_{r+1}$ and $\sigma_{r+s}$ are defined in \eqref{defin:model}. 
\end{proposition}

Recall the ``new covariate'' vectors $g_i$ and $\tilde{g}_i$, defined in Section~\ref{sec:algorithm}, which combine the covariate and relational information for node $i$:
$$
g_i = (Z_{i,1:p} \ W_{i,(r+1):K})^\top, \quad \tilde{g}_i = (\tilde{Z}_{i,1:p} \ \tilde{W}_{i,(r+1):K})^\top.
$$
These vectors depend on the choices of bases for $\operatorname{col}(X)$, $S_K(P)$, and $S_K(\hat{P})$ through $Z$, $W$, $\tilde{Z}$ and $\tilde{W}$. Due to the nature of these choices, $g_i$ and $\tilde{g}_i$ can be approximately aligned through an almost rotation matrix defined by
\begin{equation}
\label{eq:Tn}
T_n:=\left(\begin{array}{ccc}\tilde{Z}_{1:r}^{\top}{Z}_{1:r}/n & 0 &0 \\
    0 & \tilde{Z}_{(r+1):p}^{\top}{Z}_{(r+1):p}/n & 0\\
    0& 0 & \tilde{W}_{(r+1):K}^{\top} {W}_{(r+1):K}/n\end{array}\right).
\end{equation}
Using Proposition~\ref{proposition 1}, we now bound the error of this alignment. 

\begin{lemma}[Covariate alignment]
\label{proposition 2}
Suppose that Assumption \ref{cond:A4} holds. Then there exists a constant $C_2>0$, 
\begin{eqnarray}
\label{boundness}
    \left\| {g}_i-T_n^{\top} \tilde{g}_i \right\| &\leq& C_2 n^{1/2}  \tau_n, \\
\label{covariate bound}
    n^{-1}\sum_{i=1}^{n }\left\| {g}_i-T_n^{\top} \tilde{g}_i \right\| &\leq&  C_2\tau_n,\\
\label{Tnbound}
 \left\|T_n^{\top}T_n-I_{K+p-r}\right\|_{\infty} &\leq& C_2 \tau_n.    
\end{eqnarray}
\end{lemma}

\begin{proof}[Proof of Lemma~\ref{proposition 2}]
We first prove \eqref{boundness}. Since the norm of a vector is always bounded by the sum of the norms of its blocks, we have 
\begin{eqnarray}
    \label{ineqFTn}
    \left\| {g}_i-T_n^{\top}\tilde{g}_i \right\| & \leq& \left\| Z^{\top}_{i,1:r}-n^{-1}{Z}^{\top}_{1:r}\tilde{Z}_{1:r} \tilde{Z}^{\top}_{i,1:r} \right\|\\
     \nonumber&+&\left\| {Z}^{\top}_{i,(r+1):p}-n^{-1}{Z}^{\top}_{(r+1):p}\tilde{Z}_{(r+1):p} \tilde{Z}^{\top}_{i,(r+1):p} \right\|\\
    \nonumber&+&\left\| {W}^{\top}_{i,(r+1):K}-n^{-1}{W}^{\top}_{(r+1):K}\tilde{W}_{(r+1):K} \tilde{W}^{\top}_{i,(r+1):K} \right\|.
\end{eqnarray}
We will prove that each term on the right-hand side of the above inequality is of order $O(n^{1/2}\tau_n)$.
For the first term, since columns of $Z_{1:r}$ are of norm $\sqrt{n}$, by Proposition~\ref{proposition 1}, 
\begin{align*}
    \left\|{Z}^{\top}_{i,1:r}-n^{-1}{Z}^{\top}_{1:r}\tilde{Z}_{1:r} \tilde{Z}^{\top}_{i,1:r}\right\| &=  \left\|{Z}_{i,1:r}-n^{-1}\tilde{Z}_{i,1:r}\tilde{Z}_{1:r}^{\top} {Z}_{1:r}\right\|  
    \leq  \left\|{Z}_{1:r}-n^{-1}\tilde{Z}_{1:r} \tilde{Z}^{\top}_{1:r} Z_{1:r}\right\|_F \\ 
    &= \left\|\left(I_{r}-n^{-1}\tilde{Z}_{1:r} \tilde{Z}_{1:r}^{\top}\right) Z_{1:r}\right\|_F\\
    & =\left\|n^{-1}\left({Z}_{1:r}{Z}_{1:r}^{\top}-\tilde{Z}_{1:r} \tilde{Z}_{1:r}^{\top}\right) {Z}_{1:r}\right\|_F
    =\left\|(\hat{\mathcal{P}}_R-\mathcal{P}_R) Z_{1:r}\right\|_F\\
    & \leq \sum_{i=1}^{r}\left\|(\hat{\mathcal{P}}_R-\mathcal{P}_R) Z_i\right\|
    \leq C_1rn^{1/2} \tau_n.
\end{align*}
Similarly, for the second term on the right-hand side of \eqref{ineqFTn}, we have
\begin{align*}
    \left\|{Z}^{\top}_{i,(r+1):p}-n^{-1}{Z}^{\top}_{(r+1):p}\tilde{Z}_{(r+1):p} \tilde{Z}^{\top}_{i,(r+1):p}\right\| 
     &\leq \sum_{i=1}^{p-r}\left\|(\hat{\mathcal{P}}_C-\mathcal{P}_C) Z_{r+i}\right\|
    \leq C_1(p-r)n^{1/2}  \tau_n.
\end{align*}
And for the last term on the right-hand side of \eqref{ineqFTn}, 
\begin{align*}
    \left\|{W}^{\top}_{i,(r+1):K}-n^{-1}{W}^{\top}_{(r+1):K}\tilde{W}_{(r+1):K} \tilde{W}^{\top}_{i,(r+1):K}\right\|  &\leq \sum_{i=1}^{K-r}\left\|(\hat{\mathcal{P}}_N-\mathcal{P}_N) {W}_{r+i}\right\|\\
    &\leq C_1(K-r)n^{1/2} \tau_n.
\end{align*}
These three inequalities imply \eqref{boundness}. 

We now prove \eqref{covariate bound}. Summing inequality \eqref{ineqFTn} over $i$ from 1 to $n$, we get
\begin{eqnarray}
    \label{eq:sum alignment}
    \sum_{i=1}^n\left\| {g}_i-T_n^{\top}\tilde{g}_i \right\| &\le& 
    \sum_{i=1}^{n }\left\| {Z}^{\top}_{i,1:r}-n^{-1}{Z}^{\top}_{1:r}\tilde{Z}_{1:r} \tilde{Z}^{\top}_{i,1:r} \right\|\\
    \nonumber
    &+&\sum_{i=1}^{n }\left\| {Z}^{\top}_{i,(r+1):p}-n^{-1}{Z}^{\top}_{(r+1):p}\tilde{Z}_{(r+1):p} \tilde{Z}^{\top}_{i,(r+1):p} \right\|\\
    \nonumber
    &+&\sum_{i=1}^{n }\left\| {W}^{\top}_{i,(r+1):K}-n^{-1}{W}^{\top}_{(r+1):K}\tilde{W}_{(r+1):K} \tilde{W}^{\top}_{i,(r+1):K} \right\|.
\end{eqnarray}
As before, we will show that each sum on the right-hand side of \eqref{eq:sum alignment} is of order $O(\tau_n)$. Regarding the first sum, by the Cauchy-Schwartz inequality and Proposition~\ref{proposition 1}, 
\begin{align*}
    \sum_{i=1}^n\left\|{Z}^{\top}_{i,1:r}-n^{-1}{Z}^{\top}_{1:r}\tilde{Z}_{1:r} \tilde{Z}^{\top}_{i,1:r}\right\| &=  \sum_{i=1}^n\left\|{Z}_{i,1:r}-n^{-1}\tilde{Z}_{i,1:r}\tilde{Z}_{1:r}^{\top} {Z}_{1:r}\right\|  \\
    &\leq n^{1/2}  \left\|{Z}_{1:r}-n^{-1}\tilde{Z}_{1:r} \tilde{Z}^{\top}_{1:r} Z_{1:r}\right\|_F \\ 
    &= n^{1/2}\left\|\left(I_{r}-n^{-1}\tilde{Z}_{1:r} \tilde{Z}_{1:r}^{\top}\right) Z_{1:r}\right\|_F\\
    & =n^{1/2}\left\|n^{-1}\left({Z}_{1:r}{Z}_{1:r}^{\top}-\tilde{Z}_{1:r} \tilde{Z}_{1:r}^{\top}\right) {Z}_{1:r}\right\|_F\\
    &=n^{1/2}\left\|(\hat{\mathcal{P}}_R-\mathcal{P}_R) Z_{1:r}\right\|_F\\
    & \leq n^{1/2}\sum_{i=1}^{r}\left\|(\hat{\mathcal{P}}_R-\mathcal{P}_R) Z_i\right\|
    \leq C_1rn \tau_n.
\end{align*}
Similarly, the second sum and the third sum on the right-hand side of \eqref{eq:sum alignment} are bounded by $C_1(p-r)n \tau_n$ and $C_1(K-r)n \tau_n$, respectively. These inequalities and \eqref{eq:sum alignment} then imply \eqref{covariate bound}.

Finally, we prove  \eqref{Tnbound}. Since $T_n$ is a block-diagonal matrix with three non-zero blocks on the diagonal, $T_n^{\top}T_n-I_{K+p-r}$ is also block-diagonal with three non-zero blocks on the diagonal given by:
\begin{equation*}  
\begin{aligned}
&n^{-2}Z_{1:r}^{\top}\Tilde{Z}_{1:r}\Tilde{Z}_{1:r}^{\top}Z_{1:r}-I_r,\\
&n^{-2}Z_{(r+1):p}^{\top}\Tilde{Z}_{(r+1):p}\Tilde{Z}_{(r+1):p}^{\top}Z_{(r+1):p}-I_{p-r},\\
& n^{-2}{W}_{(r+1):K}^{\top}\Tilde{W}_{(r+1):K}\Tilde{W}_{(r+1):K}^{\top}W_{(r+1):K}-I_{K-r}.
\end{aligned}
\end{equation*}
We will show that each diagonal block
is of order $O(\tau_n)$. Regarding the first block, 
for any unit vectors $u,v\in \mathbb{R}^{k-r}$, by Proposition \ref{proposition 1} we have
\begin{eqnarray*} u^\top\left(Z_{1:r}^{\top}\Tilde{Z}_{1:r}\Tilde{Z}_{1:r}^{\top}Z_{1:r}-n^{2}I_r\right)v 
&=& u^\top Z_{1:r}^{\top}\left(\Tilde{Z}_{1:r}\Tilde{Z}_{1:r}^{\top}-Z_{1:r}Z_{1:r}^\top\right)Z_{1:r}v\\
&\le& \| Z_{1:r}u\|\big\|\Tilde{Z}_{1:r}\Tilde{Z}_{1:r}^{\top}-Z_{1:r}Z_{1:r}^\top\big\|\|Z_{1:r}v\| \\
&=& n\| Z_{1:r}u\|\big\|\mathcal{P}_C-\hat{\mathcal{P}}_C\big\|\|Z_{1:r}v\| \\
&\le&C_1n\tau_n\| Z_{1:r}u\|\| Z_{1:r}v\|.
\end{eqnarray*}
Since columns of $Z_{1:r}$ are of norm $n^{1/2}$, it follows that $\| Z_{1:r}u\|\| Z_{1:r}v\|\le rn$. Because $u$ and $v$ are arbitrary, this implies the infinity norm of the first diagonal block is at most $C_1rn^2\tau_n$. The same argument can be applied to the second and third diagonal blocks to show that their infinity norms are bounded by $C_1(p-r)n^2\tau_n$ and $C_1(K-r)n^2\tau_n$, respectively. Together, these bounds imply \eqref{Tnbound} and the proof of 
Lemma~\ref{proposition 2} is complete.
\end{proof}

Then we show an additional Lemma aims to bound the $T_n$'s eigenvalues that will be applied repeatly in the proof of main theorems.

\begin{corollary}[Eigenvalue bound of $T_n$]
\label{proposition 3}
Suppose that Assumption \ref{cond:A4} holds. Then with the same constant $C_2$ from Lemma~\ref{proposition 2}
\begin{equation}
\label{Tneigen}
1-C_2\tau_n \leq \lambda^2_{\min }(T_n) \leq \lambda^2_{\max }(T_n) \leq 1+C_2\tau_n,  
\end{equation}
for sufficiently large $n$
\end{corollary}
\begin{proof}[Proof of Corollary~\ref{proposition 3}]
Based on Lemma~\ref{proposition 2}, $T_n$ is nonsingular when $\tau_n<1/C_2$. Then we try to bound $T_n$'s eigenvalue $\lambda$ by bound $\|T_nu_{\lambda}\|$, where $\lambda u_\lambda=T_nu_{\lambda}$ and $u_\lambda\in \mathbb{R}^{K+p-r}$ is a unit vector. By infinity norm bound \eqref{Tnbound} in Lemma~\ref{proposition 2}
\begin{eqnarray*}
    \left|\lambda^2-1\right|&=&
    \left|\left\|T_nu_\lambda\right\|^2-1\right|= 
    u_\lambda^\top T_n^\top T_n u_\lambda-1=  u_\lambda^\top (T_n^\top T_n-I_{K+p-r} )u_\lambda\\
    &\le& 
    \|T_n^\top T_n-I_{K+p-r}\|\le \|T_n^\top T_n-I_{K+p-r}\|_{\infty}\le C_2\tau_n.
\end{eqnarray*}
Therefore, 
\begin{equation*}
    1-C_2\tau_n\le \lambda^2\le 1+C_2\tau_n.
\end{equation*}
Since $\lambda$ can be any eigenvalue of $T_n$, we have proved \eqref{Tneigen}.
\end{proof}

\section{The Proof of Theorem~\ref{existence+consistency}}\label{sec:Proof of Theorem 1}

We proceed to prove Theorem~\ref{existence+consistency}  about the existence and consistency of the proposed estimates. Overall, we follow the proof strategy for generalized linear models with fixed design \cite{yin2006asymptotic}. The main difference between our proof and the proof in \cite{yin2006asymptotic} is that the combinations of covariate and relational information for all the nodes, denoted by $\Tilde{g}_i$, are not exactly observed. Therefore, we need to carefully track down the measurement errors.
To prove Theorem~\ref{existence+consistency}, we begin with the following remarks and lemmas.
\begin{remark}
\label{remark 1}
 Assumption~\ref{cond:A1} implies that $\gamma^*$ is bounded because $$\left\|\gamma^*\right\|=n^{-1/2}\|W\gamma^*\|=n^{-1/2}\|X\beta^*+X\theta^*+\alpha^*\|\le n^{-1/2}\left(\left\|X\beta^*\right\|+\left\|X\theta^*\right\|+\left\|\alpha^*\right\|\right)\le 3C.$$
\end{remark}

\begin{remark}
\label{remark 2}
Denote $\eta:= h^{\prime}/v:\mathbb{R}\rightarrow \mathbb{R}$.
Then there exists a constant $M>0$, when $\left|t\right|\le 12C^2$, the absolute value of function $h\left(t\right),v\left(t\right),\eta\left(t\right),\eta\left(t\right)h\left(t\right), \eta\left(t\right)h^{\prime}\left(t\right)$ and their first and second derivatives are all bounded by $M$ because every real-valued continuous function on a compact set is necessarily bounded.
\end{remark}

\begin{lemma}[Lemma 3 in \cite{yin2006asymptotic}]
\label{conti}
    Let $\varphi:\mathbb{R}^{m}\to\mathbb{R}^{m}$ be a smooth injective map with $\varphi\left(x^*\right)=y^*$ and for some $\rho,\delta>0$, $$\min _{\left\|x-x^*\right\|=\delta}\left\|\varphi(x)-y^*\right\| \geq \rho.$$
    Then for any $y$ with $\left\|y-y^*\right\| \le \rho$, there exists $x$ with $\left\|x-x^*\right\| \leq \delta$ such that $\varphi(x)=y$.
\end{lemma}

\begin{lemma}
\label{lemma 7}
Assume that Assumptions \ref{cond:A1} to \ref{cond:A4} hold. For a constant $\delta_0>0$, denote
$$
    N_n(\delta_0) = \left\{\gamma:\|T^{-1}_n\gamma-\gamma^*\| \le \delta_0 n^{-1/2} \right\}.
$$
Then there exists a constant  $c>0$ such that for any $\varepsilon>0$ and sufficiently large $n$, with probability at least $1-\varepsilon$,
\begin{eqnarray}
\label{boundary}
\inf _{ \gamma \in \partial N_n(\delta_0)} \left\|T^{\top}_n\Tilde{S}(\gamma)-T^{\top}_n\Tilde{S}\left(T_n\gamma^*\right)\right\| &\geq&  cn^{-1/2},\\
\left\|T^{\top}_n\Tilde{S}\left(T_n\gamma^*\right)\right\| &\leq& cn^{-1/2}.
\label{gamma0}
\end{eqnarray}
\end{lemma}

This lemma is crucial for proving Theorem \ref{existence+consistency}. Its proof is given in Appendix~\ref{app: Tech}.

\begin{proof}[Proof of Theorem \ref{existence+consistency}]
We first prove \eqref{existence_core}. 
Instead of working with the sample score function $\Tilde{S}(\gamma)$ directly, it will be more convenient to scale it and work with
$$L(\gamma):= T^{\top}_n\Tilde{S}(T_n\gamma).$$
The estimating equation $\Tilde{S}(\gamma)=0$ is equivalent to 
$$L(T_n^{-1}\gamma)-L(\gamma^*)=-L(\gamma^*).$$
To prove \eqref{existence_core}, we apply Lemma~\ref{conti} with $\varphi(x)= L(x)-L(\gamma^*)$, $x^* = \gamma^*$, and $y^*=0$. 
According to this lemma, for $y = -L(\gamma^*)$ with $\|y-y^*\|=\|L(\gamma^*)\|=:\rho$, there exists $x=T_n^{-1}\gamma$ with $\gamma:=T_nx$ and $\|x-x^*\|=\|T_n^{-1}{\gamma}-\gamma^*\|\le\delta$
such that $\varphi(x)=y$, or equivalently $\tilde{S}({\gamma})=0$. We need to specify $\delta$ such that the following condition of the lemma holds:
$$
\min_{\|x-x^*\|=\delta}\|\varphi(x)-y^*\|=\min_{\|T_n^{-1}\gamma-\gamma^*\|=\delta}\|\ L(T_n^{-1}\gamma)-L(\gamma^*)\|\ge\rho.
$$
For consistency of ${\gamma}=T_nx$ such that $\tilde{S}(\gamma)=0$, we choose  $\delta=\delta_0n^{-1/2}$ for some $\delta_0$ and denote 
\begin{equation}
\label{Nndelta}
    N_n(\delta_0) = \left\{\gamma:\|T^{-1}_n\gamma-\gamma^*\| \le \delta_0 n^{-1/2} \right\}.
\end{equation}
By combining  \eqref{boundary} and \eqref{gamma0} in Lemma~\ref{lemma 7}, we obtain that with probability at least $1-\varepsilon$, 
$$
\min _{ \gamma \in \partial N_n(\delta_0)} \left\|L(T_n^{-1}\gamma)-L\left(\gamma^*\right)\right\|\ge \left\|L\left(\gamma^*\right)\right\|=\rho.
$$
The proof of \eqref{existence_core} is complete.

We now prove \eqref{Consistencyforpara}, starting with the consistency of $\hat{{\alpha}}$. By Proposition \ref{proposition 1} and Lemma \ref{proposition 2}, 
\begin{eqnarray*}
    \left\|\hat{{\alpha}}-{\alpha^*}\right\|&=&\left\|\Tilde{W}_{(r+1):K}\hat{\gamma}_{(p+1):(p+K-r)}-W_{(r+1):K}\gamma^*_{(p+1):(p+K-r)}\right\|\\
    &\leq& \left\|\Tilde{W}_{(r+1):K}\left(\hat{\gamma}_{(p+1):(p+K-r)}-\left[\tilde{W}_{(r+1):K}^{\top} {W}_{(r+1):K}/n\right]\gamma^*_{(p+1):(p+K-r)}\right)\right\|\\    &&+\left\|\left(\Tilde{W}_{(r+1):K}\Tilde{W}_{(r+1):K}^{\top} W_{(r+1):K}/n-W_{(r+1):K}\right)\gamma^*_{(p+1):(p+K-r)}\right\|\\   &=&n^{1/2}\left\|\hat{\gamma}_{(p+1):(p+K-r)}-\left[\tilde{W}_{(r+1):K}^{\top} {W}_{(r+1):K}/n\right]\gamma^*_{(p+1):(p+K-r)}\right\|\\  &&+n^{-1}\left\|\left(\Tilde{W}_{(r+1):K}\Tilde{W}_{(r+1):K}^{\top}-W_{(r+1):K}W_{(r+1):K}^{\top}\right)W_{(r+1):K}\gamma^*_{(p+1):(p+K-r)}\right\|\\
     &=& n^{1/2}\left\|\tilde{W}_{(r+1):K}^{\top} {W}_{(r+1):K}/n\left(\left(\tilde{W}_{(r+1):K}^{\top} {W}_{(r+1):K}/n\right)^{-1}\hat{\gamma}_{(p+1):(p+K-r)}-\gamma^*_{(p+1):(p+K-r)}\right)\right\|\\   &&+\left\|\left(\hat{\mathcal{P}}_N-\mathcal{P}_N\right)\alpha^*\right\|\\
     &\le & n^{-1/2}\left\|\Tilde{W}_{(r+1):K}\right\|\left\| {W}_{(r+1):K}\right\|\left\|\left(\tilde{W}_{(r+1):K}^{\top} {W}_{(r+1):K}/n\right)^{-1}\hat{\gamma}_{(p+1):(p+K-r)}-\gamma^*_{(p+1):(p+K-r)}\right\|\\
&&+\left\|\left(\hat{\mathcal{P}}_N-\mathcal{P}_N\right)\alpha^*\right\|\\
     &\leq &n^{1/2}\left\|T^{-1}_n\hat{\gamma}-\gamma^*\right\|+\left\|\left(\hat{\mathcal{P}}_N-\mathcal{P}_N\right)\alpha^*\right\|\\
    &\leq &o_p(n^{1/2})+n^{1/2}C_2C\tau_n\\
    &=&o_p(n^{1/2}).
\end{eqnarray*}
We now prove the consistency for $\hat{\beta}$. 
Following the above argument for the bound of $\hat{\alpha}$, we obtain $\|X(\hat{{\beta}}-\beta^*)\|=o_p(n^{1/2})$.
\begin{eqnarray*}
    \left\|X(\hat{{\beta}}-\beta^*)\right\|&=&\left\|\tilde{Z}_{(r+1):p}\hat{\gamma}_{(r+1):p}-Z_{(r+1):p}\gamma^*_{(r+1):p}\right\|\\
    &\leq& \left\|\tilde{Z}_{(r+1):p}\left(\hat{\gamma}_{(r+1):p}-\left[\tilde{Z}_{(r+1):p}^{\top} Z_{(r+1):p}/n\right]\gamma^*_{(r+1):p}\right)\right\|\\    
     &\le & n^{-1/2}\left\|\tilde{Z}_{(r+1):p}\right\|\left\| Z_{(r+1):p}\right\|\left\|\left(\tilde{Z}_{(r+1):p}^{\top} Z_{(r+1):p}/n\right)^{-1}\hat{\gamma}_{(r+1):p}-\gamma^*_{(r+1):p}\right\|\\
&&+\left\|\left(\hat{\mathcal{P}}_C-\mathcal{P}_C\right)\alpha^*\right\|\\
     &\leq &n^{1/2}\left\|T^{-1}_n\hat{\gamma}-\gamma^*\right\|+\left\|\left(\hat{\mathcal{P}}_C-\mathcal{P}_C\right)\alpha^*\right\|\\
    &\leq &o_p(n^{1/2})+n^{1/2}C_2C\tau_n\\
    &=&o_p(n^{1/2}).
\end{eqnarray*}
Denote $u=\|\hat{{\beta}}-\beta^*\|^{-1}(\hat{{\beta}}-\beta^*)$.
By the definition of $G=(X^\top X/n)^{-1}$ in Assumption \ref{cond:A5},
\begin{equation*}
    \left\|\hat{{\beta}}-\beta^*\right\|^2\left(u^{\top}G^{-1}u\right)=\left(\hat{{\beta}}-\beta^*\right)^{\top}G^{-1}\left(\hat{{\beta}}-\beta^*\right)=n^{-1}\|X(\hat{{\beta}}-\beta^*)\|^2.
\end{equation*}
Therefore by Assumption \ref{cond:A5},
\begin{equation*}
    \left\|\hat{{\beta}}-\beta^*\right\|=n^{-1/2}\left(u^{\top}G^{-1}u\right)^{-1/2}\|X(\hat{{\beta}}-\beta^*)\|\le n^{-1/2}\lambda^{-1/2}_{\min}(G)\|X(\hat{{\beta}}-\beta^*)\|=o_p(1).
\end{equation*}
For the consistency of $\hat{\theta}$,
\begin{eqnarray*}
    \left\|X(\hat{{\theta}}-\theta^*)\right\|&=&\left\|\tilde{Z}_{1:r}\hat{\gamma}_{1:r}-Z_{1:r}\gamma^*_{1:r}\right\|\\
    &\leq& \left\|\tilde{Z}_{1:r}\left(\hat{\gamma}_{1:r}-\left[\tilde{Z}_{1:r}^{\top} Z_{1:r}/n\right]\gamma^*_{1:r}\right)\right\|\\    
     &\le & n^{-1/2}\left\|\tilde{Z}_{1:r}\right\|\left\| Z_{1:r}\right\|\left\|\left(\tilde{Z}_{1:r}^{\top} Z_{1:r}/n\right)^{-1}\hat{\gamma}_{1:r}-\gamma^*_{1:r}\right\|\\
&&+\left\|\left(\hat{\mathcal{P}}_C-\mathcal{P}_C\right)\alpha^*\right\|\\
     &\leq &n^{1/2}\left\|T^{-1}_n\hat{\gamma}-\gamma^*\right\|+\left\|\left(\hat{\mathcal{P}}_C-\mathcal{P}_C\right)\alpha^*\right\|\\
    &\leq &o_p(n^{1/2})+n^{1/2}C_2C\tau_n\\
    &=&o_p(n^{1/2}).
\end{eqnarray*}
Denote $u=\|\hat{{\theta}}-\theta^*\|^{-1}(\hat{{\theta}}-\theta^*)$.
Similar to the consistency of $\hat{\theta}$, by the definition of $G=(X^\top X/n)^{-1}$ in Assumption \ref{cond:A5},
\begin{equation*}
    \left\|\hat{{\theta}}-\theta^*\right\|=n^{-1/2}\left(u^{\top}G^{-1}u\right)^{-1/2}\|X(\hat{{\theta}}-\theta^*)\|\le n^{-1/2}\lambda^{-1/2}_{\min}(G)\|X(\hat{{\theta}}-\theta^*)\|=o_p(1).
\end{equation*}

\end{proof}

\section{Proof of Lemma~\ref{lemma 7}}
\label{app: Tech}

The proof of Lemma~\ref{lemma 7} follows from several technical lemmas in this section. 
Recall matrix $T_n$ from \eqref{eq:Tn} and  $T_n^\top\tilde{g}_i\approx g_i$ by Lemma~\ref{proposition 2}. 
The next lemma is a consequence of the Mean Value Theorem and will be used repeatedly.

\begin{lemma}[Covariate alignment]
    \label{lemma 3}
Assume that Assumptions \ref{cond:A1}, \ref{cond:A2} and \ref{cond:A4} hold. Let $\eta:\mathbb{R}\rightarrow \mathbb{R}$ be a function with continuous derivative.
Then 
\begin{equation*}
\left\|T_n^{\top}\tilde{g}_i \eta(\Tilde{g}^{\top}_i T_n \gamma)-{g}_i \eta(g^{\top}_i\gamma)\right\| \leq \left(2C\|\gamma\|\sup_{\left|t\right|\leq 2C\|\gamma\|}|\eta^{\prime}\left(t\right)|+\sup_{\left|t\right|\leq C\|\gamma\|}|\eta\left(t\right)|\right)\left\|g_i-{T}^{\top}_n \tilde{g}_i\right\|,
\end{equation*}
for sufficiently large $n$.
\end{lemma}
\begin{proof} [Proof of Lemma~\ref{lemma 3}]
By the triangle inequality,
    \begin{align*}
\left\|T_n^{\top}\tilde{g}_i \eta(\Tilde{g}^{\top}_i T_n \gamma)-{g}_i \eta(g^{\top}_i\gamma)\right\| &\leq   \left\|T_n^{\top}\tilde{g}_i\left( \eta(\Tilde{g}^{\top}_i T_n \gamma)- \eta(g^{\top}_i\gamma)\right)\right\|+  \left\|\left( T_n^{\top}\tilde{g}_i-{g}_i\right) \eta(g^{\top}_i\gamma)\right\|\\
        &\leq \left\|T_n^{\top}\tilde{g}_i\right\| \left| \eta(\Tilde{g}^{\top}_i T_n \gamma)- \eta(g^{\top}_i\gamma)\right|+\left\| T_n^{\top}\tilde{g}_i-{g}_i\right\| |\eta(g_i^\top\gamma)|.
    \end{align*}
By Assumption~\ref{cond:A1}, we have $|g_i^\top\gamma|\le \|g_i\|\|\gamma\|\le C\|\gamma\|$, which implies 
$$
|\eta(g_i^\top\gamma)|\le \sup_{\left|t\right|\le C\|\gamma\|}|\eta\left(t\right)|.
$$
We will use the Mean Value Theorem to bound $\left| \eta(\Tilde{g}^{\top}_i T_n \gamma)- \eta(g^{\top}_i\gamma)\right|$. Denote 
$$h\left(t\right) = \eta\left(\Tilde{g}^{\top}_i T_n \gamma + t\left[g^{\top}_i\gamma-\Tilde{g}^{\top}_i T_n\gamma\right]\right):\mathbb{R}\to\mathbb{R}.$$ 
By the Mean Value Theorem, 
 there exists $t^*\in[0,1]$ and $z_i = \Tilde{g}^{\top}_i T_n \gamma + t^*\left[g^{\top}_i\gamma-\Tilde{g}^{\top}_i T_n\gamma\right]$ such that 
\begin{eqnarray*}
    \left| \eta(\Tilde{g}^{\top}_i T_n \gamma)- \eta(g^{\top}_i\gamma)\right| &=& |h(1)-h(0)| = |h'(t^*)| = |\eta'(z_i)\left(g^{\top}_i\gamma-\Tilde{g}^{\top}_i T_n\gamma\right)|\\
    &\leq& \|g_i- T_n^\top\Tilde{g}_i\|\|\gamma\|\left|\eta^{\prime}(z_i)\right|.
\end{eqnarray*}
By the triangle inequality,
$$
|z_i|=\big|(1-t^*)\Tilde{g}^{\top}_i T_n \gamma+t^*g_i^\top\gamma\big|\le \max\{\|\Tilde{g}^{\top}_i T_n \|,\|g_i\|\}\cdot\|\gamma\|.
$$
By Assumption~\ref{cond:A1}, we have $\|g_i\|\le C$. In addition, by Assumptions  \ref{cond:A2}, \ref{cond:A4}, and Lemma~\ref{proposition 2},
\begin{equation}
\label{covariate bound tilde}
    \left\|T_n^{\top}\tilde{g}_i\right\| \le \left\|T_n^{\top}\tilde{g}_i-g_i\right\|+\|g_i\|\le C_2 n^{1/2}  \tau_n + C \le 2C,
\end{equation}
when $n$ is big enough. It follows that $|\eta'(z_i)|\le \sup_{\left|t\right|\le 2C\|\gamma\|}\|\eta'\left(t\right)\|$.
Putting these inequalities together, we get
$$
\left\|T_n^{\top}\tilde{g}_i \eta(\Tilde{g}^{\top}_i T_n \gamma)-{g}_i \eta(g^{\top}_i\gamma)\right\| \le \left(2C\|\gamma\|\sup_{\left|t\right|\leq 2C\|\gamma\|}|\eta^{\prime}\left(t\right)|+\sup_{\left|t\right|\leq C\|\gamma\|}|\eta\left(t\right)|\right)\left\|g_i-{T}^{\top}_n \tilde{g}_i\right\|.
$$
The proof is complete.
\end{proof}

The next lemma bounds the difference between the sample information matrix $\tilde{F}$ and its population counterpart $F$, defined in \eqref{tildeF} and \eqref{F}, respectively.  
This bound will be applied multiple times in the proof of Theorem \ref{existence+consistency}.

\begin{lemma}[Information matrix bounds]
\label{lemma 4}
Denote $\varphi=(h^{\prime})^2/ v$. 
Under Assumptions \ref{cond:A1}, \ref{cond:A2}, and \ref{cond:A4}, 
for sufficiently large $n$ we have 
\begin{equation*}
  \big\|T_n^{\top}\tilde{F}(T_n\gamma)T_n-F({{\gamma}})\big\|\leq \Psi(\|\gamma\|)\tau_n, 
\end{equation*}
where $\Psi:\mathbb{R}\to\mathbb{R}$ is a non-decreasing function defined by
$$
\Psi(s) = C^2C_2s\sup_{\left|t\right|\leq 2Cs}|\varphi^{\prime}\left(t\right)|+\left(2C+C_2\right)C_2\sup_{\left|t\right|\leq 2C s}|\varphi\left(t\right)|.
$$
In addition, 
if Assumption \ref{cond:A3} also holds then for any $\gamma$ such that $\|\gamma-\gamma^*\|<\delta$, we have
\begin{equation}
\label{eigen}
1/C-\Psi(\|\gamma\|)\tau_n \leq \lambda_{\min }(T_n^{\top}\tilde{F}(T_n\gamma)T_n) \leq \lambda_{\max }(T_n^{\top}\tilde{F}(T_n\gamma)T_n) \leq C+\Psi(\|\gamma\|)\tau_n.  
\end{equation}
\end{lemma}
\begin{proof}[Proof of Lemma~\ref{lemma 4}]
We proceed to prove the first inequality in Lemma~\ref{lemma 4}. 
Recall the formulas for 
$\tilde{F}$ and $F$ in \eqref{tildeF} and \eqref{F}, respectively. Since $\varphi=(h^{\prime})^2/ v$, it follows that
\begin{eqnarray*}  
F(\gamma)&=&n^{-1}\sum_{i=1}^n\varphi({g}_i^{\top}\gamma){g}_i{g}_i^{\top},\\
T_n^{\top}\tilde{F}(T_n\gamma)T_n&=&n^{-1}\sum_{i=1}^n\varphi(\tilde{g}_i^{\top}T_n\gamma)T_n^{\top}\tilde{g}_i\tilde{g}_i^{\top}T_n.
\end{eqnarray*}
We will bound $u^{\top}(T_n^{\top}\tilde{F}(T_n\gamma)T_n-F(\gamma))u:=\Phi_1+\Phi_2$ for any 
fixed unit vector $u$,
where
\begin{eqnarray*}
   \Phi_1&=&n^{-1} \sum_{i=1}^n \varphi\left( \Tilde{g}_i^{\top}T_n\gamma\right)u^{\top}\left(T_n^{\top}\tilde{g}_i\tilde{g}_i^{\top}T_n-g_ig_i^{\top}\right)u, \\
      \Phi_2&=&n^{-1} \sum_{i=1}^n \left(\varphi(\Tilde{g}_i^{\top}T_n\gamma)-\varphi({g}_i^{\top}\gamma)\right)u^{\top}
g_ig_i^{\top} u.
\end{eqnarray*}
Regarding $\Phi_1$, by  Asumption~\ref{cond:A1} and Lemma \ref{proposition 2}, we have
\begin{align*}
    \left|\Phi_1\right| &\leq n^{-1}\max_{1\leq i \leq n} {\left|\varphi(\Tilde{g}^{\top}_iT_n \gamma)\right|}\sum_{i=1}^n \left|u^{\top}\left(T_n^{\top}\tilde{g}_i\tilde{g}_i^{\top}T_n-g_i g_i^{\top}\right)u\right|\\
&= 
n^{-1}\max_{1\leq i \leq n}{\left|\varphi(\Tilde{g}^{\top}_iT_n \gamma)\right|}\sum_{i=1}^n \left|\left(u^{\top}T_n^{\top}\tilde{g}_i\right)^2-\left(u^{\top}g_i\right)^2\right|
\\
&= n^{-1}\max_{1\leq i \leq n}  \left|\varphi(\Tilde{g}^{\top}_iT_n \gamma)\right|\sum_{i=1}^n \left| u^{\top}(T_n^{\top}\Tilde{g}_i+g_i)\right|\left| u^{\top}(T_n^{\top}\Tilde{g}_i-g_i)\right| \\
&\leq  n^{-1}\max_{1\leq i \leq n}  \left|\varphi(\Tilde{g}^{\top}_iT_n \gamma)\right|\sum_{i=1}^n\left\|T_n^{\top}\Tilde{g}_i+g_i\right\|\left\|T_n^{\top}\Tilde{g}_i- g_i\right\| \\
&\leq  n^{-1}\max_{1\leq i \leq n}  \left|\varphi(\Tilde{g}^{\top}_iT_n \gamma)\right|\sum_{i=1}^n\left(\left\| T_n^{\top}\Tilde{g}_i-g_i\right\|+2\|g_i\|\right) \left\|T_n^{\top}\Tilde{g}_i-g_i\right\|\\
&\leq  n^{-1}\max_{1\leq i \leq n}  \left|\varphi(\Tilde{g}^{\top}_iT_n \gamma)\right|\sum_{i=1}^n\left(C_2\tau_n+2C\right) \left\|T_n^{\top}\Tilde{g}_i-g_i\right\|\\
& \leq \left(\left(2C+C_2\right)\max_{1\leq i \leq n}  \left|\varphi(\Tilde{g}^{\top}_iT_n \gamma)\right|\right)n^{-1} \sum_{i=1}^n \left\|T_n^{\top}\Tilde{g}_i- g_i\right\|\\
& \leq \left(\left(2C+C_2\right)\max_{1\leq i \leq n}  \left|\varphi(\Tilde{g}^{\top}_iT_n \gamma)\right|\right)C_2\tau_n.
\end{align*}
We now bound $|\Phi_2|$. Denote $$\varphi_i\left(t\right)=\varphi(g_i^{\top}{\gamma}+t(T_n^{\top}\tilde{g}_i-g_i)^{\top}\gamma ).$$
By the Mean Value Theorem, 
there exists $t^*\in[0,1]$ and $p_i = \Tilde{g}^{\top}_i T_i + t^*(g_i-T^{\top}_n\Tilde{g}_i)$ such that
\begin{align*}
\left|\Phi_2\right|
&=n^{-1}\Big|\sum_{i=1}^n (u^{\top}g_i)^2(\varphi_i(1)-\varphi_i(0))\Big|
= n^{-1}\Big|\sum_{i=1}^n (u^{\top}g_i)^2\varphi'_i(t^*)\Big|\\
& =n^{-1}\Big|\sum_{i=1}^n (u^{\top}g_i)^2\varphi^{\prime}(p^{\top}_i\gamma)(T_n^{\top}\tilde{g}_i-g_i)^{\top}\gamma \Big|\\
& \leq n^{-1}\|\gamma\| \max_{1\leq i \leq n} |\varphi^{\prime}(p^{\top}_i\gamma)|\max_{1\leq i \leq n} \|g_i\|^2\sum_{i=1}^n \left\|T_n^{\top}\tilde{g}_i-g_i\right\|\\
& \leq n^{-1} C^2\|\gamma\| \max_{1\leq i \leq n} |\varphi^{\prime}(p^{\top}_i\gamma)|\sum_{i=1}^n \left\|T_n^{\top}\tilde{g}_i-g_i\right\|\\
& \leq C^2C_2\|\gamma\| \max_{1\leq i \leq n} |\varphi^{\prime}(p^{\top}_i\gamma)|\tau_n,
\end{align*}
where the second inequality follows from Assumption~\ref{cond:A1} and the last inequality follows from Lemma~\ref{proposition 2}.
By Assumption \ref{cond:A5}, Lemma \ref{proposition 2}, and \eqref{covariate bound tilde},
\begin{equation*}
    \|p_i\|\leq (1-t^*)\|T^{\top}_n \tilde{g}_i\|+t^*\|g_i\|\leq t^*C+2(1-t^*)C\leq 2C,
\end{equation*}
for sufficiently large $n$.

Putting these inequalities together, we have
\begin{align*}
 \left\|T_n^{\top}\tilde{F}(T_n\gamma)T_n-F(\gamma)\right\|
 &\leq \left(C^3\|\gamma\| \max_{1\leq i \leq n} |\varphi^{\prime}(p^{\top}_i\gamma)|+3C^2\max_{1\leq i \leq n}  \left|\varphi(\Tilde{g}^{\top}_iT_n \gamma)\right|\right)\tau_n\\
&\leq \left(C^2C_2\|\gamma\|\max_{\left|t\right|\leq 2C\|\gamma\|}|\varphi^{\prime}\left(t\right)|+\left(2C+C_2\right)C_2\max_{\left|t\right|\leq 2C \|\gamma\|}|\varphi\left(t\right)|\right)\tau_n.
\end{align*}

We now prove the second claim in Lemma~\ref{lemma 4}. From Assumption~\ref{cond:A3} and the first claim of Lemma~\ref{lemma 4}, we have
\begin{align*}
     \lambda_{\min }(T_n^{\top}\tilde{F}(T_n\gamma)T_n)
     &\geq \lambda_{\min }(F(\gamma))-\|T_n^{\top}\tilde{F}(T_n\gamma)T_n-F(\gamma)\|
     \geq 1/C-\Psi(\|\gamma\|)\tau_n.
\end{align*}
Similarly, 
\begin{align*}
     \lambda_{\max }(T_n^{\top}\tilde{F}(T_n\gamma)T_n)
     \le \lambda_{\max }(F(\gamma))+\|T_n^{\top}\tilde{F}(T_n\gamma)T_n-F(\gamma)\|
     \le C+\Psi(\|\gamma\|)\tau_n.
\end{align*}
The proof is completed.
\end{proof}

A direct consequence of \eqref{eigen} with $\gamma=\gamma^*$ is that the scaling matrix $T_n^{\top}\tilde{F}\left(T_n\gamma^*\right)T_n$, which appears in the proof of Theorem~\ref{existence+consistency}, is well-conditioned.

\begin{corollary}[Scaling matrix is well-conditioned]\label{cor:scaling matrix}
    Assume that the conditions in Lemma~\ref{lemma 4} hold. Then, as $n$ is sufficiently large,
    \begin{equation}
\label{eigen_inverse}
\frac{1}{2C} \leq \lambda_{\min }\left(T_n^{\top}\tilde{F}\left(T_n\gamma^*\right)T_n\right) \leq \lambda_{\max }\left(T_n^{\top}\tilde{F}\left(T_n\gamma^*\right)T_n\right) \leq 2C. 
\end{equation}
\end{corollary}
\begin{proof}[Proof of Corollary~\ref{cor:scaling matrix}]
Let $\Psi$ be the function defined in Lemma~\ref{lemma 4}. By Remarks~\ref{remark 1} and ~\ref{remark 2}, we have
\begin{equation*}
  \Psi(\left\|\gamma^*\right\|) = C^2C_2\left\|\gamma^*\right\|\sup_{\left|t\right|\leq 2C\left\|\gamma^*\right\|}|\varphi^{\prime}\left(t\right)|+3CC_2\sup_{\left|t\right|\leq 2C \left\|\gamma^*\right\|}|\varphi\left(t\right)|\le 3C^3C_2M+3CC_2M. 
\end{equation*}
Since it is assumed that $\tau_n\to 0$, this implies  $\Psi(\left\|\gamma^*\right\|)\tau_n$ is close to zero as  $n$ is sufficiently large, and \eqref{eigen} implies \eqref{eigen_inverse}. 
\end{proof}

The following lemma shows that the estimating equation \eqref{estimating equation real} and its sample counterpart \eqref{estimating equation} are sufficiently close.

\begin{lemma}[Estimating equation bounds]
\label{lemma 5}
Under the Assumption \ref{cond:A1}, \ref{cond:A2}, \ref{cond:A4} and \ref{cond:A6}, we have
\begin{equation}
\label{estimatingequinequality} \left\|T_n^{\top}\Tilde{S}\left(T_n\gamma^*\right)-S\left(\gamma^*\right)\right\| =o_p(n^{-1/2})
\end{equation}
and
\begin{equation}  
\label{estimatingequinequality_N}
\left\|\mathbb{E}\left(T_n^{\top}\Tilde{S}\left(T_n\gamma^*\right)-S\left(\gamma^*\right)\right)\right\| =o(n^{-1/2}).
\end{equation}
\end{lemma}

We need the following lemma to prove Lemma~\ref{lemma 5}.

\begin{lemma}[Lemma 5.1 in \cite{stefanski1985covariate}]
\label{techniquelemma1}
Let $(U_i)_{i=1}^\infty$ be a sequence of independent random variables with zero means and $\mathbb{E}\left[|U_{i}|^{1+\zeta}\right] <\infty$ for all $i$ and some $\zeta>0$. If a sequence of scalars $(a_i)_{i=1}^\infty$ satisfies  $\sum_{i=1}^n|a_i|=O(n)$ and $\max _{1 \leq i \leq n}|a_i| =o(n)$ then $\sum_{i=1}^n a_i U_i=o_p(n)$.  
\end{lemma}

\begin{proof}[Proof of Lemma~\ref{lemma 5}]
For the notation convenience, denote $\eta= h^{\prime}/v:\mathbb{R}\rightarrow \mathbb{R}$. Then,
\begin{eqnarray*}  
S\left(\gamma^*\right)&=&\frac{1}{n}\sum_{i=1}^n   {g}_i \eta\left({g}^{\top}_i \gamma^*\right)\left[y_i-h\left({g}^{\top}_i \gamma^*\right)\right],\\
T_n^{\top}\Tilde{S}\left(T_n\gamma^*\right)&=&\frac{1}{n}\sum_{i=1}^n  T^{\top}_n \Tilde{g}_i \eta\left(\tilde{g}^{\top}_i T_n\gamma^*\right)\left[y_i-h\left(\tilde{g}^{\top}_i T_n\gamma^*\right)\right].
\end{eqnarray*}
We first bound the difference between their expectations: 
$$
\mathbb{E}\left[T_n^{\top}\Tilde{S}\left(T_n\gamma^*\right)-S\left(\gamma^*\right)\right]=:B_1-B_2, 
$$
where
\begin{eqnarray*}
    B_{1}&=&\frac{1}{n}\sum_{i=1}^n\Big[T^{\top}_n \Tilde{g}_i \eta(\tilde{g}^{\top}_i T_n\gamma^*)-{g}_i \eta({g}^{\top}_i \gamma^*)\Big] h(g^{\top}_i \gamma^*),
    \\
    B_{2}&=&\frac{1}{n}\sum_{i=1}^n\Big[ T^{\top}_n \Tilde{g}_i \eta(\tilde{g}^{\top}_i T_n\gamma^*)h(\Tilde{g}^{\top}_i T_n^{\top}\gamma^*)-{g}_i \eta({g}^{\top}_i \gamma^*)h(g^{\top}_i \gamma^*)\Big].
\end{eqnarray*}
By the triangle inequality, Lemma~\ref{lemma 3} and covariate bound \eqref{covariate bound} we have in Lemma~\ref{proposition 2}, 
\begin{equation*}
    \begin{aligned}
    \left\|B_{1}\right\| 
    &\le n^{-1}\max_{1\leq i\leq n} |h\left(g^{\top}_i\gamma^*\right)|\sum_{i=1}^{n} \left\|\left( T^{\top}_n \Tilde{g}_i \eta\left(\tilde{g}^{\top}_i T_n\gamma^*\right)-{g}_i \eta\left({g}^{\top}_i \gamma^*\right)\right)\right\|\\
    &\le n^{-1}\sup_{\left|t\right|\le C\left\|\gamma^*\right\|}|h\left(t\right)|\left(2C\left\|\gamma^*\right\|\sup_{\left|t\right|\leq 2C\left\|\gamma^*\right\|}|\eta^{\prime}\left(t\right)|+\sup_{\left|t\right|\leq C\left\|\gamma^*\right\|}|\eta\left(t\right)|\right) \sum_{i=1}^n \left\|T^{\top}_n \Tilde{g}_i-g_i\right\|\\
    &\le C_2\tau_n\sup_{\left|t\right|\le C\left\|\gamma^*\right\|}|h\left(t\right)|\left(2C\left\|\gamma^*\right\|\sup_{\left|t\right|\leq 2C\left\|\gamma^*\right\|}|\eta^{\prime}\left(t\right)|+\sup_{\left|t\right|\leq C\left\|\gamma^*\right\|}|\eta\left(t\right)|\right).
\end{aligned}
\end{equation*}
Then by the bound for $\left\|\gamma^*\right\|$ and for the continuous function in Remarks~\ref{remark 1} and \ref{remark 2}, 
\begin{equation}
\label{bn1}
    \left\|B_{1}\right\|
    \le C_2\left(6C^2+1\right)M^2\tau_n=o(n^{-1/2}).
\end{equation}
Similarly, for $B_{2}$ we have
\begin{equation}
\label{bn2}
    \begin{aligned}
    \left\|B_{2}\right\| 
    &\le n^{-1}\sum_{i=1}^{n} \left\|\left( T^{\top}_n \Tilde{g}_i \left(\eta\cdot h\right)\left(\tilde{g}^{\top}_i T_n\gamma^*\right)-{g}_i \left(\eta\cdot h\right)\left({g}^{\top}_i \gamma^*\right)\right)\right\|\\
    &\le n^{-1}\left(2C\left\|\gamma^*\right\|\sup_{\left|t\right|\leq 2C\left\|\gamma^*\right\|}|\left(\eta\cdot h\right)^{\prime}\left(t\right)|+\sup_{\left|t\right|\leq C\left\|\gamma^*\right\|}|\left(\eta\cdot h\right)\left(t\right)|\right) \sum_{i=1}^n \left\|T^{\top}_n \Tilde{g}_i-g_i\right\|\\
    &\le C_2\tau_n\left(2C\left\|\gamma^*\right\|\sup_{\left|t\right|\leq 2C\left\|\gamma^*\right\|}|\left(\eta\cdot h\right)^{\prime}\left(t\right)|+\sup_{\left|t\right|\leq C\left\|\gamma^*\right\|}|\left(\eta\cdot h\right)\left(t\right)|\right)\\
    &\le C_2\left(6C^2+1\right)M\tau_n=o(n^{-1/2}).
\end{aligned}
\end{equation}
By \eqref{bn1}, \eqref{bn2}, we have \eqref{estimatingequinequality_N}.

Next, we prove the convergence in probability for the random part:
\begin{equation*}
      T_n^{\top} \Tilde{S}\left(T_n \gamma^*\right)-S\left(\gamma^*\right)-\mathbb{E}\left(T_n^{\top} \Tilde{S}\left(T_n \gamma^*\right)-S\left(\gamma^*\right)\right)=n^{-1}\sum_{i=1}^n \left( T_n^{\top} \tilde{g}_i \eta\left(\tilde{g}_i^{\top} T_n \gamma^*\right)- {g}_i \eta\left({g}_i^{\top} \gamma^*\right)\right)e_i,
\end{equation*}
The $j$th element of it is
\begin{equation*}
    \left(n^{-1}\sum_{i=1}^n \left( T_n^{\top} \tilde{g}_i \eta\left(\tilde{g}_i^{\top} T_n \gamma^*\right)- {g}_i \eta\left({g}_i^{\top} \gamma^*\right)\right)e_i\right)_j:= n^{-1}\sum_{i=1}^n \left( T_n^{\top} \tilde{g}_i \eta\left(\tilde{g}_i^{\top} T_n \gamma^*\right)- {g}_i \eta\left({g}_i^{\top} \gamma^*\right)\right)_{j}e_i.
\end{equation*}
\noindent We will apply Lemma~\ref{techniquelemma1} to the right-hand side of the above equation with $$U_i=e_i, \quad a_i=n^{1/2}(T_n^{\top} \tilde{g}_i \eta\left(\tilde{g}_i^{\top} T_n \gamma^*\right)- {g}_i \eta\left({g}_i^{\top} \gamma^*\right))_j.$$ 
To this end, we need to verify the three conditions in Lemma~\ref{techniquelemma1}. Condition $\mathbb{E}[|U_i|^{1+\zeta}]<\infty$ holds with $\xi=1$ due to  Lemma~\ref{proposition 2}:
$$
\mathbb{E}[|U_i|^2]=\mathbb{E}[e^2_i]=v(g^{\top}_i \gamma^*)\leq \sup_{\left|t\right|\le C\left\|\gamma^*\right\|}v\left(t\right) \le M.
$$
Condition $\max_{1\le i \le n}|a_i|=o(n)$
holds because
\begin{eqnarray*}
  n^{-1/2}\max_{1\leq i\leq n}|a_i|&\le&  \max_{1\leq i\leq n}\left\|T_n^{\top} \tilde{g}_i \eta\left(\tilde{g}_i^{\top} T_n \gamma^*\right)\right\| +  \max_{1\leq i\leq n}\left\|{g}_i \eta\left({g}_i^{\top} \gamma^*\right)\right\|\\
  &\le& \max_{1\le i\le n}\left\|T^{\top}_n\tilde{g}_i\right\|\max_{1\le i\le n}\left|\eta\left(\tilde{g}_i^{\top} T_n \gamma^*\right)\right|+\max_{1\le i\le n}\left\|{g}_i\right\|\max_{1\le i\le n}\left|\eta\left({g}_i^{\top} \gamma^*\right)\right|\\
  &\le&2C\max_{\left|t\right|\leq 2C\left\|\gamma^*\right\|}|\eta\left(t\right)|+C\max_{\left|t\right|\leq C\left\|\gamma^*\right\|}|\eta\left(t\right)|\\
  &\le& 3CM.
\end{eqnarray*}
Finally, condition $\sum_{i=1}^n \left|a_i\right|=O(n)$ holds due to Lemma~\ref{lemma 3}:
\begin{equation*}
    \begin{aligned}
\sum_{i=1}^n|a_i| &\le n^{1/2}C_2\tau_n\left(2C\left\|\gamma^*\right\|\sup_{\left|t\right| \leq n^{3/2}2C\left\|\gamma^*\right\|}|\eta^{\prime}\left(t\right)|+\sup_{\left|t\right|\leq C\left\|\gamma^*\right\|}|\eta\left(t\right)|\right)\\
    &\le n^{3/2}C_2\left(6C^2+1\right)M\tau_n=o(n).
\end{aligned}
\end{equation*}
Therefore, by Lemma~\ref{techniquelemma1}, we have
\begin{equation*}
    n^{-1}\sum_{i=1}^n \left( T_n^{\top} \tilde{g}_i \eta\left(\tilde{g}_i^{\top} T_n \gamma^*\right)- {g}_i \eta\left({g}_i^{\top} \gamma^*\right)\right)_je_i=o_p(n^{-1/2}),
\end{equation*}
which implies that
\begin{equation}
\label{randompart}
\left\| T_n^{\top} \Tilde{S}\left(T_n \gamma^*\right)-S\left(\gamma^*\right)-\mathbb{E}\left(T_n^{\top} \Tilde{S}\left(T_n \gamma^*\right)-S\left(\gamma^*\right)\right)\right\|=o_p(n^{-1/2}). 
\end{equation}  
By \eqref{randompart} and \eqref{estimatingequinequality_N},
\begin{align*}    \left\|T_n^{\top}\Tilde{S}\left(T_n\gamma^*\right)-S\left(\gamma^*\right)\right\| &\leq \left\| T_n^{\top} \Tilde{S}\left(T_n \gamma^*\right)-S\left(\gamma^*\right)-\mathbb{E}\left(T_n^{\top} \Tilde{S}\left(T_n \gamma^*\right)-S\left(\gamma^*\right)\right)\right\|\\
&+\left\|\mathbb{E}\left(T_n^{\top}\Tilde{S}\left(T_n\gamma^*\right)-S\left(\gamma^*\right)\right)\right\|\\
    &=o_p(n^{-1/2})+o(n^{-1/2})\\
    &=o_p(n^{-1/2}).
\end{align*}
The proof is completed.
\end{proof}

\begin{lemma}[Bound on the gradient of the score function]
\label{lemma 6}
Under Assumptions \ref{cond:A1} to \ref{cond:A4}, 
\begin{equation*}
\sup _{\gamma \in N_n(\delta_0)}\left\| T^{\top}_n(\partial \Tilde{S}(\gamma) / \partial \gamma^{\top})T_n -T_n^{\top}\tilde{F}\left(T_n \gamma^*\right)T_n \right\|=o(1),
\end{equation*}
where
\begin{equation*}
    N_n(\delta_0) := \left\{\gamma:\|T^{-1}_n\gamma-\gamma^*\| \le (2C/n)^{1/2}\delta_0 \right\}.
\end{equation*}
\end{lemma}

\begin{proof}[Proof of Lemma~\ref{lemma 6}]
Using the definition of $\Tilde{S}$ in \eqref{estimating equation}, we rewrite $T^{\top}_n(\partial \Tilde{S}(\gamma) / \partial {\gamma}^{\top})T_n$ as follows: 
\begin{equation}
\label{Q}
T^{\top}_n(\partial \Tilde{S}(\gamma) / \partial {\gamma}^{\top})T_n=-T^{\top}_n \tilde{F}\left(\gamma\right) T_n + \frac{1}{n} \sum_{i=1}^n \eta^{\prime}\left(\Tilde{g}_i^{\top} \gamma\right) \left(y_i-h\left(\Tilde{g}_i^{\top}\gamma\right)\right) T^{\top}_n \tilde{g}_i \tilde{g}^{\top}_iT_n,
\end{equation}
where $\eta= h^{\prime}/v$ is a scalar function depending on functions $h$ and $v$ in the definition of $\tilde{S}$. We will show that the first term on the right-hand side of \eqref{Q} is close to its population counterpart
$T_n^{\top}\tilde{F}\left(T_n \gamma^*\right)T_n$ while the second term is negligible.

We proceed to prove the first part of the claim above. According to \eqref{tildeF}, we have 
$$
\tilde{F}\left(\gamma\right)
=\frac{1}{n}\sum_{i=1}^n(\eta\cdot h')(\tilde{g}_i^\top\gamma)\tilde{g}_i\tilde{g}_i^\top.
$$
Therefore, 
\begin{eqnarray*}
    T_n^{\top}\tilde{F}(\gamma)T_n-T_n^{\top}\tilde{F}\left(T_n \gamma^*\right)T_n &=&\frac{1}{n}\sum_{i=1}^n\Big[\left(\eta\cdot h^{\prime}\right)\left(\Tilde{g}^{\top}_i \gamma\right)-\left(\eta\cdot h^{\prime}\right)\left(\Tilde{g}^{\top}_i T_n\gamma^*\right)\Big]T_n^{\top}\tilde{g}_i\tilde{g}_i^{\top}T_n \\
&=&\frac{1}{n}\sum_{i=1}^n
\big[ \varphi_i\left(1\right)-\varphi_i\left(0\right)\big]
T_n^{\top}\tilde{g}_i\tilde{g}_i^{\top}T_n, \end{eqnarray*}
where
\begin{equation*}
\varphi_i\left(t\right):= \left(\eta\cdot h^{\prime}\right)\Big(t\tilde{g}^{\top}_i\gamma+(1-t)\tilde{g}^{\top}_iT_n\gamma^*\Big).
\end{equation*}
By the Mean Value Theorem,
there exist $t_i\in[0,1]$ and $\bar{\gamma}_i = t_i\gamma + (1-t_i)T_n\gamma^*$ such that
\begin{equation*}
    \varphi_i(1)-\varphi_i(0)=\varphi^{\prime}_i(t_i)=\Tilde{g}^{\top}_i \left(\gamma-T_n\gamma^*\right)\left(\eta \cdot h^{\prime}\right)^{\prime}\left(\Tilde{g}^{\top}_i \bar{\gamma}_i\right).
\end{equation*}
Substituting this into the equation above, we get
\begin{eqnarray*}
    T_n^{\top}\tilde{F}(\gamma)T_n-T_n^{\top}\tilde{F}\left(T_n \gamma^*\right)T_n  &=&\frac{1}{n}\sum_{i=1}^n
    \Tilde{g}^{\top}_i \left(\gamma-T_n\gamma^*\right)\left(\eta \cdot h^{\prime}\right)^{\prime}\left(\Tilde{g}^{\top}_i \bar{\gamma}_i\right)
T_n^{\top}\tilde{g}_i\tilde{g}_i^{\top}T_n.
\end{eqnarray*}
We now bound the terms on the right-hand side. First, by \eqref{boundness} and the definition of $N_n(\delta_0)$, we have 
\begin{eqnarray*}
 \left|\Tilde{g}^{\top}_i \left(\gamma-T_n\gamma^*\right)\right|
\le \left\|\Tilde{g}_i^\top T_n\right\| \left\|T^{-1}_n(\gamma-T_n\gamma^*)\right\|
    \le  2C \left\|T^{-1}_n\gamma-\gamma^*\right\| = O(n^{-1/2}).
\end{eqnarray*}
Next, by \eqref{Nndelta} and  \eqref{covariate bound tilde},  \begin{eqnarray*}   \left|\Tilde{g}^{\top}_i \bar{\gamma}_i\right|
&\le&\left\|\Tilde{g}^{\top}_i T_n\right\|\cdot\|T_n^{-1}\bar{\gamma}_i\| = 
\left\|T^{\top}_n \tilde{g}_n\right\|\cdot\left\|t_iT_n^{-1}\gamma+\left(1-t_i\right)\gamma^*\right\|\\ &\le& 2C\left(t_i\left\|T_n^{-1}\gamma\right\|+(1-t_i)\left\|\gamma^*\right\|\right).
\end{eqnarray*}
Since $\gamma\in N_n(\delta_0)$, it follows that for sufficiently large $n$, 
\begin{equation}
\label{eq:T inverse gamma}
    \|T^{-1}_n\gamma\|\le \|T^{-1}_n\gamma-\gamma^{*}\|+\left\|\gamma^*\right\|\le (2C/n)^{1/2}\delta_0+\left\|\gamma^*\right\| \le 2\left\|\gamma^*\right\|.
\end{equation}
Therefore, the last two bounds imply 
$
\left|\Tilde{g}^{\top}_i \bar{\gamma}_i\right| \le  4C\|\gamma^*\|
$. Since $(\eta\cdot h')'$ is a smooth function, we obtain that $|(\eta\cdot h')'(\Tilde{g}^{\top}_i \bar{\gamma}_i)|$ is uniformly bounded over $1\le i \le n$ and $\gamma\in N_n(\delta_0)$. 
Finally, by \eqref{covariate bound tilde} and for sufficiently large $n$, we have
\begin{equation}
\label{bounded_term}
\left\|T_n^{\top}\tilde{g}_i\tilde{g}_i^{\top}T_n \right\|=
\left\|T_n^{\top}\tilde{g}_i\right\|^2 \le 4C^2.
\end{equation}
Putting these inequalities together, we obtain  
\begin{equation*}
    \max_{\gamma \in N_n(\delta_0)}\|    T_n^{\top}\tilde{F}(\gamma)T_n-T_n^{\top}\tilde{F}\left(T_n\gamma^*\right)T_n\|=O(n^{-1/2}).
\end{equation*}

We now show that the second term on the right-hand side of \eqref{Q} is negligible. To this end, we decompose it as $\Phi_1-\Phi_2-\Phi_3$, where
\begin{eqnarray*}
    \Phi_1
    &=& \frac{1}{n} \sum_{i=1}^n \eta^{\prime}\left(\Tilde{g}_i^{\top}\gamma\right)\left(h\left(\Tilde{g}_i^{\top}\gamma\right)-h\left(g_i^{\top}\gamma^*\right)\right)T^{\top}_n \Tilde{g}_i \Tilde{g}_i^{\top}T_n,\\
    \Phi_2&=&\frac{1}{n}\sum_{i=1}^n \eta^{\prime}\left(\Tilde{g}_i^{\top}T_n\gamma^*\right)e_iT^{\top}_n \Tilde{g}_i\Tilde{g}_i^{\top}T_n,\\
    \Phi_3&=&\frac{1}{n}\sum_{i=1}^n \left(\eta^{\prime}\left(\Tilde{g}_i^{\top}\gamma\right)-\eta^{\prime}\left(\Tilde{g}_i^{\top}T_n\gamma^*\right)\right)e_iT^{\top}_n \Tilde{g}_i\Tilde{g}_i^{\top}T_n.
\end{eqnarray*}
Note that $\Phi_1$ is the expectation of the second term on the right-hand side of \eqref{Q}, while $-\Phi_2-\Phi_3$ is its centered version, where $\Phi_2$  does not depend on $\gamma$ and $\Phi_3$ depends on $\gamma$. We will bound $\Phi_1,\Phi_2$, and $\Phi_3$ separately. 
Regarding $\Phi_1$, by \eqref{bounded_term} and the triangle inequality, 
\begin{equation*}
\begin{aligned}
    \left\|\Phi_1\right\|
        &\le \frac{1}{n}\sum_{i=1}^n  \left|\eta^{\prime}\left(\tilde{g}^{\top}_i\gamma\right)\right|\left|h\left(\Tilde{g}_i^{\top}\gamma\right)-h\left(g_i^{\top}\gamma^*\right)\right|\|T^{\top}_n \Tilde{g}_i \Tilde{g}_i^{\top}T_n\|\\
        &\le \frac{4C^2}{n}\sum_{i=1}^n  \left|\eta^{\prime}\left(\tilde{g}^{\top}_i\gamma\right)\right|\left|h\left(\Tilde{g}_i^{\top}\gamma\right)-h\left(g_i^{\top}\gamma^*\right)\right|\\
    & \le \frac{4C^2}{n} \sum_{i=1}^n  \left|\eta^{\prime}\left(\tilde{g}^{\top}_i\gamma\right)\right|\Big\{\left|h\left(\Tilde{g}_i^{\top}\gamma\right)-h\left(g_i^{\top}T_n^{-1}\gamma\right)\right|+\left|h\left(g_i^{\top}T_n^{-1}\gamma\right)-h\left(g_i^{\top}\gamma^*\right)\right|\Big\}.
\end{aligned}
\end{equation*}
We now bound the terms on the right-hand side of the inequality above.
By the Mean Value Theorem, there exists $\tilde{t}_i\in[0,1]$ and $\tilde{\gamma}_i = \Tilde{t}_i T^{-1}_n\gamma + (1-\tilde{t}_i)\gamma^*$ such that
\begin{equation*}   \left|h\left(g_i^{\top}T_n^{-1}\gamma\right)-h\left(g_i^{\top}\gamma^*\right)\right|=\left|h^{\prime}({g}_i^{\top}\tilde{\gamma}_i)g^{\top}_i(T^{-1}_n\gamma-\gamma^*)\right|\le  \left|h^{\prime}\left({g}_i^{\top}\tilde{\gamma}_i\right)\right|\|g_i\|\left\|T^{-1}_n\gamma-\gamma^*\right\|. 
\end{equation*}
By Assumption~\ref{cond:A2}, \eqref{eq:T inverse gamma}, and \eqref{boundness}, we have
\begin{eqnarray*}    
\left|{g}_i^{\top}\tilde{\gamma}_i\right|
&\le& \Tilde{t}_i\left|g^{\top}_iT^{-1}_n\gamma\right|+ (1-\tilde{t}_i)\left|g^{\top}_i\gamma^*\right|
\le \Tilde{t}_i \left\|g_i\right\|\left\|T^{-1}_n\gamma\right\| + (1-\tilde{t}_i)\left\|g_i\right\|\left\|\gamma^*\right\| \\
&\le&
2\Tilde{t}_i C\left\|\gamma^*\right\| + 
(1-\tilde{t}_i)C\left\|\gamma^*\right\|\\
&\le&2 C\left\|\gamma^*\right\|.
\end{eqnarray*}
Since $h'$ is a smooth function, this implies that $\left|h^{\prime}\left({g}_i^{\top}\tilde{\gamma}_i\right)\right|$ is uniformly bounded over $1\le i \le n$. Moreover, $\|g_i\|\le C$ by Assumption~\ref{cond:A2} and $\left\|T^{-1}_n\gamma-\gamma^*\right\|\le (2C/n)^{1/2}\delta_0$ because $\gamma\in N_n(\delta_0)$.  
Therefore,
$$
\left|h\left(g_i^{\top}T_n^{-1}\gamma\right)-h\left(g_i^{\top}\gamma^*\right)\right| = O(n^{-1/2}).
$$
Similarly, there exist
$\bar{t}_i\in[0,1]$ and $\bar{g}_i = \bar{t}_i T^\top_n\Tilde{g}_i + (1-\bar{t}_i)g_i$ such that
\begin{eqnarray*}    \left|h\left(\Tilde{g}_i^{\top}\gamma\right)-h\left(g_i^{\top}T_n^{-1}\gamma\right)\right|&=&\left|h^{\prime}(\bar{g}_i^{\top}T^{-1}_n\gamma)(\Tilde{g}^{\top}_i T_n-g_i^{\top})T^{-1}_n\gamma\right|\\
&\le&  \left|h^{\prime}(\bar{g}_i^{\top}T^{-1}_n\gamma)\right|\left\|T^{\top}_n\Tilde{g}_i-g_i\right\|\|T^{-1}_n\gamma\|.
\end{eqnarray*}
By \eqref{boundness}, Assumption~\ref{cond:A4}, and \eqref{eq:T inverse gamma},  
$$
\left\|T_n^{\top} \tilde{g}_i - {g}_i\right\|\le C_2 n^{1/2}  \tau_n = o(1), \quad 
    \|T^{-1}_n\gamma\|\le  2\left\|\gamma^*\right\|.
$$
Again, by \eqref{boundness}, for sufficiently large $n$,
\begin{equation*}    \left\|\bar{g}_i\right\|=\left\|{g}_i+\bar{t}_i (T^\top_n\Tilde{g}_i-g_i)\right\|\le \left\|g_i\right\|+\bar{t}_i\left\| T^\top_n\Tilde{g}_i-g_i\right\|\le C+\bar{t}_iC_2n^{1/2}\tau_n\le 2C.
\end{equation*}
Since $h'$ is a smooth function, we obtain
$$
\left|h\left(\Tilde{g}_i^{\top}\gamma\right)-h\left(g_i^{\top}T_n^{-1}\gamma\right)\right| = o(1).
$$
Finally, by  \eqref{eq:T inverse gamma} and \eqref{bounded_term}, 
$$
|\tilde{g}_i^\top\gamma| \le \|T_n^\top\tilde{g}_i\|\cdot\|T_n^{-1}\gamma\| \le 4C\|\gamma^*\|,
$$
which, together with the smoothness of $\eta'$, implies that $|\eta'(\tilde{g}_i^\top\gamma)|$ is uniformly bounded over $1\le i\le n$ and $\gamma\in N_n(\delta_0)$.  
Putting these inequalities together, we obtain  $\|\Phi_1\|=o(1)$. 

Next, we show that $\Phi_2$ is negligible by bounding its entries ($\Phi_2$ is a matrix with a bounded number of entries). For $1\le s,t\le p+K-r$, we have
\begin{eqnarray*}
\left(\Phi_2\right)_{st}=\frac{1}{n}\sum_{i=1}^n \eta^{\prime}\left(\Tilde{g}_i^{\top}T_n\gamma^*\right)\left(T^{\top}_n \Tilde{g}_i \Tilde{g}_i^{\top}T_n\right)_{st}e_i.
\end{eqnarray*}
By \eqref{bounded_term}, we have  
\begin{equation*}
     \left|\left(T^{\top}_n \Tilde{g}_i \Tilde{g}_i^{\top}T_n\right)_{st}\right| \le 4C^2, \qquad   
    \left|\Tilde{g}_i^\top T_n\gamma^*\right| \le \left\|T^{\top}_n \Tilde{g}_i \right\| \left\|\gamma^*\right\|\le 2C\|\gamma^*\|.
\end{equation*}
Since $\eta'$ is a continuous function, it follows that the coefficients $\eta^{\prime}\left(\Tilde{g}_i^{\top}T_n\gamma^*\right)\left(T^{\top}_n \Tilde{g}_i \Tilde{g}_i^{\top}T_n\right)_{st}$
in the formula for $(\Phi_2)_{st}$ are uniformly bounded over $1\le i\le n$. Note also that $e_i$, $1\le i\le n$, are independent mean-zero random variables with variances $v(g_i^\top\gamma^*)$. These variances are uniformly bounded over $1\le i\le n$ because $|g_i^\top \gamma^*|\le \|g_i\|\|\gamma^*\|\le C\|\gamma^*\|$ by Assumption~\ref{cond:A2} and $v$ is a smooth function.  
Therefore, by Markov inequality, for any $t>0$, 
\begin{eqnarray*}
    \mathbb{P}(\left(\Phi_2\right)_{st}\geq t) \le (tn)^{-2}\sum_{i=1}^n \left(\eta^{\prime}\left(\Tilde{g}_i^{\top}T_n\gamma^*\right)\right)^2 \left(\left(T^{\top}_n \Tilde{g}_i \Tilde{g}_i^{\top}T_n\right)_{st}\right)^2v(g_i^\top\gamma^*) =O(t^{-2}n^{-1}).
\end{eqnarray*}
Choosing $t=o(n^{-1/2})$, we obtain $\|\Phi_2\|=o_p(1)$.  

Finally, we bound $\Phi_3$. 
By the Mean Value Theorem, there exists $\tilde{t}_i\in[0,1]$ and $\breve{\gamma}_i = \Tilde{t}_i T^{-1}_n\gamma + (1-\tilde{t}_i)\gamma^*$ such that
\begin{equation*}  \eta^{\prime}\left(\Tilde{g}_i^{\top}\gamma\right)-\eta^{\prime}\left(\Tilde{g}_i^{\top}T_n\gamma^*\right)=\eta^{\prime\prime}\left(\Tilde{g}_i^{\top}T_n\Breve{\gamma}_i\right)\Tilde{g}_i^{\top}T_n \left(T^{-1}_n\gamma-\gamma^*\right).
\end{equation*}
By \eqref{eq:T inverse gamma} and \eqref{bounded_term}, 
\begin{eqnarray*}    
\left|\tilde{g}_i^{\top}T_n\breve{\gamma}_i\right|&\le& \Tilde{t}_i \left|\tilde{g}_i^{\top}\gamma\right| + (1-\tilde{t}_i)\left|\tilde{g}_i^{\top}T_n\gamma^*\right|\\
&\le& \Tilde{t}_i \left\|T^{\top}_n\tilde{g}_i\right\|\left\|T^{-1}_n\gamma\right\| + (1-\tilde{t}_i)\left\|T^{\top}_n\tilde{g}_i\right\|\left\|\gamma^*\right\|\\
&\le&
4\Tilde{t}_i C\left\|\gamma^*\right\|+2(1-\tilde{t}_i)C\left\|\gamma^*\right\|\\
&\le&4C\left\|\gamma^*\right\|.
\end{eqnarray*}
Since $\eta^{\prime\prime}$ is a smooth function, this implies that $\eta^{\prime\prime}(\tilde{g}_i^{\top}T_n\breve{\gamma}_i)$ is uniformly bounded over $1\le i\le n$ and $\gamma\in N_n(\delta_0)$ by a constant $M>0$.
Therefore by \eqref{bounded_term} and the definition of $N_n(\delta_0)$,
\begin{align*}
    \| \Phi_3 \|    &=\Big\|\frac{1}{n}\sum_{i=1}^n \eta^{\prime\prime}\left(\Tilde{g}_i^{\top}T_n\Breve{\gamma}_i\right)\Tilde{g}_i^{\top}T_n \left(T^{-1}_n\gamma-\gamma^*\right)e_iT^{\top}_n \Tilde{g}_i\Tilde{g}_i^{\top}T_n\Big\| \\
    & \le \frac{8C^3M}{n}\sum_{i=1}^n \left\|T^{-1}_n\gamma-\gamma^*\right\|\left|e_i\right| \\
    &\leq (8C^3M)(2C/n)^{1/2}\delta_0\left(\frac{1}{n}\sum_{i=1}^n \left|e_i\right|\right).
\end{align*}
Since the variances of $e_i$ are uniformly bounded over $1\le i\le n$ (see the argument for $\Phi_2$ above),  it follows that 
\begin{eqnarray*}
        \mathbb{E}\left[\sup _{\gamma \in N_n(\delta_0)} \left\| \Phi_3\right\|\right]&\leq& O(n^{-1/2})\cdot \mathbb{E}\left[\frac{1}{n}\sum_{i=1}^n \left|e_i\right|\right] 
    \le O(n^{-1/2})  \cdot\max_{1\le i \le n}\mathbb{E} \left|e_i\right|\\
    &\le &O(n^{-1/2})  \cdot\max_{1\le i \le n}\sqrt{\mathbb{E} \left|e_i\right|^2}  \rightarrow 0.
\end{eqnarray*}
In turns, this implies $\sup _{\gamma \in N_n(\delta_0)} \|\Phi_3\|=o_p(1)$ by the Markov inequality.
The proof is complete. 
\end{proof}

We also need the following lemma for the proof of Lemma~\ref{lemma 7}.

\begin{lemma}[Lemma 2 in \cite{yin2006asymptotic}]
    \label{lemma 9}
    Let $f:G\subset\mathbb{R}^q\to\mathbb{R}^q$ be a function with $f(x)=(f_1(x),...,f_q(x))^\top$ such that $f_1,...,f_q$ are continuously differentiable on the convex set $G$.  
    Then for any $\alpha,\beta \in G$, 
$$
f(\beta)-f(\alpha)=\left(\int_0^1 \frac{\partial f(\alpha+t(\beta-\alpha))}{\partial x} d t\right)(\beta-\alpha),
$$
where the integral is taken element-wise.
\end{lemma}

\begin{proof}[Proof of Lemma~\ref{lemma 7}]
We first prove \eqref{boundary}. By Lemma~\ref{lemma 9}, 
\begin{equation}
\label{Qexp}
    T^{\top}_n \Tilde{S}(\gamma)-T_n^{\top} \Tilde{S}\left(T_n \gamma^*\right)=\mathbb{H}(\gamma)(T^{-1}_n\gamma-{\gamma^*}),
\end{equation}
where for notation simplicity, we denote 
$$\mathbb{H}(\gamma)= \int_0^1 H\left(\gamma^*+t\left(T^{-1}_n\gamma-\gamma^*\right)\right) d t, \qquad
H(\gamma)=T^{\top}_n(\partial \Tilde{S}(\gamma) / \partial {\gamma}^{\top})T_n.$$
We show that $\mathbb{H}(\gamma)$ is well-conditioned. For that purpose, we decompose $\mathbb{H}=\Phi_1+\Phi_2$,
where $\Phi_1= T_n^{\top}\tilde{F}\left(T_n \gamma^*\right)T_n$ and $\Phi_2=\mathbb{H}(\gamma)-\Phi_1$.
From Corollary~\ref{cor:scaling matrix}, 
$$
\lambda_{\min }(\Phi_1)=\lambda_{\min }\left(T_n^{\top}\tilde{F}\left(T_n\gamma^*\right)T_n\right) \ge \frac{1}{2C}.
$$
Regarding $\Phi_2$, by Lemma~\ref{lemma 6},
\begin{eqnarray}
\label{thm1:approximationtarget}
\left\| \Phi_2\right\|
&=&\left\| \int_0^1 \left[H\left(\gamma^*+t\left(T^{-1}_n\gamma-\gamma^*\right)\right)-T_n^{\top}\tilde{F}\left(T_n \gamma^*\right)T_n\right] d t\right\|\\ 
\nonumber
&\le& \int_0^1 \left\| H\left(\gamma^*+t\left(T^{-1}_n\gamma-\gamma^*\right)\right) -T_n^{\top}\tilde{F}\left(T_n \gamma^*\right)T_n\right\|d t\\
\nonumber
&=&o_p(1).    
\end{eqnarray}
For any $\gamma\in\partial N_n(\delta_0)$ we have $\|T_n^{-1}\gamma-\gamma^*\|=\delta_0n^{-1/2}$. Therefore, by \eqref{Qexp}, 
\begin{eqnarray*}
    \left\|T^{\top}_n\Tilde{S}(\gamma)-T^{\top}_n\Tilde{S}\left(T_n\gamma^*\right)\right\| 
\ge \lambda_{\min}(\mathbb{H}(\gamma))\delta_0n^{-1/2}\ge \left(\frac{1}{2C}+o_p(1)\right)\delta_0n^{-1/2},
\end{eqnarray*}
and \eqref{boundary} is proved.

We now prove \eqref{gamma0}. By the triangle inequality,
$$
\left\|T^{\top}_n\Tilde{S}\left(T_n\gamma^*\right)\right\| \le \left\|S\left(\gamma^*\right)\right\| + \left\|T^{\top}_n\Tilde{S}\left(T_n\gamma^*\right)-S\left(\gamma^*\right)\right\|.
$$
The second term on the right-hand side of the above inequality is negligible because by \eqref{estimatingequinequality}, 
\begin{equation*}
\left\|T_n^{\top}\Tilde{S}\left(T_n\gamma^*\right)-S\left(\gamma^*\right)\right\| =o_p(n^{-1/2}).
\end{equation*}
It remains to bound $\|S(\gamma^*)\|$. Note that it follows directly from the definition of $S(\gamma^*)$ in \eqref{estimating equation real} that $\mathbb{E}[S\left(\gamma^*\right)]=0$ and $\operatorname{Cov}(S\left(\gamma^*\right))=n^{-1}F\left(\gamma^*\right)$ is a square matrix of size $(p+K-r)$. Therefore, by Markov inequality and Assumption~\ref{cond:A3}, for any $t>0$,
\begin{eqnarray*}
\mathbb{P}\left( \left\|S\left(\gamma^*\right)\right\| > t\right)
&\leq& t^{-2}\mathbb{E} \left\|S\left(\gamma^*\right)\right\|^2 = t^{-2}\mathbb{E} \left[\text{Trace}\left( S^{\top}\left(\gamma^*\right) {S}\left(\gamma^*\right)\right)\right]\\
&=&t^{-2}\mathbb{E} \left[\text{Trace}\left( S\left(\gamma^*\right) {S}^{\top}\left(\gamma^*\right)\right)\right] = t^{-2}\text{Trace}\left(\mathbb{E}\left[S\left(\gamma^*\right) {S}^{\top}\left(\gamma^*\right)\right]\right)\\
&=&t^{-2}n^{-1}\text{Trace}\left(F(\gamma^*)\right) \le t^{-2}n^{-1}(p+K-r)C.
\end{eqnarray*}
Choosing $t=O((\varepsilon n)^{-1/2})$, we obtain $\|S(\gamma^*)\|=O(n^{-1/2})$ with probability at least $1-\varepsilon$. The proof is complete.
\end{proof}

\section{The Proof of Theorem~\ref{theorem 2}}
\label{app: asym dist}

We prove Theorem~\ref{theorem 2} in this section. We will use the notations and results in the proof of Theorem~\ref{existence+consistency}. The following lemma is crucial for proving Theorem~\ref{theorem 2}. 

\begin{lemma}[Asymptotic Approximation]
\label{lem:asymp approx}
    Assume that the conditions of Theorem~\ref{theorem 2} hold. Then 
\begin{eqnarray*}
T_n^{-1}\hat{\gamma}-{\gamma}^*
=F^{-1}(\gamma
^*)S\left(\gamma^*\right) + o_p\big(n^{-1/2}\big).
\end{eqnarray*}
\end{lemma}

\begin{proof}[Proof of Lemma~\ref{lem:asymp approx}]
From \eqref{Qexp} and the fact that $\hat{\gamma}$ is a solution of the estimating equation, we have
\begin{equation*}
    0 = T^{\top}_n \Tilde{S}(\hat{\gamma})=\mathbb{H}(\hat{\gamma})(T^{-1}_n\hat{\gamma}-{\gamma^*})+T_n^{\top} \Tilde{S}\left(T_n \gamma^*\right),
\end{equation*}
where we recall that 
$$\mathbb{H}(\gamma)= \int_0^1 H\left(\gamma^*+t\left(T^{-1}_n\gamma-\gamma^*\right)\right) d t, \qquad
H(\gamma)=T^{\top}_n(\partial \Tilde{S}(\gamma) / \partial {\gamma}^{\top})T_n.$$
  From the equality above,
\begin{align*}    T_n^{-1}\hat{\gamma}-{\gamma}^*
=\mathbb{H}^{-1}(\hat{\gamma})T_n^{\top}\Tilde{S}\left(T_n\gamma^*\right).
\end{align*}
In light of Lemma~\ref{lemma 6} and particularly the bound \eqref{thm1:approximationtarget} in its proof, we will approximate $\mathbb{H}(\hat{\gamma})$ by $T_n^{\top}\tilde{F}\left(T_n \gamma^*\right)T_n$, and in turn by $F(\gamma^*)T_n$. Accordingly, we decompose the expression above as
\begin{eqnarray*}
T_n^{-1}\hat{\gamma}-{\gamma}^*
&=&\left(T_n^{\top}\tilde{F}\left(T_n \gamma^*\right)T_n\right)^{-1}T_n^{\top}\Tilde{S}\left(T_n\gamma^*\right) + \Phi_1\\
&=&F^{-1}(\gamma
^*)S\left(\gamma^*\right) + \Phi_2 + \Phi_1,
\end{eqnarray*}
where $\Phi_1$ and $\Phi_2$ are the errors of those approximations, namely,
\begin{eqnarray*}
    \Phi_1 &=& \left[\mathbb{H}^{-1}(\hat{\gamma})-\left(T_n^{\top}\tilde{F}\left(T_n \gamma^*\right)T_n\right)^{-1}\right]T_n^{\top}\Tilde{S}\left(T_n\gamma^*\right),\\
    \Phi_2 &=& \left(T_n^{\top}\tilde{F}\left(T_n \gamma^*\right)T_n\right)^{-1}T_n^{\top}\Tilde{S}\left(T_n\gamma^*\right) - F^{-1}(\gamma
^*)S\left(\gamma^*\right). 
\end{eqnarray*}

We proceed to bound $\Phi_1$ and then $\Phi_2$. By \eqref{gamma0}, we have
$\|T_n^{\top}\Tilde{S}\left(T_n\gamma^*\right)\|=O_p(n^{-1/2})
$. In addition, by Corollary~\ref{cor:scaling matrix}, all eigenvalues of $T_n^{\top}\tilde{F}\left(T_n \gamma^*\right)T_n$ belong to interval $[1/(2C),2C]$, and therefore are bounded away from zero and infinity. Moreover, by  
\eqref{thm1:approximationtarget}, 
$$
\left\|\mathbb{H}(\hat{\gamma})- T_n^{\top}\tilde{F}\left(T_n \gamma^*\right)T_n\right\| = o_p(1).
$$
These imply, in particular, that $\|\mathbb{H}^{-1}(\hat{\gamma})\|$ and $\|(T_n^{\top}\tilde{F}\left(T_n \gamma^*\right)T_n)^{-1}\|$ are both bounded by some constant due to the continuity of the inverse map away from zero.   
Therefore, 
\begin{eqnarray*}
\|\Phi_1\| &=& \left\|\mathbb{H}^{-1}(\hat{\gamma})\left[\mathbb{H}(\hat{\gamma})-T_n^{\top}\tilde{F}\left(T_n \gamma^*\right)T_n\right]\left(T_n^{\top}\tilde{F}\left(T_n \gamma^*\right)T_n\right)^{-1}T_n^{\top}\Tilde{S}\left(T_n\gamma^*\right)\right\| \\
&\le& \left\|\mathbb{H}^{-1}(\hat{\gamma})\right\|
\left\|\mathbb{H}(\hat{\gamma})-T_n^{\top}\tilde{F}\left(T_n \gamma^*\right)T_n\right\|
\left\|\left(T_n^{\top}\tilde{F}\left(T_n \gamma^*\right)T_n\right)^{-1}\right\|
\left\|T_n^{\top}\Tilde{S}\left(T_n\gamma^*\right)\right\| \\
&=& o_p(n^{-1/2}).    
\end{eqnarray*}
Next, we bound $\Phi_2$. 
By adding and subtracting $(T_n^{\top}\tilde{F}\left(T_n \gamma^*\right)T_n)^{-1}S(\gamma^*)$ and using the triangle inequality, we obtain that $\|\Phi_2\|\le \|\Phi_{21}\| + \|\Phi_{22}\|$, where By \eqref{eigen_inverse} and Lemma~\ref{lemma 5}, Lemma~\ref{lemma 4} 
\begin{eqnarray*}
    \Phi_{21} &=& \left(T_n^{\top}\tilde{F}\left(T_n \gamma^*\right)T_n\right)^{-1}\left(T^{\top}_n\Tilde{S}\left(T_n\gamma^*\right)-S\left(\gamma^*\right)\right),\\
    \Phi_{22}&=& \left[\left(T_n^{\top}\tilde{F}\left(T_n \gamma^*\right)T_n\right)^{-1}-F^{-1}(\gamma^*)\right]S(\gamma^*).
\end{eqnarray*}
By \eqref{eigen_inverse} and \eqref{estimatingequinequality}, 
$$
\|\Phi_{21}\| \le \left\|\left(T_n^{\top}\tilde{F}\left(T_n \gamma^*\right)T_n\right)^{-1}\right\|\left\|T^{\top}_n\Tilde{S}\left(T_n\gamma^*\right)-S\left(\gamma^*\right)\right\|\le 2C\cdot o_p(n^{-1/2}) = o_p(n^{-1/2}).
$$
To bound $\|\Phi_{22}\|$, note first that by Lemma~\ref{lemma 4} and Assumption~\ref{cond:A4},
\begin{eqnarray}
   \label{eq: F and F tilde}\left\|\left(T_n^{\top}\tilde{F}\left(T_n \gamma^*\right)T_n\right)-F(\gamma^*)\right\|\le \Psi(\|\gamma^*\|)\tau_n = o(n^{-1/2}).  
\end{eqnarray}
Also, by Assumption~\ref{cond:A3}, eigenvalues of $F(\gamma^*)$ belong to interval $[1/(2C),2C]$, therefore bounded away from zero and infinity. This implies that both $\|(T_n^{\top}\tilde{F}\left(T_n \gamma^*\right)T_n)^{-1}\|$ and $\|F^{-1}(\gamma^*)\|$ are bounded from above by an absolute constant, due to the continuity of the inverse map away from zero. Therefore,
\begin{eqnarray*}
\|\Phi_{22}\|&=& \left\|\left(T_n^{\top}\tilde{F}\left(T_n \gamma^*\right)T_n\right)^{-1}\left[T_n^{\top}\tilde{F}\left(T_n \gamma^*\right)T_n-F(\gamma^*)\right]F^{-1}(\gamma^*)S(\gamma^*)
\right\|\\
&\le&\left\|\left(T_n^{\top}\tilde{F}\left(T_n \gamma^*\right)T_n\right)^{-1}\right\|
\left\|T_n^{\top}\tilde{F}\left(T_n \gamma^*\right)T_n-F(\gamma^*)\right\|
\left\|F^{-1}(\gamma^*)\right\|
\left\|S(\gamma^*)
\right\|\\
&=&o(n^{-1/2})\left\|S(\gamma^*)
\right\|.
\end{eqnarray*}
From \eqref{estimating equation real}, we have
\begin{equation*}
S\left(\gamma^*\right)= \frac{1}{n}\sum_{i=1}^n {g}_i\eta\left({g}^{\top}_i\gamma^*\right) e_i,
\end{equation*}
where $\eta= h^{\prime}/v$ denotes a scalar function depending on functions $h$ and $v$, and $e_i$'s are independent random variables with zero means and variances $v(g_i^\top\gamma^*)$. By Markov inequality, for any $t>0$,
\begin{eqnarray*}
    \mathbb{P}(\|S(\gamma^*)\|>t) &\le& t^{-2} \mathbb{E}\|S(\gamma^*)\|^2=t^{-2}\frac{1}{n^2}\sum_{i=1}^n\|g_i\|^2\eta^2(g_i^\top\gamma^*)v(g_i^\top\gamma^*).
\end{eqnarray*}
Note that $\|g_i\|\le C$ by Assumption~\ref{cond:A2}. In particular, $\|g_i^\top\gamma^*\|\le C\|\gamma^*\|$, and therefore $|\eta^2(g_i^\top\gamma^*)v(g_i^\top\gamma^*)|$ are uniformly bounded over $1\le i\le n$ because $\eta^2v$ is a smooth function. This implies
$$
\mathbb{P}(\|S(\gamma^*)\|>t) = O(t^{-2}n^{-1}).
$$
Choosing $t=O(1)$, we obtain that $\|S(\gamma^*)\|=O_p(1)$. This implies $\|\Phi_{22}\|=o_p(n^{-1/2})$ and the proof is complete.     
\end{proof}

The following  Linderberg-Feller Central Limit Theorem is needed for proving Theorem~\ref{theorem 2}.

\begin{lemma}[Lindeberg-Feller Central Limit Theorem from \cite{Lindeberg1922}]
    \label{lemma 8} 
For each positive integer $n$, let $X_{n j}$, $j=1,2,...,n$, be independent random variables with $\mathbb{E} [X_{n j}]=0$ and  $\mathbb{E}[X_{n j}^2]=\sigma_{n j}^2<\infty$. Denote $B_n^2=\sum_{j=1}^n \sigma_{n j}^2$ and assume that for each $\varepsilon>0$,
\begin{equation*}
\lim_{n\to\infty}
    \frac{1}{B_n^2} \sum_{i=1}^n \mathbb{E}\left[X_{n j}^2 I\left\{\left|X_{n j}\right|>\varepsilon B_n\right\}\right] =0.
\end{equation*}
Then
\begin{equation*}
 \frac{1}{B_n} \sum_{j=1}^n X_{n j} \to N(0,1),   
\end{equation*}
where the convergence is in distribution.
\end{lemma}

\begin{proof}[Proof of Theorem \ref{theorem 2}] We fix a unit vector $z\in\mathbb{R}^{p+K-r}$ and derive the asymptotic distribution for the scalar random variable $z^{\top}(T_n^{-1}\hat{\gamma}-\gamma^*)$;  the claims in \eqref{gamma}, \eqref{beta}, and \eqref{theta} will follow from specific choices of $z$.  
By Lemma~\ref{lem:asymp approx}, we have
\begin{equation*}
   z^{\top}(T_n^{-1}\hat{\gamma}-\gamma^*)= z^{\top}F^{-1}\left(\gamma^*\right)S\left(\gamma^*\right)+o_p(n^{-1/2}).
\end{equation*}
Multiplying both sides of this equation with a normalizing factor $\sqrt{n}(z^{\top} F^{-1}\left(\gamma^*\right)  z)^{-1/2}$, which is of order $O(n^{1/2})$ by Assumption~\ref{cond:A3}, we obtain 
\begin{align*}
    \frac{\sqrt{n}z^{\top}(T_n^{-1}\hat{\gamma}-\gamma^*)}{ \left(z^{\top} F^{-1}\left(\gamma^*\right)z \right)^{1/2}}
    =\frac{n^{-1/2}\sum_{i=1}^n z^{\top}F^{-1}\left(\gamma^*\right){g}_i  \eta\left(g^{\top}_i\gamma^*\right)e_i}{ \left(z^{\top} F^{-1}\left(\gamma^*\right)z \right)^{1/2}} +o_p(1),
\end{align*}
where $\eta= h^{\prime}/v:\mathbb{R}\rightarrow \mathbb{R}$ is a scalar function. We will prove the asymptotic normality of the first term on the right-hand side of the above equation by verifying the Lindeberg conditions in Lemma~\ref{lemma 8}. 
Denote 
$$\bar{e}_i:= z^{\top}F^{-1}\left(\gamma^*\right)g_i \eta\left(g^{\top}_i\gamma^*\right)e_i.$$ 
It is straightforward that $\mathbb{E}[\bar{e}_i]=0$ and $\mathbb{E}[\bar{e}^2_i]<\infty$ because $g_i$'s are bounded by Assumption~\ref{cond:A2}.
By \eqref{F}, the sum of variances of $\bar{e}_i$ are
\begin{equation*}   B_n^2 = \sum_{i=1}^n\mathbb{E}\left[\bar{e}^2_i\right]=\sum_{i=1}^n z^{\top}F^{-1}\left(\gamma^*\right)g_ig^{\top}_i\frac{\left(h^{\prime}\left({g}^{\top}_i\gamma\right)\right)^{2}}{v\left({g}^{\top}_i\gamma\right)}F^{-1}\left(\gamma^*\right) z=n \left(z^{\top}F_{n}^{-1}\left(\gamma^*\right)z\right).
\end{equation*}
It remains to verify the tail control condition, that is, to show that for each $\varepsilon>0$, 
$$
\Phi := \frac{1}{B_n^2} \sum_{i=1}^n \mathbb{E}\left[\bar{e}_i^2 I\left\{\left|\bar{e}_i\right|>\varepsilon B_n\right\}\right] \to 0.
$$
Note that $B_n^2\ge C^{-1}n$ by Assumption~\ref{cond:A3}. Therefore, 
\begin{eqnarray*}
\Phi  &\leq& \frac{1}{B_n^2}\sum_{i=1}^n\mathbb{E}\left[\left|\bar{e}_i\right|^2 I\left\{|\bar{e}_i|>\varepsilon   C ^{-1/2}n^{1/2} \right\}\right]\\
&=& \frac{1}{B_n^2}\sum_{i=1}^n\mathbb{E}\left[\left|\bar{e}_i\right|^\xi \left(\left|\bar{e}_i\right|^{2-\xi}I\left\{|\bar{e}_i|>0\right\}\right) I\left\{|\bar{e}_i|>\varepsilon   C ^{-1/2}n^{1/2} \right\}\right]\\
&\leq& \frac{1}{B_n^2}\sum_{i=1}^n \mathbb{E}\left[|\bar{e}_i|^\xi(\varepsilon^{-1} C^{1/2}n^{-1/2})^{\xi-2}\right] \\
&\leq& \varepsilon^{2-\xi}C^{(\xi-2)/2} n^{-(\xi-2)/2} \left(  z^{\top} F^{-1}\left(\gamma^*\right)z  \right)^{-1} \max _{1 \leq i \leq n} \mathbb{E}|\bar{e}_i|^\xi.
\end{eqnarray*}
To bound $\max _{1 \leq i \leq n} \mathbb{E}|\bar{e}_i|^\xi$, note that the coefficient $z^{\top}F^{-1}(\gamma^*)g_i \eta(g^{\top}_i\gamma^*)$ in the definition of $\bar{e}_i$ is uniformly bounded over $1\le i\le n$ because $\|F^{-1}(\gamma^*)\|$ is bounded by Assumption~\ref{cond:A3}, $\|g_i\|\le C$ by Assumption~\ref{cond:A2}, and $\eta\left(g^{\top}_i\gamma^*\right)$ is bounded by the continuity of $\eta$ and the fact that $\|g^{\top}_i\gamma^*\|\le C\|\gamma^*\|$. Therefore, by Assumption~\ref{cond:A6}, 
\begin{eqnarray*}
   \max _{1 \leq i \leq n} \mathbb{E}|\bar{e}_i|^\xi = O\left(\max _{1 \leq i \leq n}\mathbb{E}|{e}_i|^\xi\right)=O(1).
\end{eqnarray*}
This implies $\Phi=O(n^{-(\xi-2)/2})$, and the tail control condition is proved. Therefore, 
by the Lemma~\ref{lemma 8},  
\begin{equation}
\label{thm2resultmid}
    \frac{\sqrt{n}z^{\top}(T_n^{-1}\hat{\gamma}-\gamma^*)}{  \left( z^{\top} F^{-1}\left(\gamma^*\right)z \right)^{1/2}} \rightarrow \mathcal{N}(0,1).
\end{equation}
Finally, we can replace the denominator of the left-hand side of \eqref{thm2resultmid} with the approximation $( v^{\top} (T_n^{\top}\tilde{F}(\hat{\gamma})T_n)^{-1} v )^{1/2}$, because of \eqref{eq: F and F tilde} and the fact that $F^{-1}(\gamma^*)$ is bounded from below by Assumption~\ref{cond:A3}. We conclude
\begin{equation}
\label{thm2result}
    \frac{\sqrt{n}z^{\top}(T_n^{-1}\hat{\gamma}-\gamma^*)}{  \left( z^{\top} \left(T_n^{\top}\tilde{F}(\hat{\gamma})T_n\right)^{-1} z\right)^{1/2}} \rightarrow \mathcal{N}(0,1).
\end{equation}

We now proceed to prove that  \eqref{gamma}, \eqref{beta}, and \eqref{theta} are consequences of  \eqref{thm2result} with proper choices of $z$.
Regarding \eqref{theta}, we choose $$z_{1:r}= \frac{n^{-1}{Z}_{1:r}^{\top} \hat{\mathcal{P}}_{R} X G u}{\|n^{-1}{Z}_{1:r}^{\top} \hat{\mathcal{P}}_{R}X G u\|}, \qquad z_{(r+1):(K+p-r)}=0.$$ 
For this choice to be valid, we need to check that the denominator in the formula of $z_{1:r}$ is not zero. Indeed, by condition \eqref{beta_asymptotic_condition}, 
\begin{eqnarray*}    
n^{-1}\left\|{Z}_{1:r}^{\top}\hat{\mathcal{P}}_{R} X G u\right\|&\geq& n^{-1}\left\|{Z}_{1:r}^{\top} X G u\right\|- n^{-1}\left\|{Z}_{1:r}^{\top}(\hat{\mathcal{P}}_{R}-\mathcal{P}_{R})X G u\right\|\\
&\geq& c-n^{-1}\left\|{Z}_{1:r}^{\top}\right\|\left\|\hat{\mathcal{P}}_{R}-\mathcal{P}_{R}\right\|\left\|XGu\right\|\\
&=&c-n^{-1/2}\left\|XGu\right\|\left\|\hat{\mathcal{P}}_{R}-\mathcal{P}_{R}\right\|\\
&=&c-\left(u^{\top}G\left(X^{\top}X/n\right)Gu\right)^{1/2}\left\|\hat{\mathcal{P}}_{R}-\mathcal{P}_{R}\right\|\\
&=&c-\left(u^{\top}Gu\right)^{1/2}\left\|\hat{\mathcal{P}}_{R}-\mathcal{P}_{R}\right\|\\
&\ge & c-C^{1/2}C_2\tau_n\\
&>& c/2.
\end{eqnarray*}
With this choice of $z$,  \eqref{thm2result} reduces to
\begin{equation*}
\label{hatgamma1}    \frac{\sqrt{n}v_{1:r}^{\top}\left(\left(\tilde{Z}^{\top}_{1:r} {Z}_{1:r}/n\right)^{-1}\hat{\gamma}_{1:r}-\gamma_{1:r}\right)}{\left(v_{1:r}^{\top}  \left(\tilde{Z}^{\top}_{1:r} {Z}_{1:r}/n\right)^{-1}\Tilde{F}^{-1}_1(\hat{\gamma})\left({Z}^{\top}_{1:r} \tilde{Z}_{1:r}/n\right)^{-1}v_{1:r} \right)^{1/2}} \rightarrow {\mathcal{N}}(0,1).
\end{equation*}
From \eqref{eq:par trans theta}, the above expression is simplified to
\begin{equation*}
    \frac{\sqrt{n}{u}^{\top}\left(\hat{{\theta}}-n^{-1}G^{\top} X^{\top}\hat{\mathcal{P}}_{{R}} Z_{1:r}\gamma_{1:r}\right)}{\left(u^{\top}  G X^{\top}  \tilde{Z}_{1:r} \Tilde{F}^{-1}_1(\hat{\gamma}) \tilde{Z}_{1:r}^{\top}  X G u \right)^{1/2}} \rightarrow \mathcal{N}(0,1),
\end{equation*}
and \eqref{theta} is proved. 

Next, we prove \eqref{beta}
by choosing $z$ such that 
$$z_{1:r}=0, \qquad z_{(r+1):p}= \frac{n^{-1}{Z}_{(r+1):p}^{\top} \hat{\mathcal{P}}_{R} X G u}{\|n^{-1}{Z}_{(r+1):p}^{\top} \hat{\mathcal{P}}_{R}X G u\|}, \qquad z_{(p+1):(K+p-r)}=0.$$ 
The denominator of the formula for $z_{(r+1):p}$ is non-zero because
by condition \eqref{beta_asymptotic_condition}, 
\begin{eqnarray*}    
n^{-1}\left\|{Z}_{(r+1):p}^{\top}\hat{\mathcal{P}}_{C} X G u\right\|\geq  c-C^{1/2}C_2\tau_n
> c/2.
\end{eqnarray*}
With this choice of $z$,  \eqref{thm2result} is equivalent to
\begin{equation*}
\label{hatgamma2}    \frac{\sqrt{n}v_{(r+1):p}^{\top}\left(\left(\tilde{Z}^{\top}_{(r+1):p} {Z}_{(r+1):p}/n\right)^{-1}\hat{\gamma}_{(r+1):p}-\gamma_{(r+1):p}\right)}{\left(v_{(r+1):p}^{\top}  \left(\tilde{Z}^{\top}_{(r+1):p} {Z}_{(r+1):p}/n\right)^{-1}\Tilde{F}^{-1}_1(\hat{\gamma})\left({Z}^{\top}_{(r+1):p} \tilde{Z}_{(r+1):p}/n\right)^{-1}v_{(r+1):p} \right)^{1/2}} \rightarrow {\mathcal{N}}(0,1).
\end{equation*}
From \eqref{eq:par trans beta}, the above expression is simplified to
\begin{equation*}
    \frac{\sqrt{n}{u}^{\top}\left(\hat{{\beta}}-n^{-1}G^{\top} X^{\top}\hat{\mathcal{P}}_{{C}} Z_{(r+1):p}\gamma_{(r+1):p}\right)}{\left(u^{\top}  G X^{\top}  \tilde{Z}_{(r+1):p} \Tilde{F}^{-1}_2(\hat{\gamma}) \tilde{Z}_{(r+1):p}^{\top}  X G u \right)^{1/2}} \rightarrow \mathcal{N}(0,1),
\end{equation*}
and \eqref{beta} is proved.

Finally, we show \eqref{gamma}. For any unit vector $u\in \mathbb{R}^{K-r}$, choose 
 $$z_{1:p}=0,\qquad z_{(p+1):(p+K-r)}= \frac{\left(\tilde{W}_{(r+1):K}^{\top} {W}_{(r+1):K} / n\right)\Tilde{F}^{1/2}_3(\hat{\gamma}) u}{\left\|\left(\tilde{W}_{(r+1):K}^{\top} {W}_{(r+1):K} / n\right)\Tilde{F}^{1/2}_3(\hat{\gamma}) u\right\|}.$$ 
Then by a direct calculation, 
\begin{eqnarray*}
        z^{\top} \left(T_n^{\top}\tilde{F}(\hat{\gamma})T_n\right)^{-1} z
        =\frac{1}{\|\left(\tilde{W}_{(r+1):K}^{\top} {W}_{(r+1):K} / n\right)\Tilde{F}^{1/2}_3(\hat{\gamma}) u\|^2},
\end{eqnarray*}
while $\sqrt{n}z^{\top}(T_n^{-1}\hat{\gamma}-\gamma^*)$ equals
\begin{eqnarray*}
&&\sqrt{n}z_{(p+1):(p+K-r)}^{\top}\left(\left(\tilde{W}_{(r+1):K}^{\top} {W}_{(r+1):K} / n\right)^{-1}\hat{\gamma}_{(p+1):(p+K-r)}-\gamma_{(p+1):(p+K-r)}^*\right)\\
    &=&\frac{\sqrt{n} u^{\top}\Tilde{F}^{-1/2}_3(\hat{\gamma})\left(\hat{\gamma}_{(p+1):(p+K-r)}-\left(\tilde{W}_{(r+1):K}^{\top} {W}_{(r+1):K} / n\right)\gamma_{(p+1):(p+K-r)}^*\right)}{\left\|\left(\tilde{W}_{(r+1):K}^{\top} {W}_{(r+1):K} / n\right)\Tilde{F}^{1/2}_3(\hat{\gamma}) u\right\|}.
\end{eqnarray*}
Therefore, \eqref{thm2result} reduces to 
\begin{equation*} \sqrt{n}u^\top\Tilde{F}^{-1/2}_3(\hat{\gamma})J(\hat{\gamma},\Tilde{W})\rightarrow\mathcal{N}(0,1),
\end{equation*}
where $$J(\hat{\gamma},\Tilde{W})=\hat{\gamma}-\left(\tilde{W}_{(r+1):K}^{\top} {W}_{(r+1):K} / n\right)\gamma_{(p+1):(p+K-r)}.$$
Since unit vector $u\in\mathbb{R}^{K-r}$ is arbitrary, by Cramer-Wold devices, 
\begin{equation*} \sqrt{n}\Tilde{F}^{-1/2}_3(\hat{\gamma})J(\hat{\gamma},\Tilde{W})\rightarrow\mathcal{N}(0,I_{K-r}).
\end{equation*}
By the continuous mapping theorem,
\begin{equation}
\label{thm2gamma}    nJ(\hat{\gamma},\Tilde{W})^{\top}\Tilde{F}^{-1}_3(\hat{\gamma})J(\hat{\gamma},\Tilde{W})\rightarrow \chi_{K-r}^2.
\end{equation}
By the Schur complement formula, 
 $\Tilde{F}^{-1}_3(\hat{\gamma})$ equals
\begin{equation*}   \frac{1}{n}\left(\Tilde{W}_{(r+1):K}^{\top}\kappa(\hat{\gamma})\Tilde{W}_{(r+1):K}\right)-\frac{1}{n}(\tilde{W}_{(r+1):K}^{\top}\kappa(\hat{\gamma})\tilde{Z}) \left(\tilde{Z}^{\top}\kappa(\hat{\gamma})\tilde{Z}\right)^{-1}(\tilde{Z}^{\top}\kappa(\hat{\gamma}) \tilde{W}_{(r+1):K}),
\end{equation*}
where $$\kappa(\hat{\gamma})=\operatorname{diag}\left((h^{\prime}(\tilde{g}^{\top}_i\hat{\gamma}))^2/v(\tilde{g}^{\top}_i\hat{\gamma})\right).$$
Therefore, \eqref{thm2gamma} can be further simplified to
\begin{equation*}
    n\left(\hat{\alpha}-n^{-1}\Tilde{W}_{(r+1):K}\Tilde{W}^{\top}_{(r+1):K} \alpha^*\right)^{\top} \tilde{O}\left(\hat{\alpha}-n^{-1}\Tilde{W}_{(r+1):K}\Tilde{W}^{\top}_{(r+1):K} \alpha^*\right) \ \to \  \chi_{K-r}^2,
\end{equation*}
where $\tilde{O}=n^{-1}(\kappa(\hat{\gamma})-\kappa(\hat{\gamma})\tilde{Z}(\tilde{Z}^{\top}\kappa(\hat{\gamma})\tilde{Z})^{-1}\tilde{Z}^{\top}\kappa(\hat{\gamma}))$.
The proof is complete.
\end{proof}

\section{The Proof of Corollary~\ref{cor:uniqueness}}
\label{app: uniqueness}
\begin{proof}
A solution to the estimating equation $\tilde{S}(\gamma)=0$ is a critical point of the likelihood function. It is unique if the likelihood function is concave or, equivalently, if $\partial\Tilde{S}(\gamma)/\partial \gamma$ is a negative-definite matrix for any $\gamma$. When the link function is natural, 
\begin{equation*}
    \frac{\partial \Tilde{S}(\gamma)}{\partial \gamma}= -\sum^n_{i=1} \left(\tilde{Z} \ \tilde{W}_{(r+1):K}\right)^{\top}\kappa(\gamma) \left(\tilde{Z} \ \tilde{W}_{(r+1):K}\right), \quad \text{where} \    \kappa({\gamma})=\operatorname{diag}\left(\frac{(h^{\prime}(\tilde{g}^{\top}_i{\gamma}))^2}{v(\tilde{g}^{\top}_i{\gamma})}\right)
\end{equation*}
We will show that each summand in the formula of $\partial\Tilde{S}(\gamma)/\partial \gamma$ is a positive definite matrix. First, regarding $\kappa(\gamma)$, by the definition of the smooth increasing function $h$ in \eqref{tmp}, each diagonal entry of $\kappa(\gamma)$ is positive. Therefore, it remains to show that  $\tilde{Z} \ \tilde{W}_{(r+1):K}$ is of full rank. 
Since $\operatorname{span}(\tilde{Z})=\operatorname{col}\left(X\right)$ and $\operatorname{span}(\tilde{W}_{(r+1):K})=\hat{\mathcal{N}}$, this is equivalent to $\operatorname{col}\left(X\right)\cap \hat{\mathcal{N}}=0$.
This identity holds if we prove that for any unit vector $u\in \operatorname{col}\left(X\right)$, the projection of $u$ onto $\hat{\mathcal{N}}$ has norm strictly less than one. To show that, we write 
$u=n^{-1/2}{Z}x$ for some
$x \in \mathbb{R}^{p}$ with $\|x\|=1$ and note that the projection onto $\hat{\mathcal{N}}$ is $\hat{\mathcal{P}}_N$. By Proposition~\ref{proposition 1}, the singular value decomposition in \eqref{defin:model}, and Assumption~\ref{cond:A4}, 
\begin{eqnarray*}
\left\|\hat{\mathcal{P}}_N u\right\|
&\leq& \left\|{\mathcal{P}}_N u\right\|+\left\|\left(\hat{\mathcal{P}}_N-{\mathcal{P}}_N\right)u\right\|\\
&\le &n^{-1/2}\left\|{\mathcal{P}}_N Z_{1:r}x_{1:r}\right\|+n^{-1/2}\left\|{\mathcal{P}}_NZ_{(r+1):K}x_{(r+1):K}\right\|+C_1
\tau_n\\
&= &n^{-1}\left\|W^\top_{(r+1):K} Z_{1:r}x_{1:r}\right\|+n^{-1}\left\|W^\top_{(r+1):K}Z_{(r+1):K}x_{(r+1):K}\right\|+C_1
\tau_n\\
&=& n^{-1}\left\|W^\top_{(r+1):K}Z_{(r+1):K}x_{(r+1):K}\right\|+C_1
\tau_n\\
&\le & n^{-1}\left\|W^\top_{(r+1):K}Z_{(r+1):K}\right\|\left\|x_{(r+1):K}\right\|+C_1
\tau_n\\
&\le& \sigma_{r+1}\|x\|+C_1\tau_n< 1,
\end{eqnarray*}
for sufficiently large $n$.
The proof is complete.
\end{proof}

\section{Extension to Laplacian Individual Effect}
\label{app: laplacian}
The appendix of \cite{le2022linear} demonstrates the extension of the subspace linear model incorporating the graph Laplacian. Similarly, our model can be extended in the same manner. 
We present the necessary assumptions and theoretical results for the reader's convenience.

We assume that the parameter vector $\alpha$ lies within the subspace spanned by the $K$ eigenvectors of $P=\mathbb{E} L=\mathbb{E} D-\mathbb{E} A$ associated with its smallest eigenvalues, 
while the estimation procedure is based on the perturbed version $\hat{P}=L=D-A$ of $P$, where $D$ is the diagonal matrix with node degrees $d_i$ on the diagonal. 

\begin{assumption}[Eigenvalue gap of the expected Laplacian]
\label{cond:A7}
Let $L=D-A$ be the Laplacian of a random network generated from the ``inhomogeneous Erd\"{o}s-R\'{e}nyi" and $P=\mathbb{E} L$. Denote by $\lambda_1 \leq \lambda_2 \leq \cdots \leq \lambda_n$ the eigenvalues of $P$. Assume that the $K$ smallest eigenvalues of $P$ are well separated from the remaining eigenvalues and their range is not too large:
\begin{equation*}
    \min _{i \leq K, i^{\prime}>K}\left|\lambda_i-\lambda_{i^{\prime}}\right| \geq \rho^{\prime} d, \quad \max _{i, i^{\prime} \leq K}\left|\lambda_i-\lambda_{i^{\prime}}\right| \leq d / \rho^{\prime},
\end{equation*}
where $\rho^{\prime}>0$ is a constant and $d=n \cdot \max _{i j} P_{i j}$.
\end{assumption}
Under Assumption~\ref{cond:A7}, the small projection perturbation assumption holds:
\begin{theorem}[Concentration of perturbed projection for the Laplacian]
Let $w_1, \ldots, w_n$ and $\lambda_1 \leq \lambda_2 \leq \cdots \leq \lambda_n$ be eigenvectors and corresponding eigenvalues of $\mathbb{E} L=\mathbb{E} D-\mathbb{E} A$ and similarly, let $\hat{w}_1, \ldots, \hat{w}_n$ and $\hat{\lambda}_1 \leq \hat{\lambda}_2 \leq \cdots \leq \hat{\lambda}_n$ be the eigenvectors and eigenvalues of $L=D-A$. Denote $W=\left(w_1, \ldots, w_K\right)$ and $\hat{W}=\left(\hat{w}_1, \ldots, \hat{w}_K\right)$. Assume that Assumption \ref{cond:A7} holds and $d \geq C \log n$ for a sufficiently large constant $C$. Then for any fixed unit vector $v$, with high probability we have
$$
\left\|\left(\hat{W} \hat{W}^{\top}-W W^{\top}\right) v\right\| \leq \frac{C\left[K\left(1+n\|W\|_{\infty}^2\right)\right]^{1 / 2} \log n}{d}.
$$
\end{theorem}

\newpage

\section{Additional simulation results for model misspecification}\label{sec:mis}
We present additional simulations under the same design as Section~\ref{sec:Simulation Studies}, now examining the impact of misspecifying either \(K\) (network subspace dimension) or \(r\) (intersection dimension). Results are provided in Tables~\ref{tableLogitkSBM}--\ref{tablePoirDiag}. Unless noted, medians are reported across the same two-level Monte Carlo scheme used previously.

\begin{table}[H]
\centering
        \caption{Median MSE ($\times 10^{2}$), coverage probability, and MSPE ($\times 10^{2}$) for subspace logistic regression under SBM with random network perturbations and misspecified $K$.}
{\begin{tabular}{ccccccccccc}
\hline
\hline
\multirow{2}{*}{ n } & \multirow{2}{*}{ avg.\ degree } & \multicolumn{3}{c}{ $K=2$}& \multicolumn{3}{c}{$K=3$ (True Model)}& \multicolumn{3}{c}{$K=4$}  \\
 & &MSE  & Coverage& MSPE & MSE   & Coverage & MSPE & MSE  & Coverage & MSPE  \\
\hline 
\multirow{3}{*}{ 500 } & $2 \log n$ & 1.14 & 94.4\% & 1.71 & 1.16 & 94.6\% & 1.11 & 1.18 & 94.4\% & 1.12 \\
& $\sqrt{n}$ & 1.18 & 94.6\% & 1.45 & 1.15 & 94.8\% & 0.64 & 1.17 & 94.6\% & 0.69 \\ 
& $n^{2 / 3}$ & 1.08 & 95.1\%& 1.35 & 1.13 & 95.0\% & 0.31 & 1.14 & 95.1\% & 0.35
\\
\hline 
\multirow{3}{*}{ 1000} & $2 \log n$ & 0.51 & 95.3\% & 1.69 & 0.56 & 94.7\% & 0.96 & 0.57 & 94.8\% & 0.98\\
& $\sqrt{n}$  & 0.51 & 95.2\% & 1.46 & 0.57 & 94.9\% & 0.43 & 0.57 & 95.0\% & 0.45 \\
& $n^{2 / 3}$ & 0.55 & 94.3\% & 1.39 & 0.58 & 95.0\% & 0.18 & 0.57 & 95.0\% & 0.20 \\
\hline 
\multirow{3}{*}{ 2000 } & $2 \log n$ & 0.29  & 94.3\% & 1.66 & 0.35 & 93.1\% & 0.90 & 0.31 & 94.3\% & 0.90 \\
& $\sqrt{n}$  & 0.36 & 91.3\% & 1.29 & 0.31 & 94.7\% & 0.30 & 0.31 & 94.9\% & 0.31 \\
& $n^{2 / 3}$ & 0.44 & 87.5\% & 1.31 & 0.30 & 95.1\% & 0.10 & 0.31 & 95.0\% & 0.11 \\
\hline 
\multirow{3}{*}{ 4000 }& $2 \log n$ & 0.14 & 93.9\% & 1.56 & 0.16 & 92.7\%& 0.75 & 0.17 & 92.6\% & 0.75 \\ 
& $\sqrt{n}$ & 0.14 & 93.3\% & 1.41 & 0.14 & 94.9\% & 0.19 & 0.14 & 94.7\% & 0.20\\
& $n^{2 / 3}$  & 0.17 & 91.6\% & 1.30 & 0.14 & 95.0\% & 0.06 & 0.14 & 94.9\% & 0.06 \\ 
\hline
\hline
\end{tabular}}
    \label{tableLogitkSBM}
\end{table}

\begin{table}[H]
\centering
        \caption{Median MSE ($\times 10^{2}$), coverage probability, and MSPE ($\times 10^{2}$) for subspace logistic regression under SBM with random network perturbations and misspecified $r$.}
{\begin{tabular}{ccccccccccc}
\hline
\hline
\multirow{2}{*}{ n } & \multirow{2}{*}{ avg.\ degree } & \multicolumn{3}{c}{ $r=0$}& \multicolumn{3}{c}{$r=1$ (True Model)}& \multicolumn{3}{c}{$r=2$}  \\
 & &MSE  & Coverage& MSPE & MSE   & Coverage & MSPE & MSE  & Coverage & MSPE  \\
\hline 
\multirow{3}{*}{ 500 } & $2 \log n$ & 1.16 & 94.6\% & 1.13 & 1.16 & 94.6\% & 1.11 & -- & -- & 1.62 \\
& $\sqrt{n}$ & 1.18 & 94.8\% & 0.68 & 1.15& 94.8\% & 0.64 & -- & -- & 1.32 \\ 
& $n^{2 / 3}$ & 1.15 & 95.0\%& 0.35 & 1.13 & 95.0\% & 0.31 & -- & -- & 1.20
\\
\hline 
\multirow{3}{*}{ 1000} & $2 \log n$ & 0.57 & 94.7\% & 0.98 & 0.56 & 94.7\% & 0.96 & -- & -- & 1.57\\
& $\sqrt{n}$  & 0.57 & 95.0\% & 0.45 & 0.57 & 94.9\% & 0.43 & -- & -- & 1.24 \\
& $n^{2 / 3}$ & 0.57 & 95.1\% & 0.20 & 0.59 & 95.0\% &0.18 & -- & -- & 1.12 \\
\hline 
\multirow{3}{*}{ 2000 } & $2 \log n$ & 0.29  & 94.3\% & 0.85 & 0.35 & 93.1\% & 0.90 & -- & -- & 1.57\\
& $\sqrt{n}$  & 0.31 & 94.8\% & 0.29 & 0.31 & 94.7\% & 0.30 & -- & -- & 1.17 \\
& $n^{2 / 3}$ & 0.31 & 95.1\% & 0.11 & 0.30 & 95.1\% & 0.10 & -- & -- & 1.05 \\
\hline 
\multirow{3}{*}{ 4000 }& $2 \log n$ & 0.17 & 92.5\% & 0.75 & 0.16 & 92.7\%& 0.75 & -- & -- & 1.39 \\ 
& $\sqrt{n}$ & 0.14 & 94.9\% & 0.19 & 0.14 & 94.9\% & 0.19 & -- & -- & 1.09\\
& $n^{2 / 3}$  & 0.14 & 95.0\% & 0.06 & 0.14 & 95.0\% & 0.06 & -- & -- & 1.00 \\ 
\hline
\hline
\end{tabular}}
    \label{tableLogitrSBM}
\end{table}

    \begin{table}[H]
    \centering
        \caption{Median MSE ($\times 10^{2}$), coverage probability, and MSPE for subspace Poisson regression under SBM with random network perturbations and misspecified $K$.}
{\begin{tabular}{ccccccccccc}
\hline
\hline
\multirow{2}{*}{ n } & \multirow{2}{*}{ avg.\ degree } & \multicolumn{3}{c}{ $K=2$}& \multicolumn{3}{c}{$K=3$ (True Model)}& \multicolumn{3}{c}{$K=4$}  \\
 & &MSE  & Coverage& MSPE & MSE   & Coverage & MSPE & MSE  & Coverage & MSPE  \\
\hline 
\multirow{3}{*}{ 500 } & $2 \log n$ & 0.23 & 72.8\% & 2.85 & 0.35 & 75.8\% & 2.30 & 0.21 & 76.3\% & 2.25 \\
& $\sqrt{n}$ & 0.12 & 89.9\% & 2.09 & 0.16 & 86.2\% & 1.41 & 0.15 & 86.2\% & 1.40 \\ 
& $n^{2 / 3}$ & 0.10 & 93.4\%& 1.19 & 0.10 & 93.5\% & 0.53 & 0.10 & 93.5\% & 1.30
\\
\hline 
\multirow{3}{*}{ 1000} & $2 \log n$ & 0.06 & 92.4\% & 1.90 & 0.08 & 87.2\% & 1.66 & 0.07 & 90.2\% & 1.60\\
& $\sqrt{n}$  & 0.09 & 85.0\% & 1.15 & 0.06 & 92.9\% & 0.82 & 0.06 & 92.0\% & 0.83 \\
& $n^{2 / 3}$ & 0.18 & 62.2\% & 0.71 & 0.06 & 94.2\% &0.27 & 0.06 & 94.3\% & 0.27 \\
\hline 
\multirow{3}{*}{ 2000 } & $2 \log n$ & 0.34  & 6.2\% & 1.22 & 0.13 & 43.4\% & 1.01 & 0.13 & 43.7\% & 1.00 \\
& $\sqrt{n}$  & 0.23 & 20.7\% & 0.71 & 0.04 & 86.3\% & 0.43 & 0.04 & 86.3\% & 0.43 \\
& $n^{2 / 3}$ & 0.21 & 26.4\% & 0.51 & 0.03 & 94.2\% & 0.13 & 0.03 & 94.3\% & 0.13 \\
\hline 
\multirow{3}{*}{ 4000 }& $2 \log n$ & 0.02 & 84.7\% & 1.47 & 0.04 & 71.4\%& 1.67 & 0.03 & 73.6\% & 1.66 \\ 
& $\sqrt{n}$ & 0.11 & 15.2\% & 0.43 & 0.01 & 92.8\% & 0.51 & 0.02 & 92.7\% & 0.51\\
& $n^{2 / 3}$  & 0.17 & 3.9\% & 0.64 & 0.01 & 94.8\% & 0.13 & 0.01 & 94.4\% & 0.13 \\ 
\hline
\hline
\end{tabular}}
    \label{tablePoikSBM}
\end{table}

\begin{table}[H]
\centering
        \caption{Median MSE ($\times 10^{2}$), coverage probability, and MSPE for subspace Poisson regression under SBM with random network perturbations and misspecified $r$.}
{\begin{tabular}{ccccccccccc}
\hline
\hline
\multirow{2}{*}{ n } & \multirow{2}{*}{ avg.\ degree } & \multicolumn{3}{c}{ $r=0$}& \multicolumn{3}{c}{$r=1$ (True Model)}& \multicolumn{3}{c}{$r=2$}  \\
 & &MSE  & Coverage& MSPE & MSE   & Coverage & MSPE & MSE  & Coverage & MSPE  \\
\hline 
\multirow{3}{*}{ 500 } & $2 \log n$ & 0.20 & 78.7\% & 2.13 & 0.35 & 75.8\% & 2.30 & -- & -- & 3.34 \\
& $\sqrt{n}$ & 0.14 & 88.1\% & 1.24 & 0.16 & 86.2\% & 1.41 & -- & -- & 2.67 \\ 
& $n^{2 / 3}$ & 0.10 & 94.2\%& 0.46 & 0.10 & 93.5\% & 0.53 & -- & -- & 2.31
\\
\hline 
\multirow{3}{*}{ 1000} & $2 \log n$ & 0.07 & 90.8\% & 1.46 & 0.08 & 87.2\% & 1.66 & -- & -- & 2.67 \\
& $\sqrt{n}$  & 0.06 & 93.1\% & 0.72 & 0.06 & 92.9\% & 0.82 & -- & -- & 2.34 \\
& $n^{2 / 3}$ & 0.06 & 94.3\% & 0.24 & 0.06 & 94.2\% &0.27 & -- & -- & 2.24 \\
\hline 
\multirow{3}{*}{ 2000 } & $2 \log n$ & 0.12  & 49.4\% & 0.93 & 0.13 & 43.4\% & 1.01 & -- & -- & 1.92 \\
& $\sqrt{n}$  & 0.04 & 88.4\% & 0.38 & 0.04 & 86.3\% & 0.43 & -- & -- & 1.89 \\
& $n^{2 / 3}$ & 0.03 & 94.6\% & 0.12 & 0.03 & 94.2\% & 0.13 & -- & -- & 1.88 \\
\hline 
\multirow{3}{*}{ 4000 }& $2 \log n$ & 0.04 & 70.9\% & 1.47 & 0.04 & 71.4\%& 1.67 & -- & -- & 2.84 \\ 
& $\sqrt{n}$ & 0.01 & 93.0\% & 0.43 & 0.01 & 92.8\% & 0.51 & -- & -- & 2.53\\
& $n^{2 / 3}$  & 0.01 & 94.5\% & 0.11 & 0.01 & 94.8\% & 0.13 & -- & -- & 2.40\\ 
\hline
\hline
\end{tabular}}
    \label{tablePoirSBM}
\end{table}

\begin{table}[H]
\caption{Median MSE ($\times 10^{2}$), coverage probability, and MSPE ($\times 10^{2}$) for subspace logistic regression under DCBM with random network perturbations and misspecified $K$.}
\centering
\begin{tabular}{ccccccccccc}
\hline
\hline
\multirow{2}{*}{ n } & \multirow{2}{*}{ avg.\ degree } & \multicolumn{3}{c}{ $K=2$} & \multicolumn{3}{c}{$K=3$ (True Model)} & \multicolumn{3}{c}{$K=4$}  \\
& & MSE & Coverage & MSPE & MSE & Coverage & MSPE & MSE & Coverage & MSPE \\
\hline 
\multirow{3}{*}{ 500 } & $2 \log n$ & 1.08 & 95.4\% & 1.53 & 1.18 & 95.4\% & 1.05 & 1.22 & 95.1\% & 1.07 \\
& $\sqrt{n}$ & 1.10 & 95.0\% & 1.40 & 1.19 & 94.8\% & 0.60 & 1.23 & 95.0\% & 0.64 \\ 
& $n^{2 / 3}$ & 1.13 & 94.4\% & 1.45 & 1.22 & 95.3\% & 0.28 & 1.25 & 95.0\% & 0.32 \\
\hline 
\multirow{3}{*}{ 1000} & $2 \log n$ & 0.53 & 95.2\% & 1.63 & 0.56 & 95.1\% & 0.91 & 0.59 & 94.9\% & 0.90 \\
& $\sqrt{n}$  & 0.58 & 94.1\% & 1.41 & 0.57 & 95.0\% & 0.40 & 0.59 & 95.0\% & 0.42 \\
& $n^{2 / 3}$ & 0.62 & 93.0\% & 1.31 & 0.57 & 95.1\% & 0.16 & 0.59 & 95.0\% & 0.18 \\
\hline 
\multirow{3}{*}{ 2000 } & $2 \log n$ & 0.28 & 94.7\% & 1.59 & 0.29 & 95.1\% & 0.82 & 0.29 & 95.2\% & 0.83 \\
& $\sqrt{n}$  & 0.38 & 89.9\% & 1.26 & 0.29 & 95.0\% & 0.27 & 0.30 & 95.0\% & 0.28 \\
& $n^{2 / 3}$ & 0.45 & 86.4\% & 1.15 & 0.28 & 95.1\% & 0.09 & 0.30 & 95.1\% & 0.10 \\
\hline 
\multirow{3}{*}{ 4000 } & $2 \log n$ & 0.19 & 94.3\% & 1.62 & 0.15 & 94.2\% & 0.70 & 0.15 & 94.1\% & 0.71 \\ 
& $\sqrt{n}$ & 0.20 & 90.5\% & 1.28 & 0.14 & 95.0\% & 0.18 & 0.15 & 94.8\% & 0.19 \\
& $n^{2 / 3}$  & 0.24 & 84.9\% & 1.26 & 0.14 & 95.1\% & 0.05 & 0.14 & 94.9\% & 0.06 \\ 
\hline
\hline
\end{tabular}
\label{tableLogitkDCBM}
\end{table}

\begin{table}[H]
\caption{Median MSE ($\times 10^{2}$), coverage probability, and MSPE ($\times 10^{2}$) for subspace logistic regression under DCBM with random network perturbations and misspecified $r$.}
\centering
\begin{tabular}{ccccccccccc}
\hline
\hline
\multirow{2}{*}{ n } & \multirow{2}{*}{ avg.\ degree } & \multicolumn{3}{c}{ $r=0$} & \multicolumn{3}{c}{$r=1$ (True Model)} & \multicolumn{3}{c}{$r=2$}  \\
& & MSE & Coverage & MSPE & MSE & Coverage & MSPE & MSE & Coverage & MSPE \\
\hline 
\multirow{3}{*}{ 500 } & $2 \log n$ & 1.15 & 95.3\% & 1.07 & 1.18 & 95.4\% & 1.05 & -- & -- & 1.70 \\
& $\sqrt{n}$ & 1.23 & 95.2\% & 0.64 & 1.19 & 94.8\% & 0.60 & -- & -- & 1.39 \\ 
& $n^{2 / 3}$ & 1.24 & 95.0\% & 0.32 & 1.22 & 95.3\% & 0.28 & -- & -- & 1.20 \\
\hline 
\multirow{3}{*}{ 1000} & $2 \log n$ & 0.57 & 95.0\% & 0.93 & 0.56 & 95.1\% & 0.91 & -- & -- & 1.68\\
& $\sqrt{n}$  & 0.59 & 95.0\% & 0.42 & 0.57 & 95.0\% & 0.40 & -- & -- & 1.31 \\
& $n^{2/3}$ & 0.60 & 94.9\% & 0.19 & 0.57 & 95.1\% & 0.16 & -- & -- & 1.16 \\
\hline 
\multirow{3}{*}{ 2000 } & $2 \log n$ & 0.28 & 95.0\% & 0.83 & 0.29 & 95.1\% & 0.82 & -- & -- & 1.64\\
& $\sqrt{n}$  & 0.30 & 95.0\% & 0.28 & 0.29 & 95.0\% & 0.27 & -- & -- & 1.28 \\
& $n^{2/3}$ & 0.30 & 95.2\% & 0.10 & 0.28 & 95.1\% & 0.09 & -- & -- & 1.15 \\
\hline 
\multirow{3}{*}{ 4000 }& $2 \log n$ & 0.15 & 94.0\% & 0.71 & 0.15 & 94.2\% & 0.70 & -- & -- & 1.40 \\ 
& $\sqrt{n}$ & 0.15 & 95.2\% & 0.19 & 0.14 & 95.0\% & 0.18 & -- & -- & 1.11\\
& $n^{2/3}$  & 0.15 & 95.0\% & 0.06 & 0.14 & 95.1\% & 0.05 & -- & -- & 1.02 \\ 
\hline
\hline
\end{tabular}
\label{tableLogitrDCBM}
\end{table}

\begin{table}[H]
\caption{Median MSE ($\times 10^{2}$), coverage probability, and MSPE for subspace Poisson regression under DCBM with random network perturbations and misspecified $K$.}
\centering
\begin{tabular}{ccccccccccc}
\hline
\hline
\multirow{2}{*}{ n } & \multirow{2}{*}{ avg.\ degree } & \multicolumn{3}{c}{ $K=2$} & \multicolumn{3}{c}{$K=3$ (True Model)} & \multicolumn{3}{c}{$K=4$}  \\
& & MSE & Coverage & MSPE & MSE & Coverage & MSPE & MSE & Coverage & MSPE \\
\hline 
\multirow{3}{*}{ 500 } & $2 \log n$ & 0.51 & 39.9\% & 1.29 & 0.21 & 75.2\% & 1.09 & 0.22 & 76.3\% & 1.07 \\
& $\sqrt{n}$ & 0.44 & 49.4\% & 1.12 & 0.11 & 90.4\% & 0.72 & 0.11 & 89.8\% & 0.72 \\ 
& $n^{2 / 3}$ & 0.47 & 46.9\% & 1.06 & 0.09 & 93.8\% & 0.32 & 0.09 & 93.8\% & 0.33 \\
\hline 
\multirow{3}{*}{ 1000} & $2 \log n$ & 0.63 & 2.55\% & 1.31 & 0.27 & 29.9\% & 0.97 & 0.23 & 39.1\% & 0.96 \\
& $\sqrt{n}$ & 0.53 & 6.7\% & 1.21 & 0.07 & 82.4\% & 0.57 & 0.08 & 81.1\% & 0.57 \\
& $n^{2 / 3}$ & 0.51 & 8.4\% & 1.15 & 0.04 & 93.5\% & 0.21 & 0.04 & 93.8\% & 0.21 \\
\hline 
\multirow{3}{*}{ 2000} & $2 \log n$ & 0.47 & 0.10\% & 1.51 & 0.14 & 29.1\% & 1.03 & 0.10 & 27.8\% & 1.03 \\
& $\sqrt{n}$ & 0.39 & 1.20\% & 1.41 & 0.03 & 86.4\% & 0.49 & 0.04 & 82.8\% & 0.49 \\
& $n^{2 / 3}$ & 0.33 & 1.95\% & 1.39 & 0.02 & 94.0\% & 0.16 & 0.02 & 93.9\% & 0.16 \\
\hline 
\multirow{3}{*}{ 4000} & $2 \log n$ & 0.19 & 16.3\% & 1.46 & 0.04 & 89.8\% & 1.06 & 0.02 & 91.8\% & 1.01 \\ 
& $\sqrt{n}$ & 0.20 & 31.0\% & 1.09 & 0.01 & 94.3\% & 0.35 & 0.01 & 93.9\% & 0.35 \\
& $n^{2 / 3}$ & 0.24 & 0.30\% & 0.93 & 0.01 & 94.6\% & 0.09 & 0.01 & 94.6\% & 0.09 \\ 
\hline
\hline
\end{tabular}
\label{tablePoikDCBM}
\end{table}

\begin{table}[H]
\caption{Median MSE ($\times 10^{2}$), coverage probability, and MSPE for subspace Poisson regression under DCBM with random network perturbations and misspecified $r$.}
\centering
\begin{tabular}{ccccccccccc}
\hline
\hline
\multirow{2}{*}{ n } & \multirow{2}{*}{ avg.\ degree } & \multicolumn{3}{c}{ $r=0$} & \multicolumn{3}{c}{$r=1$ (True Model)} & \multicolumn{3}{c}{$r=2$}  \\
& & MSE & Coverage & MSPE & MSE & Coverage & MSPE & MSE & Coverage & MSPE \\
\hline 
\multirow{3}{*}{ 500 } & $2 \log n$ & 0.20 & 76.9\% & 1.03 & 0.21 & 75.2\% & 1.09 & -- & -- & 1.56 \\
& $\sqrt{n}$ & 0.10 & 91.4\% & 0.67 & 0.11 & 90.4\% & 0.72 & -- & -- & 1.54 \\ 
& $n^{2/3}$ & 0.09 & 93.7\% & 0.32 & 0.09 & 93.8\% & 0.32 & -- & -- & 1.51 \\
\hline 
\multirow{3}{*}{ 1000} & $2 \log n$ & 0.20 & 42.7\% & 0.92 & 0.27 & 29.9\% & 0.97 & -- & -- & 1.36 \\
& $\sqrt{n}$ & 0.08 & 84.2\% & 0.53 & 0.07 & 82.4\% & 0.57 & -- & -- & 1.31 \\
& $n^{2/3}$ & 0.04 & 93.9\% & 0.21 & 0.04 & 93.5\% & 0.21 & -- & -- & 1.28 \\
\hline 
\multirow{3}{*}{ 2000} & $2 \log n$ & 0.13 & 30.0\% & 0.97 & 0.14 & 29.1\% & 1.03 & -- & -- & 1.39 \\
& $\sqrt{n}$ & 0.05 & 83.4\% & 0.45 & 0.03 & 86.4\% & 0.49 & -- & -- & 1.26 \\
& $n^{2/3}$ & 0.02 & 94.3\% & 0.15 & 0.02 & 94.0\% & 0.16 & -- & -- & 1.21 \\
\hline 
\multirow{3}{*}{ 4000} & $2 \log n$ & 0.01 & 92.4\% & 0.97 & 0.04 & 89.8\% & 1.06 & -- & -- & 2.10 \\ 
& $\sqrt{n}$ & 0.01 & 94.0\% & 0.32 & 0.01 & 94.3\% & 0.35 & -- & -- & 1.98\\
& $n^{2/3}$ & 0.01 & 94.8\% & 0.09 & 0.01 & 94.6\% & 0.09 & -- & -- & 1.94 \\ 
\hline
\hline
\end{tabular}
\label{tablePoirDCBM}
\end{table}

\begin{table}[H]
\caption{Median MSE ($\times 10^{2}$), coverage probability, and MSPE ($\times 10^{2}$) for subspace logistic regression under Diagonal Graphon with random network perturbations and misspecified $K$.}
\centering
\begin{tabular}{ccccccccccc}
\hline
\hline
\multirow{2}{*}{ n } & \multirow{2}{*}{ avg.\ degree } & \multicolumn{3}{c}{ $K=2$} & \multicolumn{3}{c}{$K=3$ (True Model)} & \multicolumn{3}{c}{$K=4$}  \\
& & MSE & Coverage & MSPE & MSE & Coverage & MSPE & MSE & Coverage & MSPE \\
\hline 
\multirow{3}{*}{ 500 } & $2 \log n$ & 1.23 & 93.6\% & 1.12 & 1.22 & 92.8\% & 0.38 & -- & -- & 0.38 \\
& $\sqrt{n}$ & 1.19 & 94.0\% & 1.07 & 1.19 & 93.8\% & 0.26 & -- & -- & 0.27 \\ 
& $n^{2 / 3}$ & 1.16 & 94.4\% & 1.04 & 1.13 & 94.4\% & 0.18 & -- & -- & 0.18 \\
\hline 
\multirow{3}{*}{ 1000} & $2 \log n$ & 0.63 & 91.2\% & 1.36 & 0.67 & 93.3\% & 0.32 & -- & -- & 0.33 \\
& $\sqrt{n}$  & 0.69 & 91.5\% & 1.32 & 0.63 & 94.1\% & 0.17 & -- & -- & 0.18 \\
& $n^{2 / 3}$ & 0.69 & 91.0\% & 1.35 & 0.60 & 94.9\% & 0.10 & -- & -- & 0.13 \\
\hline 
\multirow{3}{*}{ 2000 } & $2 \log n$ & 0.49 & 82.6\% & 1.28 & 0.32 & 92.6\% & 0.26 & -- & -- & 0.26 \\
& $\sqrt{n}$  & 0.50 & 81.9\% & 1.27 & 0.29 & 94.1\% & 0.20 & -- & -- & 0.11 \\
& $n^{2 / 3}$ & 0.52 & 81.4\% & 1.24 & 0.28 & 94.8\% & 0.05 & -- & -- & 0.05 \\
\hline 
\multirow{3}{*}{ 4000 } & $2 \log n$ & 0.38 & 67.9\% & 1.01 & 0.15 & 94.0\% & 0.17 & -- & -- & 0.19 \\ 
& $\sqrt{n}$ & 0.53 & 54.8\% & 0.99 & 0.14 & 94.3\% & 0.06 & -- & -- & 0.06 \\
& $n^{2 / 3}$  & 0.50 & 53.5\% & 0.99 & 0.14 & 94.8\% & 0.03 & -- & -- & 0.03 \\ 
\hline
\hline
\end{tabular}
\label{tableLogitkDiag}
\end{table}

\begin{table}[H]
\caption{Median MSE ($\times 10^{2}$), coverage probability, and MSPE ($\times 10^{2}$) for subspace logistic regression under Diagonal Graphon with random network perturbations and misspecified $r$.}
\centering
\begin{tabular}{ccccccccccc}
\hline
\hline
\multirow{2}{*}{ n } & \multirow{2}{*}{ avg.\ degree } & \multicolumn{3}{c}{ $r=0$} & \multicolumn{3}{c}{$r=1$ (True Model)} & \multicolumn{3}{c}{$r=2$}  \\
& & MSE & Coverage & MSPE & MSE & Coverage & MSPE & MSE & Coverage & MSPE \\
\hline 
\multirow{3}{*}{ 500 } & $2 \log n$ & 1.30 & 92.7\% & 0.41 & 1.22 & 92.8\% & 0.38 & -- & -- & 1.57 \\
& $\sqrt{n}$ & 1.27 & 93.2\% & 0.31 & 1.19 & 93.8\% & 0.26 & -- & -- & 1.51 \\ 
& $n^{2 / 3}$ & 1.12 & 94.6\% & 0.22 & 1.13 & 94.4\% & 0.18 & -- & -- & 1.38 \\
\hline 
\multirow{3}{*}{ 1000} & $2 \log n$ & 0.65 & 93.6\% & 0.33 & 0.67 & 93.3\% & 0.32 & -- & -- & 1.31 \\
& $\sqrt{n}$  & 0.62 & 94.4\% & 0.19 & 0.63 & 94.1\% & 0.17 & -- & -- & 1.15 \\
& $n^{2 / 3}$ & 0.61 & 94.7\% & 0.13 & 0.60 & 94.9\% & 0.10 & -- & -- & 1.30 \\
\hline 
\multirow{3}{*}{ 2000 } & $2 \log n$ & 0.30 & 93.1\% & 0.26 & 0.32 & 92.6\% & 0.26 & -- & -- & 1.66 \\
& $\sqrt{n}$  & 0.29 & 94.0\% & 0.11 & 0.29 & 94.1\% & 0.20 & -- & -- & 1.42 \\
& $n^{2 / 3}$ & 0.28 & 94.8\% & 0.06 & 0.28 & 94.8\% & 0.05 & -- & -- & 1.31 \\
\hline 
\multirow{3}{*}{ 4000 }& $2 \log n$ & 0.15 & 93.7\% & 0.17 & 0.15 & 94.0\% & 0.17 & -- & -- & 1.10 \\ 
& $\sqrt{n}$ & 0.14 & 94.6\% & 0.06 & 0.14 & 94.3\% & 0.06 & -- & -- & 1.07 \\
& $n^{2 / 3}$  & 0.13 & 95.0\% & 0.03 & 0.14 & 94.8\% & 0.03 & -- & -- & 1.02 \\ 
\hline
\hline
\end{tabular}
\label{tableLogitrDiag}
\end{table}

\begin{table}[H]
\caption{Median MSE ($\times 10^{2}$), coverage probability, and MSPE for subspace Poisson regression under Diagonal Graphon with random network perturbations and misspecified $K$.}
\centering
\begin{tabular}{ccccccccccc}
\hline
\hline
\multirow{2}{*}{ n } & \multirow{2}{*}{ avg.\ degree } & \multicolumn{3}{c}{ $K=2$} & \multicolumn{3}{c}{$K=3$ (True Model)} & \multicolumn{3}{c}{$K=4$}  \\
& & MSE & Coverage & MSPE & MSE & Coverage & MSPE & MSE & Coverage & MSPE \\
\hline 
\multirow{3}{*}{ 500 } & $2 \log n$ & 1.96 & 0.2\% & 0.80 & 0.53 & 72.4\% & 0.44 & -- & -- & 0.44 \\
& $\sqrt{n}$ & 1.97 & 0.0\% & 0.78 & 0.38 & 84.9\% & 0.25 & -- & -- & 0.25 \\ 
& $n^{2 / 3}$ & 2.12 & 0.0\% & 0.76 & 0.23 & 93.4\% & 0.07 & -- & -- & 0.07 \\
\hline 
\multirow{3}{*}{ 1000} & $2 \log n$ & 0.34 & 29.1\% & 1.03 & 0.27 & 57.7\% & 0.22 & -- & -- & 0.22 \\
& $\sqrt{n}$ & 0.25 & 44.6\% & 0.93 & 0.11 & 88.8\% & 0.09 & -- & -- & 0.09 \\
& $n^{2 / 3}$ & 0.22 & 50.2\% & 0.91 & 0.08 & 93.8\% & 0.03 & -- & -- & 0.03 \\
\hline 
\multirow{3}{*}{ 2000} & $2 \log n$ & 0.96 & 0.0\% & 0.88 & 0.09 & 77.0\% & 0.24 & -- & -- & 0.24 \\
& $\sqrt{n}$ & 0.89 & 0.0\% & 0.83 & 0.05 & 91.7\% & 0.08 & -- & -- & 0.08 \\
& $n^{2 / 3}$ & 0.86 & 0.0\% & 0.81 & 0.04 & 93.7\% & 0.02 & -- & -- & 0.02 \\
\hline 
\multirow{3}{*}{ 4000} & $2 \log n$ & 6.74 & 0.0\% & 1.11 & 0.61 & 0\% & 1.10 & -- & -- & 1.10 \\ 
& $\sqrt{n}$ & 5.17 & 0.0\% & 0.54 & 0.06 & 58.3\% & 0.29 & -- & -- & 0.29 \\
& $n^{2 / 3}$ & 5.04 & 0.0\% & 0.39 & 0.02 & 93.0\% & 0.06 & -- & -- & 0.06 \\ 
\hline
\hline
\end{tabular}
\label{tablePoikDiag}
\end{table}

\begin{table}[H]
\caption{Median MSE ($\times 10^{2}$), coverage probability, and MSPE for subspace Poisson regression under Diagonal Graphon with random network perturbations and misspecified $r$.}
\centering
\begin{tabular}{ccccccccccc}
\hline
\hline
\multirow{2}{*}{ n } & \multirow{2}{*}{ avg.\ degree } & \multicolumn{3}{c}{ $r=0$} & \multicolumn{3}{c}{$r=1$ (True Model)} & \multicolumn{3}{c}{$r=2$}  \\
& & MSE & Coverage & MSPE & MSE & Coverage & MSPE & MSE & Coverage & MSPE \\
\hline 
\multirow{3}{*}{ 500 } & $2 \log n$ & 0.41 & 69.4\% & 0.23 & 0.53 & 72.4\% & 0.44 & -- & -- & 0.90 \\
& $\sqrt{n}$ & 0.25 & 84.9\% & 0.13 & 0.38 & 84.9\% & 0.25 & -- & -- & 0.73 \\ 
& $n^{2/3}$ & 0.16 & 93.4\% & 0.05 & 0.23 & 93.4\% & 0.07 & -- & -- & 0.64 \\
\hline 
\multirow{3}{*}{ 1000} & $2 \log n$ & 0.14 & 80.2\% & 0.19 & 0.27 & 57.7\% & 0.22 & -- & -- & 0.60 \\
& $\sqrt{n}$ & 0.09 & 88.8\% & 0.08 & 0.11 & 88.8\% & 0.09 & -- & -- & 0.39 \\
& $n^{2/3}$ & 0.06 & 93.8\% & 0.03 & 0.08 & 93.8\% & 0.03 & -- & -- & 0.33 \\
\hline 
\multirow{3}{*}{ 2000} & $2 \log n$ & 0.07 & 86.3\% & 0.13 & 0.09 & 77.0\% & 0.24 & -- & -- & 0.65 \\
& $\sqrt{n}$ & 0.05 & 91.0\% & 0.04 & 0.05 & 91.7\% & 0.08 & -- & -- & 0.51 \\
& $n^{2/3}$ & 0.04 & 94.1\% & 0.02 & 0.04 & 93.7\% & 0.02 & -- & -- & 0.45 \\
\hline 
\multirow{3}{*}{ 4000} & $2 \log n$ & 0.47 & 0.0\% & 0.82 & 0.61 & 0\% & 1.10 & -- & -- & 2.05 \\ 
& $\sqrt{n}$ & 0.36 & 76.0\% & 0.20 & 0.06 & 58.3\% & 0.29 & -- & -- & 1.43 \\
& $n^{2/3}$ & 0.02 & 93.3\% & 0.04 & 0.02 & 93.0\% & 0.06 & -- & -- & 1.14 \\ 
\hline
\hline
\end{tabular}
\label{tablePoirDiag}
\end{table}

Next, we introduce a new simulation setup where $X_1, X_2$ are generated from a uniform distribution $U\left(-2,2\right)$, while other parameters are updated to $\beta^*=(0.5,0.5)^{\top}$ and $\gamma_{3:5}^*=(0.5,0.5,0.5)^{\top}$.
Since inference can be conducted for both $\beta_1$ and $\beta_2$, we report the median of the same statistics averaged over both parameters under the stochastic block model (SBM) in Tables \ref{tableLogitkbeta_2} to \ref{tablePoirbeta_2}.

\begin{table}[H]
    \caption{Median MSE ($\times 10^{2}$), coverage probability, and MSPE ($\times 10^{2}$) for subspace logistic regression under SBM with random network perturbations and misspecified $K$.}
    \centering
    \begin{tabular}{cccccccccccc}
    \hline
    \hline
    \multirow{2}{*}{ n } & \multirow{2}{*}{ avg.\ degree } & \multicolumn{3}{c}{ $K=2$ } & \multicolumn{3}{c}{$K=3$ (True Model)} & \multicolumn{3}{c}{$K=4$}  \\
    & & MSE  & Coverage & MSPE & MSE  & Coverage & MSPE & MSE  & Coverage & MSPE  \\
    \hline
    \multirow{3}{*}{ 500 } 
    & $2 \log n$ & 0.79 & 94.4\% & 1.70 & 0.78 & 94.8\% & 1.20 & 0.79 & 94.9\% & 1.22 \\
    & $\sqrt{n}$ & 0.80 & 94.1\% & 1.47 & 0.82 & 94.9\% & 0.71 & 0.81 & 95.0\% & 0.74 \\ 
    & $n^{2 / 3}$ & 0.80 & 94.4\% & 1.32 & 0.83 & 94.9\% & 0.35 & 0.84 & 95.1\% & 0.39 \\
    
    \multirow{3}{*}{ 1000} 
    & $2 \log n$ & 0.49 & 90.9\% & 1.72 & 0.42 & 93.8\% & 1.00 & 0.42 & 93.9\% & 1.01 \\
    & $\sqrt{n}$  & 0.45 & 92.7\% & 1.42 & 0.41 & 94.9\% & 0.45 & 0.41 & 94.8\% & 0.47 \\
    & $n^{2 / 3}$ & 0.42 & 93.7\% & 1.25 & 0.41 & 95.1\% & 0.20 & 0.42 & 94.9\% & 0.22 \\
    
    \multirow{3}{*}{ 2000 } 
    & $2 \log n$ & 0.32 & 86.0\% & 1.44 & 0.24 & 91.9\% & 0.87 & 0.24 & 91.9\% & 0.88 \\
    & $\sqrt{n}$  & 0.29 & 88.6\% & 1.35 & 0.20 & 94.9\% & 0.30 & 0.20 & 94.8\% & 0.31 \\
    & $n^{2 / 3}$ & 0.26 & 90.6\% & 1.11 & 0.20 & 94.9\% & 0.11 & 0.20 & 95.0\% & 0.12 \\
    
    \multirow{3}{*}{ 4000 } 
    & $2 \log n$ & 0.24 & 74.5\% & 1.57 & 0.13 & 91.2\% & 0.77 & 0.13 & 91.2\% & 0.77 \\ 
    & $\sqrt{n}$ & 0.19 & 81.6\% & 1.27 & 0.10 & 94.9\% & 0.20 & 0.10 & 94.8\% & 0.21 \\
    & $n^{2 / 3}$  & 0.19 & 82.7\% & 1.17 & 0.10 & 95.0\% & 0.07 & 0.10 & 95.3\% & 0.07 \\
    \hline
    \hline
    \end{tabular}
    \label{tableLogitkbeta_2}
\end{table}

\begin{table}[H]
    \caption{Median MSE ($\times 10^{2}$), coverage probability, and MSPE ($\times 10^{2}$) for subspace logistic regression under SBM with random network perturbations and misspecified $r$.}
    \centering
    \begin{tabular}{cccccccc}
    \hline
    \hline
    \multirow{2}{*}{ n } & \multirow{2}{*}{ avg.\ degree } & \multicolumn{3}{c}{ $r=0$ (True Model) } & \multicolumn{3}{c}{$r=1$}  \\
    & & MSE  & Coverage & MSPE & MSE  & Coverage & MSPE  \\
    \hline
    \multirow{3}{*}{ 500 } 
    & $2 \log n$ & 0.78 & 94.8\% & 1.20 & 11.34 & 43.3\% & 1.28 \\
    & $\sqrt{n}$ & 0.82 & 94.9\% & 0.71 & 13.09 & 41.2\% & 0.75 \\ 
    & $n^{2 / 3}$ & 0.83 & 94.9\% & 0.35 & 12.75 & 44.3\% & 0.43 \\
    
    \multirow{3}{*}{ 1000} 
    & $2 \log n$ & 0.42 & 93.8\% & 1.00 & 23.09 & 0.0\% & 1.10 \\
    & $\sqrt{n}$  & 0.41 & 94.9\% & 0.45 & 24.08 & 0.0\% & 0.49 \\
    & $n^{2 / 3}$ & 0.41 & 95.1\% & 0.20 & 24.39 & 0.0\% & 0.19 \\
    
    \multirow{3}{*}{ 2000 } 
    & $2 \log n$ & 0.24 & 91.9\% & 0.87 & 21.29 & 0.0\% & 1.12 \\
    & $\sqrt{n}$  & 0.20 & 94.9\% & 0.30 & 24.16 & 0.0\% & 0.61 \\
    & $n^{2 / 3}$ & 0.20 & 94.9\% & 0.11 & 24.88 & 0.0\% & 0.46 \\
    
    \multirow{3}{*}{ 4000 } 
    & $2 \log n$ & 0.13 & 91.2\% & 0.77 & 19.57 & 0.0\% & 1.20 \\ 
    & $\sqrt{n}$ & 0.10 & 94.9\% & 0.20 & 24.59 & 0.0\% & 0.81 \\
    & $n^{2 / 3}$  & 0.10 & 95.0\% & 0.07 & 24.73 & 0.0\% & 0.75 \\
    \hline
    \hline
    \end{tabular}
    \label{tableLogitrbeta_2}
\end{table}

\begin{table}[H]
    \caption{Median MSE ($\times 10^{2}$), coverage probability, and MSPE for subspace Poisson regression under SBM with random network perturbations and misspecified $K$.}
    \centering
    \begin{tabular}{ccccccccccc}
    \hline
    \hline
    \multirow{2}{*}{ n } & \multirow{2}{*}{ avg.\ degree } & \multicolumn{3}{c}{ $K=2$ } & \multicolumn{3}{c}{ $K=3$ (True Model)} & \multicolumn{3}{c}{ $K=4$ }  \\
    & & MSE  & Coverage & MSPE & MSE  & Coverage & MSPE & MSE  & Coverage & MSPE  \\
    \hline 
    \multirow{3}{*}{ 500 } 
    & $2 \log n$ & 0.39 & 63.7\% & 1.82 & 0.30 & 73.8\% & 1.55 & 0.32 & 72.1\% & 1.52 \\
    & $\sqrt{n}$ & 0.30 & 72.6\% & 1.33 & 0.18 & 86.0\% & 0.93 & 0.19 & 85.9\% & 0.93 \\ 
    & $n^{2 / 3}$ & 0.21 & 84.4\% & 0.68 & 0.12 & 93.5\% & 0.37 & 0.12 & 93.4\% & 0.37 \\
    
    \multirow{3}{*}{ 1000} 
    & $2 \log n$ & 0.10 & 84.6\% & 1.19 & 0.10 & 85.2\% & 1.05 & 0.10 & 85.2\% & 1.05 \\
    & $\sqrt{n}$  & 0.08 & 89.3\% & 0.75 & 0.07 & 91.8\% & 0.52 & 0.07 & 91.9\% & 0.53 \\
    & $n^{2 / 3}$ & 0.07 & 91.6\% & 0.48 & 0.06 & 94.2\% & 0.18 & 0.06 & 94.4\% & 0.18 \\
    
    \multirow{3}{*}{ 2000 } 
    & $2 \log n$ & 0.10 & 65.7\% & 0.96 & 0.07 & 75.5\% & 0.75 & 0.07 & 76.0\% & 0.75 \\
    & $\sqrt{n}$  & 0.05 & 83.6\% & 0.56 & 0.03 & 92.1\% & 0.30 & 0.03 & 91.9\% & 0.30 \\
    & $n^{2 / 3}$ & 0.04 & 88.8\% & 0.40 & 0.03 & 94.6\% & 0.09 & 0.03 & 94.5\% & 0.09 \\
    
    \multirow{3}{*}{ 4000 } 
    & $2 \log n$ & 0.07 & 43.7\% & 1.34 & 0.06 & 53.4\% & 1.14 & 0.06 & 54.0\% & 1.14 \\ 
    & $\sqrt{n}$ & 0.03 & 76.1\% & 0.70 & 0.02 & 91.7\% & 0.34 & 0.03 & 91.3\% & 0.34 \\
    & $n^{2 / 3}$  & 0.02 & 84.3\% & 0.45 & 0.01 & 94.6\% & 0.09 & 0.02 & 94.7\% & 0.09 \\
    
    \hline
    \hline
    \end{tabular}
    \label{tablePoikbeta_2}
\end{table}

\begin{table}[H]
    \caption{Median MSE ($\times 10^{2}$), coverage probability, and MSPE for subspace Poisson regression under SBM with random network perturbations and misspecified $r$.}
    \centering
    \begin{tabular}{cccccccc}
    \hline
    \hline
    \multirow{2}{*}{ n } & \multirow{2}{*}{ avg.\ degree } & \multicolumn{3}{c}{ $r=0$ (True Model) } & \multicolumn{3}{c}{$r=1$}  \\
    & & MSE  & Coverage & MSPE & MSE  & Coverage & MSPE  \\
    \hline
    \multirow{3}{*}{ 500 } 
    & $2 \log n$ & 0.30 & 73.8\% & 1.55 & 9.50 & 17.7\% & 1.78 \\
    & $\sqrt{n}$ & 0.18 & 86.0\% & 0.93 & 11.75 & 25.6\% & 1.39 \\ 
    & $n^{2 / 3}$ & 0.12 & 93.5\% & 0.37 & 11.65 & 27.8\% & 1.09 \\
    
    \multirow{3}{*}{ 1000} 
    & $2 \log n$ & 0.10 & 85.2\% & 1.05 & 5.78 & 0.0\% & 1.26 \\
    & $\sqrt{n}$  & 0.07 & 91.8\% & 0.52 & 4.78 & 0.0\% & 1.05 \\
    & $n^{2 / 3}$ & 0.06 & 94.2\% & 0.18 & 3.88 & 4.75\% & 1.25 \\
    
    \multirow{3}{*}{ 2000 } 
    & $2 \log n$ & 0.07 & 75.5\% & 0.75 & 18.32 & 0.0\% & 0.79 \\
    & $\sqrt{n}$  & 0.03 & 92.1\% & 0.30 & 22.14 & 0.0\% & 0.33 \\
    & $n^{2 / 3}$ & 0.03 & 94.6\% & 0.09 & 23.84 & 0.0\% & 0.14 \\
    
    \multirow{3}{*}{ 4000 } 
    & $2 \log n$ & 0.06 & 53.4\% & 1.14 & 20.07 & 0.0\% & 1.30 \\ 
    & $\sqrt{n}$ & 0.02 & 91.7\% & 0.34 & 24.48 & 0.0\% & 0.72 \\
    & $n^{2 / 3}$  & 0.01 & 94.6\% & 0.09 & 24.90 & 0.0\% & 0.61 \\
    \hline
    \hline
    \end{tabular}
    \label{tablePoirbeta_2}
\end{table}

Across all simulation settings, we observe that underestimating the embedding dimension $\hat{K}$ results in moderate degradation of inference accuracy. Specifically, coverage probabilities decline by approximately 10\%--20\%, and the mean squared error (MSE) increases noticeably under both logistic and Poisson regression models. In contrast, overestimating $\hat{K}$ has minimal impact on performance, as the model remains correctly specified with respect to the signal subspace.

However, overestimating the subspace dimension $\hat{r}$ leads to substantially more severe consequences: coverage probabilities frequently drop to 0\%, and estimation errors increase sharply. This stark contrast highlights the sensitivity of the method to the specification of $\hat{r}$.

A potential explanation for these patterns is that both underestimating $\hat{K}$ and overestimating $\hat{r}$ cause the estimated subspace $\hat{\mathcal{N}}$ to omit important signal-bearing directions. When $\hat{K}$ is underestimated, relevant eigenvectors associated with the network effect are excluded from the embedding. Conversely, overestimating $\hat{r}$ disrupts the alignment between the covariates and the network subspace, leading to partial loss of the signal during the projection step.

\newpage
\section{Additional simulation results under network embedding perturbations}\label{sec:embed n=500}
Results under perturbations from different embedding algorithms for sample size $n=500$ are provided in Tables \ref{embed500logistic} and \ref{embed500Poisson}, showing patterns consistent with those discussed in the main text.
\begin{table}[H]
    \centering
        \caption{Median MSE ($\times 10^{2}$), coverage probability and MSPE ($\times 10^{2}$) for subspace logistic regression with different types of network of size $500$ under network embedding perturbations.}
\begin{tabular}{ccccccccccc}
\hline \hline
\multirow{2}{*}{Method} &\multirow{2}{*}{Network}&\multirow{2}{*}{avg.\ degree}& 
\multirow{2}{*}{MSE}&\multirow{2}{*}{Coverage} &\multicolumn{2}{c}{MSPE}\\
&&&&&Our Model &Logistic Reg\\
\hline \hline
\multirow{9}{*}{DeepWalk}
&& $2 \log n$ & 1.25 &94.5\% & 0.21 & 2.32 \\
&SBM & $\sqrt{n}$ & 1.19 & 94.8\%  & 0.20 & 2.23 \\
&& $n^{2 / 3}$ & 1.21 & 95.2\% & 0.19 & 1.34 \\
\cline{2-7}
&& $2 \log n$ & 1.91 & 94.3\%  & 0.22 & 1.10 \\
&DCBM & $\sqrt{n}$ & 1.26 & 94.8\%  & 0.21 & 0.56\\
&& $n^{2 / 3}$ & 1.26 & 94.9\%  & 0.20 & 2.30 \\
\cline{2-7}
&& $2 \log n$ & 1.24 & 95.2\%  & 0.17 & 1.73 \\
&Diag & $\sqrt{n}$ & 1.15 & 95.0\%  & 0.17 & 1.44 \\
&& $n^{2 / 3}$ & 1.19 & 95.0\%  & 0.17 & 1.33 \\
\cline{1-7}
\multirow{9}{*}{Node2Vec}&& $2 \log n$ & 1.25 & 94.5\%  & 0.21 & 0.64 \\
&SBM & $\sqrt{n}$ & 1.20 & 94.8\%  & 0.20 & 2.11 \\
&& $n^{2 / 3}$& 1.20 & 95.1\% & 0.21 & 2.29 \\
\cline{2-7}
&& $2 \log n$ & 1.47 & 94.7\%  & 0.21 & 1.45 \\
&DCBM & $\sqrt{n}$ & 1.38 & 94.7\% & 0.21 & 1.30\\
&& $n^{2 / 3}$ & 1.38 & 94.6\% & 0.21 & 2.42 \\
\cline{2-7}
&& $2 \log n$ & 1.21 & 95.0\%  & 0.17 & 0.57\\
&Diag & $\sqrt{n}$ & 1.25 & 95.0\%  & 0.17 & 1.51 \\
&& $n^{2 / 3}$ & 1.18 & 94.9\%  & 0.17 & 1.52 \\
\cline{1-7}
\multirow{9}{*}{Diff2Vec}&& $2 \log n$ & 1.26 & 94.3\% & 0.21 & 1.24 \\
&SBM & $\sqrt{n}$ & 1.30 & 94.6\% & 0.20 & 1.25 \\
&& $n^{2 / 3}$ & 1.21 & 95.1\% & 0.19 & 1.24 \\
\cline{2-7}
&& $2 \log n$ & 1.97 & 86.1\% & 0.26 & 1.05 \\
&DCBM & $\sqrt{n}$ & 1.39 & 94.5\%  & 0.21 & 1.33\\
&& $n^{2 / 3}$ & 1.31 & 94.8\% & 0.19 & 1.36 \\
\cline{2-7}
&& $2 \log n$ & 1.41 & 94.3\%  & 0.24 & 1.47 \\
&Diag & $\sqrt{n}$ & 1.29 & 95.0\% & 0.21 & 1.04 \\
&& $n^{2 / 3}$ & 1.24 & 95.2\%  & 0.20 & 1.13 \\
\hline
\hline
\end{tabular}
\label{embed500logistic}
\end{table}

\begin{table}[H]
    \centering
        \caption{Median MSE ($\times 10^{2}$), coverage probability, and MSPE for subspace Poisson regression with different types of network of size $500$ under network embedding perturbations.}
\begin{tabular}{ccccccccccc}
\hline \hline
\multirow{2}{*}{Method} &\multirow{2}{*}{Network}&\multirow{2}{*}{avg.\ degree}&\multirow{2}{*}{MSE}& 
\multirow{2}{*}{Coverage} &\multicolumn{2}{c}{MSPE}\\
&&&&&Our Method &Poisson Reg\\
\hline \hline
\multirow{9}{*}{DeepWalk}
&& $2 \log n$ & 0.33 & 93.5\%  & 3.07 & 48.2 \\
&SBM & $\sqrt{n}$   & 0.25 & 93.9\% & 2.51 & 66.9 \\
&& $n^{2 / 3}$  & 0.24 & 94.5\% & 2.34 & 21.7 \\
\cline{2-7}
&& $2 \log n$  & 0.58& 92.6\% & 2.98 & 44.3 \\
&DCBM & $\sqrt{n}$ & 0.21 & 93.6\%  & 4.81 & 35.1 \\
&& $n^{2 / 3}$   & 0.29& 94.2\% & 2.22 & 46.1 \\
\cline{2-7}
 && $2 \log n$   & 0.29 & 94.9\% & 1.30 & 33.2 \\
&Diag & $\sqrt{n}$   & 0.20& 95.0\% & 1.57 & 22.9 \\
&& $n^{2 / 3}$  & 0.18 & 95.0\% & 1.79 & 17.5 \\
\cline{1-7}
\multirow{9}{*}{Node2Vec}&& $2 \log n$   & 0.32 & 92.9\% & 2.87 & 13.8 \\
&SBM & $\sqrt{n}$   & 0.19 & 94.3\% & 3.10 & 28.6 \\
&& $n^{2 / 3}$   & 0.25 & 94.4\% & 2.30 & 42.9 \\
\cline{2-7}
&& $2 \log n$  & 0.34 & 93.0\% & 2.15 & 5.44 \\
&DCBM & $\sqrt{n}$ & 0.20 & 92.5\%  &3.97 &  62.7 \\
&& $n^{2 / 3}$   & 0.27 & 94.0\% & 3.42 & 59.0 \\
\cline{2-7}
&& $2 \log n$   & 0.27 & 94.9\% & 1.15 & 9.42 \\
&Diag & $\sqrt{n}$  & 0.27 & 94.8\% & 1.62 & 38.6 \\
&& $n^{2 / 3}$  & 0.27 & 95.0\% & 1.47 & 32.6 \\
\cline{1-7}
\multirow{9}{*}{Diff2Vec}
&& $2 \log n$ & 0.35 & 92.5\% & 2.37 & 22.5\\
&SBM & $\sqrt{n}$  & 0.31 & 94.1\% & 1.89 & 21.9 \\
&& $n^{2 / 3}$  & 0.21 & 94.0\% & 2.06 & 21.4 \\
\cline{2-7}
&& $2 \log n$ & 1.11& 54.2\% & 3.42 & 19.0\\
&DCBM & $\sqrt{n}$  & 0.35 & 93.4\% & 2.80 & 24.2 \\
&& $n^{2 / 3}$  & 0.22 & 94.2\% & 2.64 & 32.4 \\
\cline{2-7}
&& $2 \log n$ & 0.23 & 91.8\% & 3.24 & 25.7 \\
&Diag & $\sqrt{n}$ & 0.31 & 93.6\% & 1.89 & 8.38 \\
&& $n^{2 / 3}$  & 0.26 & 94.5\% & 1.48 & 13.7 \\
\hline
\hline
\end{tabular}
\label{embed500Poisson}
\end{table}

\newpage
\section{Additional simulation results to explore the influence of embedding dimensions}\label{sec:embed}
To explore the impact of increasing embedding dimension, we design a new simulation setting: 
Compared to the simulation in Section \ref{sec:Simulation Studies} where we take embedding dimension $K_{embed}$ to be $3$, the only difference is that $K_{embed}$ takes values from $3$ to $10$. 
We focus on the influence of embedding dimensions on the stochastic block model (SBM).
The value of the top 10 eigenvalues of $P=\mathbb{E}[\mathcal{F}\mathcal{F}^\top]$ with overlaid boxplots illustrating the distribution of the top $K_{embed}$ eigenvalues computed from $B$ similarity matrices $\hat{P} = \mathcal{F}\mathcal{F}^\top$ are summarized in Figure \ref{fig:SBM_deepwalk_plot}, \ref{fig:SBM_node2vec_plot},
\ref{fig:SBM_diff2vec_plot}.
Tables \ref{deepwalk_SBM_embedlogisticmulti} to \ref{diff2vec_SBM_embedPoimulti} report performance metrics under perturbations for various embedding algorithms, evaluated across embedding dimensions from $3$ to $10$ and different average degrees, with sample size $n=2000$.

\begin{figure}[H]
    \centering
        \begin{subfigure}[b]{0.325\linewidth}
        \centering
        {\includegraphics[width=\linewidth]{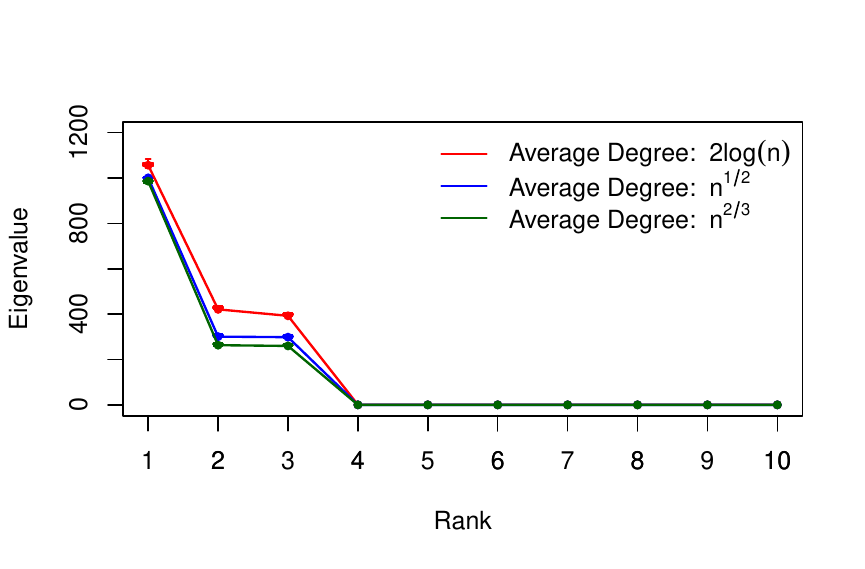}}
        \subcaption{$K_{embed}=3$}
    \end{subfigure}
    \hfill
    \begin{subfigure}[b]{0.325\linewidth}
        \centering
        {\includegraphics[width=\linewidth]{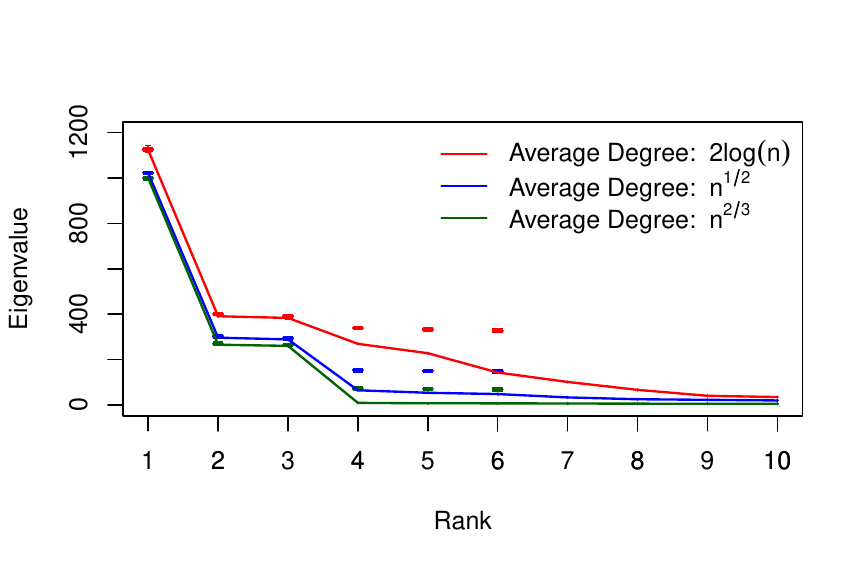}}
        \subcaption{$K_{embed}=6$}
    \end{subfigure}
    \hfill
    \begin{subfigure}[b]{0.325\linewidth}
        \centering
        {\includegraphics[width=\linewidth]{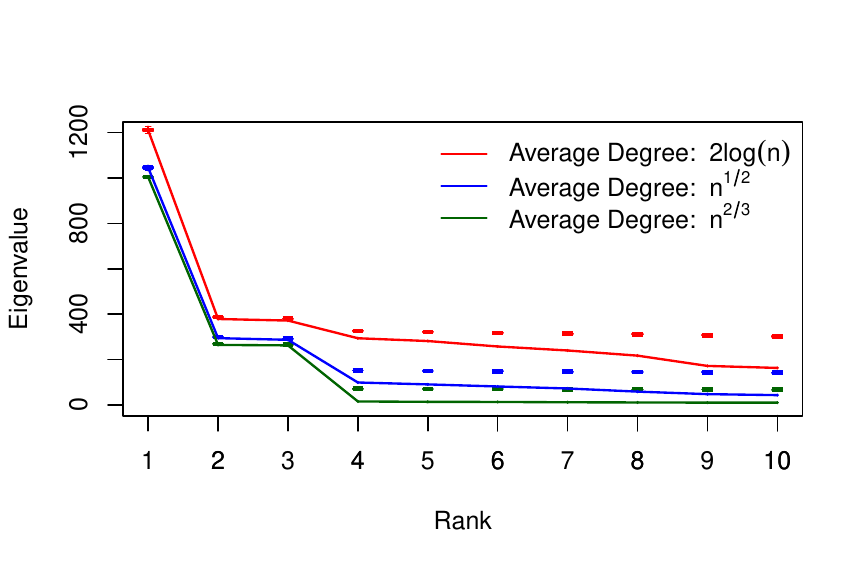}}
        \subcaption{$K_{embed}=10$}
    \end{subfigure}
    \hfill
    \caption{Top 10 eigenvalues of the stochastic block model (SBM) relational matrix under DeepWalk, with overlaid boxplots representing the distribution of the top $K_{\text{embed}}$ eigenvalues computed from $B$ similarity matrices.}
    \label{fig:SBM_deepwalk_plot}
\end{figure}

\begin{figure}[H]
    \centering
        \begin{subfigure}[b]{0.325\linewidth}
        \centering
        {\includegraphics[width=\linewidth]{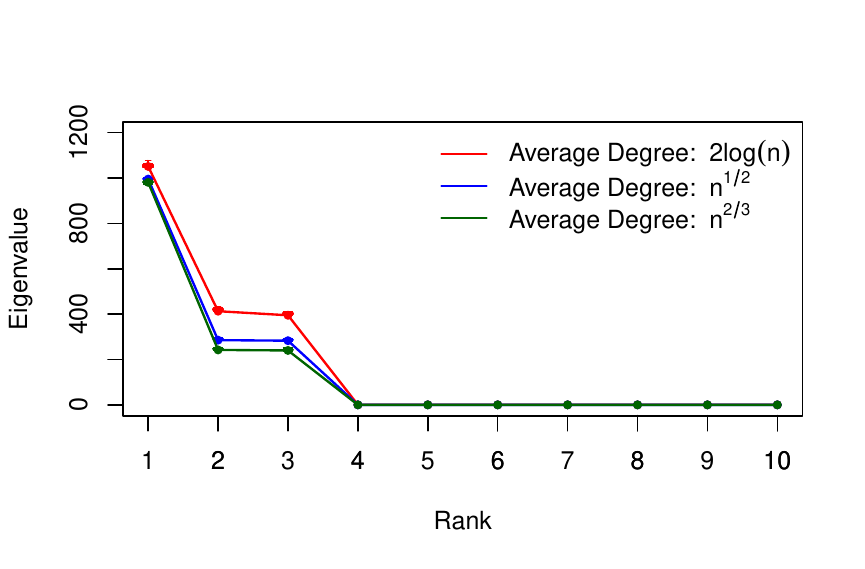}}
        \subcaption{$K_{embed}=3$}
    \end{subfigure}
    \hfill
    \begin{subfigure}[b]{0.325\linewidth}
        \centering
        {\includegraphics[width=\linewidth]{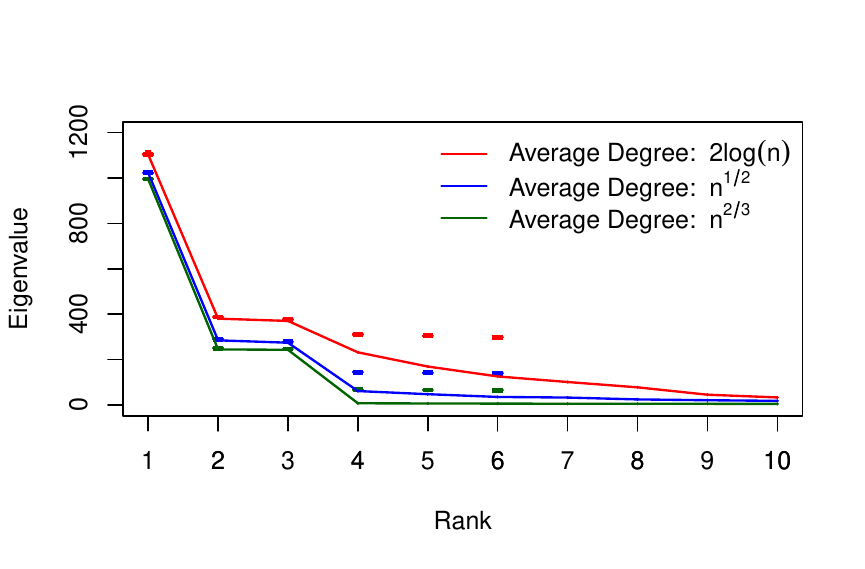}}
        \subcaption{$K_{embed}=6$}
    \end{subfigure}
    \hfill
    \begin{subfigure}[b]{0.325\linewidth}
        \centering
        {\includegraphics[width=\linewidth]{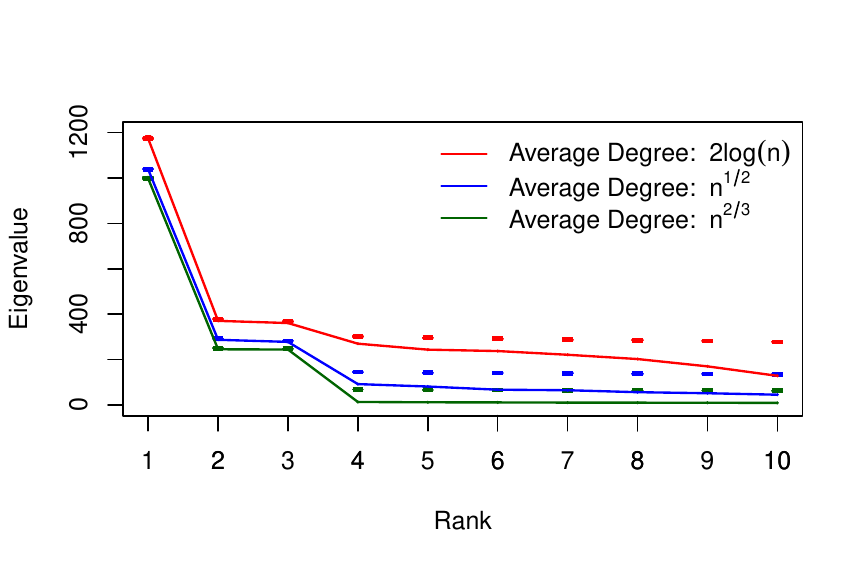}}
        \subcaption{$K_{embed}=10$}
    \end{subfigure}
    \hfill
    \caption{Top 10 eigenvalues of the stochastic block model (SBM) relational matrix under Node2Vec, with overlaid boxplots representing the distribution of the top $K_{\text{embed}}$ eigenvalues computed from $B$ similarity matrices.}
    \label{fig:SBM_node2vec_plot}
\end{figure}

\begin{figure}[H]
    \centering
        \begin{subfigure}[b]{0.325\linewidth}
        \centering
        {\includegraphics[width=\linewidth]{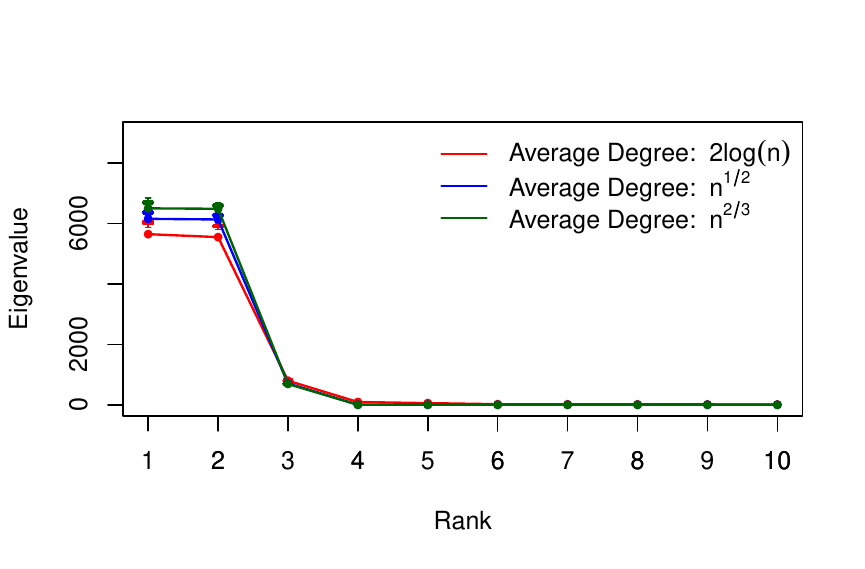}}
        \subcaption{$K_{embed}=3$}
    \end{subfigure}
    \hfill
    \begin{subfigure}[b]{0.325\linewidth}
        \centering
        {\includegraphics[width=\linewidth]{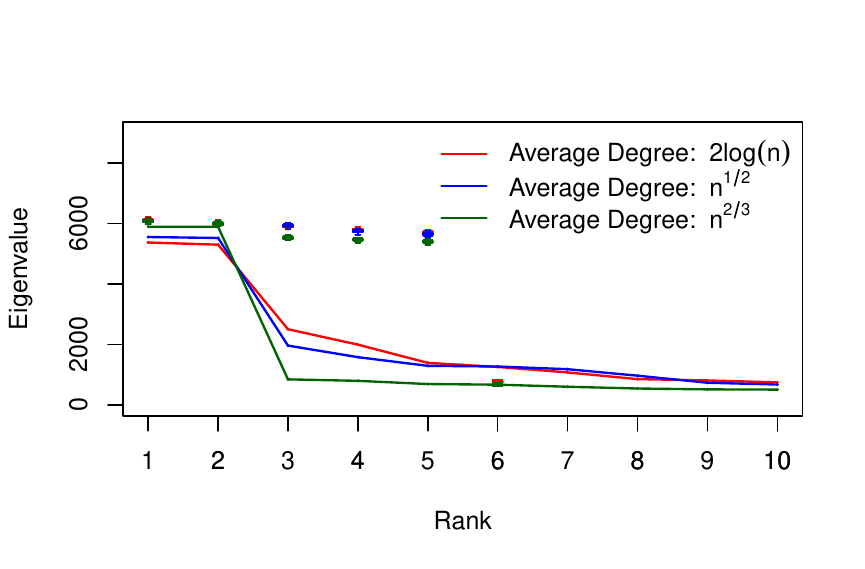}}
        \subcaption{$K_{embed}=6$}
    \end{subfigure}
    \hfill
    \begin{subfigure}[b]{0.325\linewidth}
        \centering
        {\includegraphics[width=\linewidth]{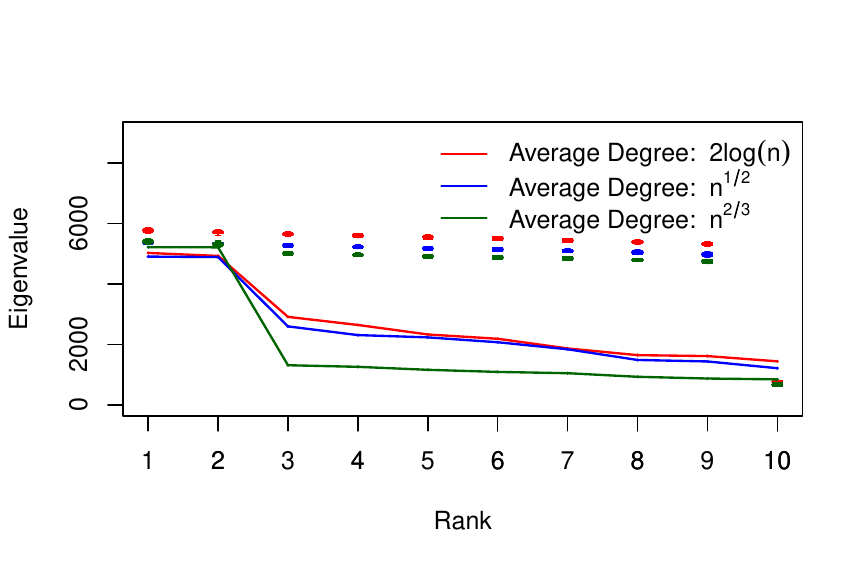}}
        \subcaption{$K_{embed}=10$}
    \end{subfigure}
    \hfill
    \caption{Top 10 eigenvalues of the stochastic block model (SBM) relational matrix under Diff2Vec, with overlaid boxplots representing the distribution of the top $K_{\text{embed}}$ eigenvalues computed from $B$ similarity matrices.}
    \label{fig:SBM_diff2vec_plot}
\end{figure}

\begin{table}[H]
    \centering
        \caption{Median MSE ($\times 10^{2}$), coverage probability, MSPE ($\times 10^{2}$) for subspace logistic regression with different embedding dimensions of a SBM network of size $2000$ under DeepWalk network embedding perturbations.}
\begin{tabular}{ccccc}
\hline \hline
Embedding Dimension  & avg.\ degree & MSE & Coverage & MSPE\\
\hline \hline
\multirow{3}{*}{$K_{\text{embed}}=3$}
& $2 \log n$ & 0.30 & 94.9\% & 0.08  \\
& $\sqrt{n}$ & 0.29 & 94.8\% & 0.07  \\
& $n^{2 / 3}$ & 0.29 & 94.7\% & 0.08 \\
\hline
\multirow{3}{*}{$K_{\text{embed}}=4$}
& $2 \log n$ & 0.31 & 93.8\% & 0.09  \\
& $\sqrt{n}$ & 0.29 & 95.0\% & 0.08  \\
& $n^{2 / 3}$ & 0.28 & 94.9\% & 0.08 \\
\hline
\multirow{3}{*}{$K_{\text{embed}}=6$}
 & $2 \log n$  & 0.32 & 93.6\% & 0.09  \\
 & $\sqrt{n}$  & 0.28 & 95.0\% & 0.08  \\
& $n^{2 / 3}$ & 0.28 & 95.1\% & 0.08  \\
\hline
\multirow{3}{*}{$K_{\text{embed}}=8$}
& $2 \log n$  & 0.30 & 93.7\% & 0.10   \\
& $\sqrt{n}$  & 0.28 & 94.9\% & 0.08  \\
& $n^{2 / 3}$  & 0.28 & 95.1\% & 0.08   \\
\hline
\multirow{3}{*}{$K_{\text{embed}}=10$}
& $2 \log n$ & 0.30 & 94.3\% & 0.10  \\
& $\sqrt{n}$ & 0.28 & 94.9\% & 0.08  \\
& $n^{2 / 3}$  & 0.28 & 95.0\% & 0.08  \\
\hline
\hline
\end{tabular}
\label{deepwalk_SBM_embedlogisticmulti}
\end{table}

\begin{table}[H]
    \centering
        \caption{Median MSE ($\times 10^{2}$), coverage probability, and MSPE ($\times 10^2$) for subspace Poisson regression with different embedding dimensions of a SBM network of size $2000$ under DeepWalk network embedding perturbations.}
\begin{tabular}{ccccc}
\hline \hline
Embedding Dimension & avg.\ degree & MSE & Coverage & MSPE\\
\hline \hline
\multirow{3}{*}{$K_{\text{embed}}=3$}
& $2 \log n$ & 0.06 & 92.9\% & 1.99  \\
& $\sqrt{n}$ & 0.06 & 94.0\% & 1.58  \\
& $n^{2 / 3}$ & 0.06 & 93.8\% & 1.27  \\
\hline
\multirow{3}{*}{$K_{\text{embed}}=4$}
& $2 \log n$ & 0.08 & 90.2\% & 2.02  \\
& $\sqrt{n}$ & 0.06 & 94.5\% & 1.35  \\
& $n^{2 / 3}$ & 0.06 & 94.8\% & 1.22  \\
\hline
\multirow{3}{*}{$K_{\text{embed}}=6$}
& $2 \log n$  & 0.11 & 86.0\% & 2.19 \\
& $\sqrt{n}$  & 0.06 & 94.3\% & 1.38  \\
& $n^{2 / 3}$ & 0.06 & 94.7\% & 1.23  \\
\hline
\multirow{3}{*}{$K_{\text{embed}}=8$}
& $2 \log n$  & 0.08 & 90.4\% & 2.33  \\
& $\sqrt{n}$  & 0.06 & 94.4\% & 1.36  \\
& $n^{2 / 3}$ & 0.06 & 94.6\% & 1.19  \\
\hline
\multirow{3}{*}{$K_{\text{embed}}=10$}
& $2 \log n$  & 0.08 & 92.3\% & 2.43  \\
& $\sqrt{n}$  & 0.06 & 94.5\% & 1.37  \\
& $n^{2 / 3}$  & 0.05 & 94.8\% & 1.28  \\
\hline
\hline
\end{tabular}
\label{deepwalk_SBM_embedPoimulti}
\end{table}

\begin{table}[H]
    \centering
        \caption{Median MSE ($\times 10^{2}$), coverage probability, MSPE ($\times 10^{2}$) for subspace logistic regression with different embedding dimensions of a SBM network of size $2000$ under Node2Vec network embedding perturbations.}
\begin{tabular}{ccccc}
\hline \hline
Embedding Dimension  & avg.\ degree & MSE & Coverage & MSPE\\
\hline \hline
\multirow{3}{*}{$K_{\text{embed}}=3$}
& $2 \log n$ & 0.30 & 94.3\% & 0.09  \\
& $\sqrt{n}$ & 0.28 & 94.8\% & 0.08  \\
& $n^{2 / 3}$ & 0.29 & 94.9\% & 0.08 \\
\hline
\multirow{3}{*}{$K_{\text{embed}}=4$}
& $2 \log n$ & 0.31 & 94.4\% & 0.08  \\
& $\sqrt{n}$ & 0.28 & 94.9\% & 0.08  \\
& $n^{2 / 3}$ & 0.28 & 95.0\% & 0.08 \\
\hline
\multirow{3}{*}{$K_{\text{embed}}=6$}
 & $2 \log n$  & 0.31 & 93.8\% & 0.09  \\
 & $\sqrt{n}$  & 0.28 & 95.0\% & 0.08  \\
& $n^{2 / 3}$ & 0.29 & 94.9\% & 0.08  \\
\hline
\multirow{3}{*}{$K_{\text{embed}}=8$}
& $2 \log n$  & 0.29 & 94.5\% & 0.09   \\
& $\sqrt{n}$  & 0.29 & 95.1\% & 0.08  \\
& $n^{2 / 3}$  & 0.28 & 94.9\% & 0.08   \\
\hline
\multirow{3}{*}{$K_{\text{embed}}=10$}
& $2 \log n$ & 0.30 & 94.3\% & 0.09  \\
& $\sqrt{n}$ & 0.28 & 95.0\% & 0.08  \\
& $n^{2 / 3}$  & 0.28 & 95.2\% & 0.08  \\
\hline
\hline
\end{tabular}
\label{node2vec_SBM_embedlogisticmulti}
\end{table}

\begin{table}[H]
    \centering
        \caption{Median MSE ($\times 10^{2}$), coverage probability, and MSPE ($\times 10^2$) for subspace Poisson regression with different embedding dimensions of a SBM network of size $2000$ under Node2Vec network embedding perturbations.}
\begin{tabular}{ccccc}
\hline \hline
Embedding Dimension & avg.\ degree & MSE & Coverage & MSPE\\
\hline \hline
\multirow{3}{*}{$K_{\text{embed}}=3$}
& $2 \log n$ & 0.08 & 93.2\% & 1.93  \\
& $\sqrt{n}$ & 0.06 & 93.7\% & 1.54  \\
& $n^{2 / 3}$ & 0.07 & 94.0\% & 1.23  \\
\hline
\multirow{3}{*}{$K_{\text{embed}}=4$}
& $2 \log n$ & 0.07 & 93.1\% & 1.72  \\
& $\sqrt{n}$ & 0.06 & 94.5\% & 1.37  \\
& $n^{2 / 3}$ & 0.06 & 94.6\% & 1.54  \\
\hline
\multirow{3}{*}{$K_{\text{embed}}=6$}
& $2 \log n$  & 0.09 & 91.0\% & 1.68 \\
& $\sqrt{n}$  & 0.06 & 94.6\% & 1.45  \\
& $n^{2 / 3}$ & 0.06 & 94.7\% & 1.39  \\
\hline
\multirow{3}{*}{$K_{\text{embed}}=8$}
& $2 \log n$  & 0.07 & 93.1\% & 1.71  \\
& $\sqrt{n}$  & 0.06 & 94.7\% & 1.35  \\
& $n^{2 / 3}$ & 0.06 & 94.8\% & 1.32  \\
\hline
\multirow{3}{*}{$K_{\text{embed}}=10$}
& $2 \log n$  & 0.08 & 91.9\% & 1.96  \\
& $\sqrt{n}$  & 0.06 & 94.6\% & 1.34  \\
& $n^{2 / 3}$  & 0.06 & 94.6\% & 1.42  \\
\hline
\hline
\end{tabular}
\label{node2vec_SBM_embedPoimulti}
\end{table}

\begin{table}[H]
    \centering
        \caption{Median MSE ($\times 10^{2}$), coverage probability, MSPE ($\times 10^{2}$) for subspace logistic regression with different embedding dimensions of a SBM network of size $2000$ under Diff2Vec network embedding perturbations.}
\begin{tabular}{ccccc}
\hline \hline
Embedding Dimension  & avg.\ degree & MSE & Coverage & MSPE\\
\hline \hline
\multirow{3}{*}{$K_{\text{embed}}=3$}
& $2 \log n$ & 0.32 & 93.5\% & 0.10  \\
& $\sqrt{n}$ & 0.30 & 94.8\% & 0.07  \\
& $n^{2 / 3}$ & 0.30 & 94.6\% & 0.07 \\
\hline
\multirow{3}{*}{$K_{\text{embed}}=4$}
& $2 \log n$ & 0.66 & 80.5\% & 1.02  \\
& $\sqrt{n}$ & 0.41 & 90.4\% & 1.02  \\
& $n^{2 / 3}$ & 0.33 & 92.8\% & 1.14 \\
\hline
\multirow{3}{*}{$K_{\text{embed}}=6$}
 & $2 \log n$  & 13.4 & 0\% & 1.92  \\
 & $\sqrt{n}$  & 26.1 & 0\% & 1.90  \\
& $n^{2 / 3}$ & 0.33 & 92.8\% & 1.10  \\
\hline
\multirow{3}{*}{$K_{\text{embed}}=8$}
& $2 \log n$  & 4.82 & 0\% & 2.03   \\
& $\sqrt{n}$  & 13.7 & 0\% & 1.93  \\
& $n^{2 / 3}$  & 0.33 & 93.2\% & 1.11   \\
\hline
\multirow{3}{*}{$K_{\text{embed}}=10$}
& $2 \log n$ & 1.35 & 39.6\% & 1.98  \\
& $\sqrt{n}$ & 20.6 & 0\% & 2.00  \\
& $n^{2 / 3}$  & 0.34 & 92.5\% & 1.10  \\
\hline
\hline
\end{tabular}
\label{diff2vec_SBM_embedlogisticmulti}
\end{table}

\begin{table}[H]
    \centering
        \caption{Median MSE ($\times 10^{2}$), coverage probability, and MSPE ($\times 10^2$) for subspace Poisson regression with different embedding dimensions of a SBM network of size $2000$ under Diff2Vec network embedding perturbations.}
\begin{tabular}{ccccc}
\hline \hline
Embedding Dimension & avg.\ degree & MSE & Coverage & MSPE\\
\hline \hline
\multirow{3}{*}{$K_{\text{embed}}=3$}
& $2 \log n$ & 0.09 & 87.7\% & 2.58  \\
& $\sqrt{n}$ & 0.07 & 94.0\% & 1.19  \\
& $n^{2 / 3}$ & 0.07 & 93.9\% &1.42  \\
\hline
\multirow{3}{*}{$K_{\text{embed}}=4$}
& $2 \log n$ & 0.44 & 7.8\% & 110  \\
& $\sqrt{n}$ & 0.22 & 46.1\% & 79.4  \\
& $n^{2 / 3}$ & 1.88 & 0\% & 112  \\
\hline
\multirow{3}{*}{$K_{\text{embed}}=6$}
& $2 \log n$  & 13.4 & 0\% & 175 \\
& $\sqrt{n}$  & 25.2 & 0\% & 186 \\
& $n^{2 / 3}$ & 0.16 & 49.4\% & 118  \\
\hline
\multirow{3}{*}{$K_{\text{embed}}=8$}
& $2 \log n$  & 5.19 & 0\% & 189  \\
& $\sqrt{n}$  & 13.4 & 0\% & 165  \\
& $n^{2 / 3}$ & 0.14 & 52.0\% & 141  \\
\hline
\multirow{3}{*}{$K_{\text{embed}}=10$}
& $2 \log n$ & 3.05 & 0\% & 178  \\
& $\sqrt{n}$ & 20.2 & 0\% & 172  \\
& $n^{2 / 3}$ & 0.13 & 59.1\% & 113  \\
\hline
\hline
\end{tabular}
\label{diff2vec_SBM_embedPoimulti}
\end{table}

The results for DeepWalk and Node2Vec (Tables~\ref{deepwalk_SBM_embedlogisticmulti}--\ref{node2vec_SBM_embedPoimulti}) show that inference remains reliable even when \(K_{\text{embed}}\) is moderately larger than the intrinsic rank: coverage stays close to nominal once the average degree is \(\gtrsim \sqrt{n}\), and is only slightly conservative in sparser regimes. The eigenvalue diagnostics (Figures~\ref{fig:SBM_deepwalk_plot}–\ref{fig:SBM_node2vec_plot}) reveal a clear eigen-gap at \(K=3\) and tight concentration of the leading eigenvalues across replicates when the network is not too sparse, which aligns with the small projection perturbation condition (Assumption~\ref{cond:A4}).

By contrast, Diff2Vec (Tables~\ref{diff2vec_SBM_embedlogisticmulti}--\ref{diff2vec_SBM_embedPoimulti}) displays high variability in the leading eigenvalues (Figure~\ref{fig:SBM_diff2vec_plot}) and a lack of concentration across replicates, especially as \(K_{\text{embed}}\) increases. In these regimes, the small projection perturbation condition is violated, leading to substantial coverage distortions and unstable prediction error. This behavior reflects instability of the embedding itself rather than a limitation of our inference procedure. Finally, for very sparse networks (\(\text{avg.\ degree}=2\log n\)), even DeepWalk/Node2Vec can show mild undercoverage at large \(K_{\text{embed}}\), consistent with weaker concentration of \(\hat P\).

\end{document}